\documentclass[11pt,a4paper]{article}
\usepackage{amssymb}
\usepackage[hmargin=2.5cm,vmargin=2.5cm]{geometry}
\usepackage{amssymb,latexsym,amsmath,amsfonts,amsthm}
\usepackage{graphicx}
\usepackage{tikz}
\usepackage{pgflibraryshapes}
\usetikzlibrary{arrows,decorations.markings}
\usepackage{epsfig}
\usepackage{comment,verbatim}
\usepackage{hyperref}

\DeclareMathOperator{\diag}{diag} \DeclareMathOperator{\Res}{Res}
\DeclareMathOperator{\Airy}{Airy} \DeclareMathOperator{\erfc}{erfc}
\DeclareMathOperator{\Bessel}{Bessel}

\newcommand{\DKM}{\mathrm{NC}}
\newcommand{\DG}{\mathrm{NC}}
\newcommand{\er}{\mathbb{R}}
\newcommand{\cee}{\mathbb{C}}

\newcommand{\R}{\mathbb{R}}
\newcommand{\C}{\mathbb{C}}
\newcommand{\Z}{\mathbb{Z}}
\newcommand{\lam}{\lambda}
\newcommand{\Lam}{\Lambda}

\newcommand{\til}{\tilde}

\renewcommand{\Re}{\mathrm{Re}\,}
\renewcommand{\Im}{\mathrm{Im}\,}
\newcommand{\Ai}{\mathrm{Ai}}

\newcommand{\ud}{\,\mathrm{d}}
\newcommand{\om}{\omega}
\newcommand{\wtil}{\widetilde}
\newcommand{\what}{\widehat}
\newcommand{\err}{\mathcal{R}}
\def\supp{\mathop{\mathrm{supp}}\nolimits}
\newcommand{\D}{D \left(0, \rho n^{-1/3}\right)}
\newcommand{\Pnot}{P^{(0)}}
\newcommand{\Pinf}{P^{(\infty)}}
\newcommand{\Kcr}{K^{\mathrm cr}}

\newtheorem{theorem}{Theorem}[section]
\newtheorem{lemma}[theorem]{Lemma}
\newtheorem{proposition}[theorem]{Proposition}

\newtheorem{rhp}[theorem]{RH problem}
\newtheorem{bvp}[theorem]{Boundary value problem}

\theoremstyle{definition}
\newtheorem{definition}[theorem]{Definition}

\theoremstyle{remark}

\newtheorem{remark}[theorem]{Remark}

\numberwithin{equation}{section}

\begin{document}

\title{Universality and critical behavior in the chiral two-matrix
model}

\author{Steven Delvaux\footnotemark[1],\quad Dries Geudens\footnotemark[1],\quad Lun Zhang\footnotemark[1]}
\date{\today}

\maketitle
\renewcommand{\thefootnote}{\fnsymbol{footnote}}
\footnotetext[1]{Department of Mathematics, University of Leuven (KU Leuven),
Celestijnenlaan 200B, B-3001 Leuven, Belgium. E-mail: \{steven.delvaux,
dries.geudens, lun.zhang\}\symbol{'100}wis.kuleuven.be.}


\begin{abstract} We study the chiral two-matrix model
with polynomial potential functions $V$ and $W$, which was
introduced by Akemann, Damgaard, Osborn and Splittorff. We show that
the squared singular values of each of the individual matrices in
this model form a determinantal point process with correlation
kernel determined by a matrix-valued Riemann-Hilbert problem. The
size of the Riemann-Hilbert matrix depends on the degree of the
potential function $W$ (or $V$ respectively). In this way we obtain
the chiral analogue of a result of Kuijlaars-McLaughlin for the
non-chiral two-matrix model. The Gaussian case corresponds to $V,W$
being linear.

For the case where $W(y)=y^2/2+\alpha y$ is quadratic, we derive the large
$n$-asymptotics of the Riemann-Hilbert problem by means of the Deift-Zhou steepest descent
method. This proves universality in this case. An important ingredient in the
analysis is a third-order differential equation.

Finally we show that if also $V(x)=x$ is linear, then a multi-critical limit of
the kernel exists which is described by a $4\times 4$ matrix-valued Riemann-Hilbert problem
associated to the Painlev\'e~II equation $q''(x) = xq(x)+2q^3(x)-\nu-1/2$. In
this way we obtain the chiral analogue of a recent result by Duits and the
second author.

\end{abstract}

\setcounter{tocdepth}{3} \tableofcontents

\section{Introduction}
\label{section:introduction}

Chiral random matrix theory provides a useful calculational framework for
solving various physical problems, especially in quantum chromodynamics.
Originally this was developed for the chiral 1-matrix model with a Gaussian
potential \cite{Verb1,Verb3,Verb2}.

Recently a chiral 2-matrix model was introduced by
Akemann-Damgaard-Osborn-Splittorff \cite{ADOS}. The model is a
Hermitian analogue of a non-Hermitian model introduced earlier by
Osborn \cite{Osb}. The motivation for this matrix model comes again
from quantum chromodynamics. It turns out that the mixed correlation
functions for $n\to\infty$ strongly depend on a certain constant in
the model, which has a physical interpretation as the \lq pion decay
constant\rq\ $F_{\pi}$. This suggests a way to compute $F_\pi$ via
numerical lattice simulations. A numerical implementation of this
idea is provided in \cite{DegrandSch} and some recent developments
are, e.g., in \cite{LBHW}. The method for calculating $F_{\pi}$ was
originally developed without random matrices \cite{DHSS,DHSST2}. The
chiral two-matrix model yields an elegant theoretical framework that
brings this method to its full power, see also \cite{ADOS,BasAke}
for the connection between the matrix model and the limiting
physical theory.

The chiral two-matrix model in \cite{ADOS} in its simplest, i.e. \lq
quenched\rq\ form, is defined by the probability distribution
\[ \frac{1}{\widehat Z_n} \exp\left(-n\textrm{Tr}(\Phi^*\Phi+\Psi^*\Psi)\right)
\ud \Phi\ud\Psi, \] defined on pairs of rectangular complex matrices
$(\Phi,\Psi)$, both of size $n\times (n+\nu)$, where $n$ and $\nu$
are integers, the superscript $^*$ stands for the conjugate
transpose, and $\widehat Z_n$ is a normalization constant. Here
$\mathrm d \Phi$ and $\mathrm d \Psi$ are the flat complex Lebesgue
measures on the entries of $\Phi$ and $\Psi$. After the change of
variables, see \cite{ADOS},
\[ \Phi_1=\Phi+\mu_1 \Psi, \quad \Phi_2=\Phi+\mu_2\Psi, \qquad \mu_1,\mu_2 \in \R, \]
 this becomes
\begin{equation}\label{ADOS:model:0} \frac{1}{Z_n}\exp
\left(-n\textrm{Tr}(c_1\Phi_1^*\Phi_1+c_2\Phi_2^*\Phi_2-\tau(\Phi_1^*\Phi_2+\Phi_2^*\Phi_1))\right) \ud
\Phi_1\ud\Phi_2,
\end{equation}
for certain constants $c_1,c_2>0$ and $\tau \in \R$ and a
normalization constant $Z_n$. Here $c_1,c_2,\tau$ are known
functions of $\mu_1$ and $\mu_2$. We assume without loss of
generality that $\tau>0$. For this measure to be finite we also need
$0<\tau^2<c_1c_2$.

As noted in \cite{ADOS}, one can consider the more general model
\begin{equation}\label{ADOS:model} \frac{1}{ Z}
\exp\left(-n\textrm{Tr}(V(\Phi_1^*\Phi_1)+W(\Phi_2^*\Phi_2)-\tau(\Phi_1^*\Phi_2+\Phi_2^*\Phi_1))\right)\ud
\Phi_1\ud\Phi_2,
\end{equation}
where $V,W$ are polynomials with positive leading coefficient. This
is the setting that we will consider in this paper. We assume
without loss of generality that $\nu$ is a nonnegative integer.



We note that the general model in \cite{ADOS} has extra determinantal factors
in the probability distributions \eqref{ADOS:model:0}--\eqref{ADOS:model}. What
we consider here is the \lq quenched\rq\ case where there are no determinants
of this kind. In fact, the correlation functions of the general model can be
expressed in terms of their quenched variants \cite{ADOS,AV}.\smallskip

The chiral two-matrix model can be analyzed with the help of biorthogonal
polynomials \cite{ADOS,EM}. The monic polynomials are denoted $P_{j,n}(x)$ and
$Q_{k,n}(y)$, of degree $j$ and $k$ respectively, and they satisfy the
biorthogonality relations
\begin{equation}\label{biorth}
\int_0^{\infty}\!\!\int_0^{\infty}  w_n(x,y)P_{j,n}(x)Q_{k,n}(y)\ud x \ud y =
\kappa_k\delta_{j,k},\qquad \kappa_k\neq 0,\qquad j,k=0,1,2,\cdots,
\end{equation}
with respect to the weight function
\begin{equation}\label{Bessel:weight}
w_n(x,y) = (xy)^{\nu/2}I_{\nu}(2\tau n\sqrt{xy})e^{-n(V(x)+W(y))},
\end{equation}
see \cite[Eqs.~(2.15)~and~(2.16)]{ADOS}, with $I_{\nu}$ the modified Bessel
function
\begin{equation}\label{def:Inu}
I_{\nu}(x) = \sum_{k=0}^{\infty} \frac{1}{k!\Gamma(k+\nu+1)} (x/2)^{2k+\nu}.
\end{equation}
The weight function is well-defined for all $\nu>-1$, not
necessarily an integer. Following the approach by
Ercolani-McLaughlin \cite{EMcL} one shows that the biorthogonal
polynomials exist, are unique, and have real and simple zeros. In
the case where $V,W$ are both linear, the biorthogonal polynomials
are explicitly known to be given by Laguerre polynomials
\cite{ADOS}.

Incidentally, there also exists a non-Hermitian version of the
chiral two-matrix model, due to Osborn \cite{Osb}. In that case,
$I_{\nu}(x)$ in \eqref{Bessel:weight} is replaced by the modified
Bessel function of the second kind~$K_{\nu}(x)$, which is defined in
\eqref{def:Knu} below.\smallskip

Our interest lies in the singular values of the matrix $\Phi_1$. They form a
determinantal point process with a correlation kernel called $H_N$ in
\cite{ADOS}. We find it more convenient to study the determinantal point
process of the squared singular values of $\Phi_1$ which is then described by
the correlation kernel
\begin{equation}\label{kernel:Kn}
K_n(x_1,x_2) = \sum_{k=0}^{n-1} \frac{1}{\kappa_k}
\left(\int_0^{\infty} w_n(x_1,y)Q_{k,n}(y) \ud y \right)P_{k,n}(x_2),
\end{equation}
with $\kappa_k$ as in \eqref{biorth}. There are three other kernel
functions in \cite{ADOS} (as usual in the Eynard-Mehta setting
\cite{EM}) but we will be only interested in $K_n(x_1,x_2)$. Indeed,
in this paper we will express $K_n(x_1,x_2)$ in terms of a
Riemann-Hilbert (RH) problem. It is an open problem to apply our
method to the mixed correlation functions, which describe the
interaction between the singular values of $\Phi_1$ and $\Phi_2$.

The kernel \eqref{kernel:Kn} describes a well-defined determinantal
point process on the positive real line for any value of $\nu>-1$.
Our results on this point process will also hold for general
$\nu>-1$ and have a random matrix interpretation in case $\nu$ is
integer, i.e.~then the particles correspond to the squared singular
values of the matrix $\Phi_1$ in \eqref{ADOS:model}. In the
following we will often refer to the process described by the kernel
\eqref{kernel:Kn} as the (squared) singular value process, even if
$\nu$ is not an integer.

Summarizing, biorthogonal polynomials allow to obtain the correlation functions
in the chiral two-matrix model. One can also find analytic expressions for the
distributions of the individual singular values, see \cite{AkeDam,AkeIps}.
\smallskip

In this paper we will show that the biorthogonal polynomials
$P_{j,n}(x)$ and $Q_{k,n}(y)$ can be characterized as \emph{multiple
orthogonal polynomials} \cite{KMcL,VAGK} with respect to a suitable
system of weight functions. Consequently we will express the kernel
$K_n(x,y)$ in terms of a RH problem. This yields the chiral analogue
of a result of Kuijlaars-McLaughlin \cite{KMcL}; see also
\cite{BEH,EMcL,Kap} for some RH problems of a different nature for
the non-chiral two-matrix model.

The paper \cite{ADOS} contains a detailed analysis for the case of linear
potentials $V(x)=c_1x$, $W(y)=c_2y$. In the present paper we focus on the
quadratic case
\[
W(y)=y^2/2+\alpha y, \qquad \alpha \in \mathbb{R}.
\]
The other potential $V(x)$ will be allowed to be an arbitrary
polynomial with positive leading coefficient. Under these
assumptions on $V$ and $W$, we will be able to perform a Deift-Zhou
asymptotic analysis of the RH problem, yielding the asymptotics of
the kernel $K_n$. In this way we get the chiral analogues of the
results by Duits-Kuijlaars-Mo \cite{Duits2,DKM,DKM2,Mo}.

Our results imply universality in the case of a quadratic potential
$W(y)$. Universality results for the chiral 1-matrix model were
obtained in \cite{ADMN,KF,KVL}.

Under the additional assumption that $V(x)=x$, we give an $\alpha\tau$-phase
diagram and discuss the phase transitions. In particular the phase diagram
indicates a multicritical point for the values of parameters $\alpha=-1$ and
$\tau=1$. We describe the local behavior of the singular value process near
this multicritical point by means of a triple scaling limit leading to the
chiral version of the main result in \cite{DG}. The new critical kernel that we
obtain in the scaling limit is expressed in terms of a $4\times 4$ matrix-valued RH
problem that was introduced by one of the authors in \cite{Del3} to
describe a critical phenomenon for non-intersecting squared Bessel paths. Our
kernel will be built from the same RH problem, but in an essentially different
way than in \cite{Del3}.


\section{Statement of results}
\label{section:statement:results}

Our first results hold in the general case where $V$ and $W$ are
polynomial potentials with positive leading coefficients. We find it
convenient to rewrite \eqref{Bessel:weight} as
\begin{equation}
\label{Bessel:weight2} w_n(x,y) = f_n(xy)e^{-n(V(x)+W(y))},
\end{equation}
with
\begin{equation}\label{def:f}
f_n(x):=x^{\nu/2}I_{\nu}(2\tau n\sqrt{x}) .
\end{equation}
We also introduce the functions
\begin{equation}\label{hl:def}
h_{l,n}(x) := \int_0^{\infty} y^l f_n(xy)e^{-nW(y)}\ \ud y,
\end{equation}
for $l\in \mathbb{N} \cup\{0\}$ and write
\begin{equation}\label{r:def} r:=\deg W-1.\end{equation}

\subsection{Multiple orthogonality relations}

This is our first main result.

\begin{theorem}\label{theorem:MOP} (Multiple orthogonality
relations)

(a) The biorthogonal polynomial $P_{j,n}(x)$ in \eqref{biorth} is the type~II
multiple orthogonal polynomial with respect to the weight functions
\[  e^{-nV(x)}h_{l,n}(x), \qquad l=0,\ldots,2r. \]
More precisely,
\begin{equation}\label{MOP:1}
\int_0^{\infty} P_{j,n}(x)x^{k}e^{-nV(x)}h_{l,n}(x)\ud x = 0,
\end{equation}
for $l=0,\ldots,2r$ and
$k=0,\ldots,\left\lfloor\frac{j-l-1}{2r+1}\right\rfloor$, where we recall the
notations \eqref{hl:def}--\eqref{r:def} and $\left\lfloor\cdot\right\rfloor$
denotes the integer part of a number.

(b) The polynomial $P_{j,n}(x)$ also satisfies the alternative system of
multiple orthogonality relations
\begin{equation}\label{MOP:2} \int_0^{\infty} P_{j,n}(x)x^{k}w_{l,n}(x)\ud x = 0,
\end{equation}
for $l=0,\ldots,2r$ and
$k=0,\ldots,\left\lfloor\frac{j-l-1}{2r+1}\right\rfloor$, with the weight
functions
\begin{equation}\label{vl:def} w_{l,n}(x) :=
\left\{\begin{array}{ll} e^{-nV(x)}h_{l,n}(x),& l=0,\ldots,r,\\
e^{-nV(x)}x h'_{l-r-1,n}(x),& l=r+1,\ldots,2r. \end{array}\right.
\end{equation}
\end{theorem}

Theorem~\ref{theorem:MOP} will be proved in Section~\ref{section:RHP}. A
similar result holds of course for the biorthogonal polynomials $Q_{j,n}(x)$.

\subsection{Riemann-Hilbert problem and correlation kernel}

Theorem~\ref{theorem:MOP} asserts that the polynomials $P_{j,n}$ satisfy
multiple orthogonality relations of type~II. Hence, they are characterized by
the following RH problem \cite{VAGK}. In the statement of the theorem we will
use the system of multiple orthogonality relations \eqref{MOP:2} rather than
\eqref{MOP:1}. We write $\mathbb{R}^+:=[0,\infty)$.

\begin{rhp} \label{rhp:Y} We look for a $(2r+2)\times (2r+2)$
matrix-valued function $Y : \mathbb C \setminus \mathbb R^+ \to \mathbb
C^{(2r+2) \times (2r+2)}$ satisfying the following conditions.
\begin{enumerate}
\item[\rm (1)] $Y(z)$ is analytic (entrywise) for $z\in\mathbb C\setminus\mathbb R^+$.
\item[\rm (2)] $Y$ has limiting values $Y_{\pm}$ on $\mathbb R^+$,
where $Y_+$ ($Y_-$) denotes the limiting value from the upper (lower)
half-plane, and these limiting values satisfy the jump relation
\begin{equation}\label{defjumpmatrixY}
Y_{+}(x) = Y_{-}(x)
\begin{pmatrix} 1 & w_{0,n}(x) & \cdots & w_{2r,n}(x)\\
0 & 1 &  \cdots &  0 \\
\vdots & \ddots & \ddots  & \vdots \\
0 & 0  & \cdots & 1
\end{pmatrix},
\end{equation}
for $x \in \mathbb R^+$, where $w_{l,n}(x)$, $l=0,\ldots,2r$, is given in
\eqref{vl:def}.
\item[\rm (3)] As $z\to\infty$, we have that
\begin{equation}\label{asymptoticconditionY}
    Y(z) = \left(I+\mathcal O\left(\frac{1}{z}\right)\right)
    \diag(z^{n},z^{-n_0},z^{-n_1},\ldots,z^{-n_{2r}}),
\end{equation}
where $n_l=\left\lfloor\frac{n+2r-l}{2r+1}\right\rfloor$ for $l=0,\ldots,2r$.
\item[\rm (4)] $Y(z)$ has the following behavior near the origin:
  \begin{equation*}
  \begin{aligned}
  Y(z)\diag(1,h(z)^{-1},\ldots,h(z)^{-1})=\mathcal O(1), \\
  Y^{-T}(z)\diag (h^{-1}(z),1,\ldots,1)=\mathcal O(1),
  \end{aligned}
  \end{equation*}
  as $z\to 0$, $z\in\mathbb{C}\setminus\mathbb{R}^+$, where the $\mathcal O$-condition is taken to be
  entrywise,
  the superscript
  ${}^{-T}$ denotes the inverse transpose,   and
  \begin{equation}\label{ha}
  h(z)=\left\{
         \begin{array}{ll}
           |z|^{\nu}, & \quad  \hbox{if $-1<\nu<0$,} \\
           \log|z|, & \quad\hbox{if $\nu=0$,} \\
           1, & \quad \hbox{if $\nu>0$}.
         \end{array}
       \right.
  \end{equation}
\end{enumerate}
\end{rhp}

This RH problem has a unique solution. It is constructed out of the type II
multiple orthogonal polynomials and their Cauchy transforms. In particular, the
$(1,1)$ entry of $Y$ is $P_{n,n}$. The inverse transpose $Y^{-T}$ is given in
terms of the associated type I multiple orthogonal polynomials. We refer to
\cite{VAGK} for the details.

Condition (4) is needed to ensure that the solution of the RH problem is
unique. The exact form of (4) follows from an analysis of the
formulas for $Y$ and $Y^{-T}$ in terms of the multiple orthogonal
polynomials mentioned above. This can be done as in \cite[Proof of
Theorem 2.4]{KMVV2004}. An essential ingredient is that
$w_{l,n}(z)=\mathcal O(z^\nu)$, as $z \to 0$ for $l=0,\ldots,2r$, which
is immediate from \eqref{def:f} and \eqref{vl:def}.

The correlation kernel $K_n$ in \eqref{kernel:Kn} has the
following representation in terms of $Y$, see \cite{DK1}
\begin{equation} \label{eq: kernel in Y}
 K_n(x,y)= \frac{1}{2\pi i(x-y)}\begin{pmatrix}0 & w_{0,n}(y)& w_{1,n}(y)& \cdots &
w_{2r,n}(y)\end{pmatrix}Y_+^{-1}(y)Y_+(x)\begin{pmatrix}1&0&\cdots&0\end{pmatrix}^T,
\end{equation}
where the superscript ${}^T$ denotes the transpose and both row
vectors have length $2r+2$.

This representation allows us to derive the large $n$ limit of the
correlation kernel in the case of a quadratic potential
$W(y)=y^2/2+\alpha y$ by performing a Deift/Zhou steepest descent
analysis on RH problem~\ref{rhp:Y}. As will be clear from the
analysis, this corresponds to a quartic potential in the non-chiral
two-matrix model studied by Duits-Kuijlaars-Mo in
\cite{Duits2,DKM,Mo}. We will largely follow the line of thought in
these works, however, at several places complications will arise.

The key to the steepest descent analysis is a third order differential equation
for $h_{0,n}(x)$ from \eqref{hl:def}, to be stated in Section~\ref{section:thirdorderdiff}. It plays
the same role as the Pearcey equation in \cite{Duits2,DKM} but is considerably
more complicated. The steepest descent analysis itself will be described in
detail in Section~\ref{section:steepestdescent:quadratic}.

\subsection{Vector equilibrium problem and connection with the non-chiral two-matrix model}
\label{subsection:vector:equilibrium:problem}

From this point we will assume that the second potential is quadratic
\begin{equation} \label{eq:quadratic potential W}
W(y)=\frac{y^2}{2}+\alpha y,  \qquad \alpha \in \R,
\end{equation}
and that $V$ is a polynomial with a positive leading coefficient. As one of our
main results we will characterize the limiting mean squared singular value
distribution of $\Phi_1$ in this case using a vector equilibrium problem. More
precisely, it is the first measure of a triplet of measures minimizing an
energy functional under certain conditions. Before we introduce this vector
equilibrium problem we review a related equilibrium problem that arises in the
study of the non-chiral two-matrix model.

\subsubsection*{Vector equilibrium problem for the non-chiral two-matrix model}

In the recent paper \cite{DKM}, Duits, Kuijlaars, and Mo consider a Hermitian
two-matrix model of the form
\begin{equation}\label{twomatrix:DKM}
\frac{1}{Z_n^\DKM}\exp\left(-n\textrm{Tr}(V^{\DKM}(M_1)+W^{\DKM}(M_2)-\tau
M_1M_2)\right)\ud M_1\ud M_2,
\end{equation}
defined on pairs $(M_1,M_2)$ of $n\times n$ Hermitian matrices.
Here, $Z_n^\DKM$ is a normalization constant, $V^{\DKM}$ is a
general even polynomial with a positive leading coefficient,
$W^{\DKM}$ is a quartic polynomial given by
\begin{equation}\label{potential:DKM}
W^{\DKM}(y)=\frac{1}{4}y^4+\frac{\alpha}{2}y^2,\qquad \alpha\in \mathbb{R},
\end{equation}
and $\tau>0$ is the coupling constant. Throughout this paper, we use the
superscript $^{\DKM}$ to distinguish functions and constants related to the
non-chiral two-matrix model from similar notions in the chiral two-matrix
model.

We define, as usual (see \cite{SaffTotikBook}), the logarithmic energy of a
measure $\nu$  by
\begin{equation}
I(\nu)=\iint \log \frac{1}{|x-y|}\ud\nu(x)\ud\nu(y),
\end{equation}
and the mutual energy of two measures $\nu_1$, $\nu_2$ by
\begin{equation}
I(\nu_1,\nu_2)=\iint \log \frac{1}{|x-y|}\ud\nu_1(x)\ud\nu_2(y).
\end{equation}
The main result of \cite{DKM} is that the limiting mean eigenvalue distribution
of $M_1$ can be described by the first component of a triplet
$(\mu_1^{\DKM},\mu_2^{\DKM},\mu_3^{\DKM})$, which are three measures minimizing
the energy functional
\begin{multline}\label{eq:energy functional DKM}
E^{\DKM}(\nu_1,\nu_2,\nu_3):=\sum_{j=1}^3I(\nu_j)-\sum_{j=1}^{2}I(\nu_j,\nu_{j+1})
  +\int V_1^{\DKM}(x)\ud\nu_1(x)+\int V_3^{\DKM}(x)\ud\nu_3(x),
\end{multline}
with $V_1^{\DKM}$ and $V_3^{\DKM}$ being certain symmetric external
fields on $\R$, and where the minimization is among all positive
measures such that
\begin{itemize}
\item[(a)] $\nu_1$ is a measure on $\mathbb{R}$ with total mass $1$;
\item[(b)] $\nu_2$ is a measure on $i\mathbb{R}$ with total mass
$2/3$ that satisfies the constraint
\begin{equation}\label{upper:constaint:DKM}
\nu_2\leq \sigma_2^{\DKM},
\end{equation}
where $\sigma_2^{\DKM}$ is a certain positive symmetric measure on the imaginary axis;
\item[(c)] $\nu_3$ is a measure on $\mathbb{R}$ with total mass
$1/3$;
\item[(d)] $I(\nu_j)< \infty$ for $j=1,2,3$.
\end{itemize}
This equilibrium problem clearly depends on the input data
$V_1^\DKM$, $V_3^\DKM$, and $\sigma_2^\DKM$. For the exact
definitions of these notions we refer to \cite{DKM}. The unique
solvability of the vector equilibrium problem follows from
\cite{DKM,HK}.

\subsubsection*{Vector equilibrium problem for the chiral two-matrix model}

Given the triplet $(V,W,\tau)$ defining the chiral two-matrix model, we define
\begin{equation} \label{eq: relation chiral/hermitian potentials}
V^\DKM(x)=\frac12 V(x^2) \qquad \text{and} \qquad W^\DKM(y)=\frac12 W(y^2).
\end{equation}
Note that this definition is consistent with \eqref{potential:DKM} and
\eqref{eq:quadratic potential W}. The triplet $(V^\DKM,W^\DKM,\tau)$ then
characterizes an associated non-chiral two-matrix model.

The vector equilibrium problem that is appropriate in our chiral setting is a
\lq squared\rq\ version of the above vector equilibrium problem for the
associated non-chiral two-matrix model. More precisely, we consider the energy
functional
\begin{equation}\label{eq:energy functional}
\begin{aligned}
E(\nu_1, \nu_2,\nu_3) := \sum_{j=1}^3I(\nu_j)-\sum_{j=1}^{2}I(\nu_j,\nu_{j+1})+
\int V_1(x)\ud\nu_1(x)+\int V_3(x)\ud\nu_3(x),
\end{aligned}
\end{equation}
and define the input data
\begin{equation}\label{eq:data 2} V_1(x) := 2V_1^{\DKM}(\sqrt{x}),\quad
V_3(x):=2V_3^{\DKM}(\sqrt{x}),\quad \ud\sigma_2(x) :=
2\ud\sigma_2^{\DKM}(i\sqrt{-x}).
\end{equation}
Then $V_1$ and $V_3$ are external fields on $\R^+$ and $\mathrm
\sigma_2$ is a positive measure on $\R^-$. The vector equilibrium
problem is then to minimize $E(\nu_1, \nu_2,\nu_3)$ among all
positive measures $\nu_1$, $\nu_2$ and $\nu_3$ satisfying the
following conditions.
\begin{enumerate}
\item[(a)] $\nu_1$ is a measure on $\mathbb{R}^{+}$ with total mass $1$.
\item[(b)] $\nu_2$ is a measure on $\mathbb{R}^{-}:=(-\infty,0]$ with total mass $2/3$ such that
\begin{equation}\label{eq:upper constraint}
\nu_2\leq \sigma_2,
\end{equation}
where $\sigma_2$ is defined in \eqref{eq:data 2}.
\item[(c)] $\nu_3$ is a measure on $\mathbb{R}^{+}$ with total mass $1/3$.
\item[(d)] $I(\nu_j)< \infty$ for $j=1,2,3$.
\end{enumerate}

This vector equilibrium problem has a unique solution described in the following theorem. We denote the support of a measure $\nu$ by $S(\nu)$.

\begin{theorem} \label{th: equilibrium problem}
The equilibrium problem \eqref{eq:energy functional}--\eqref{eq:upper constraint} has a unique solution $(\mu_1,\mu_2,\mu_3)$. Moreover, if $0 \notin S(\mu_1)$ or $0 \notin S(\sigma_2-\mu_2)$ then
\[
S(\mu_1)=\bigcup_{j=1}^N [a_j,b_j],
\]
for some $N \in \mathbb N$ and $0 \leq a_1 < b_1 < a_2 < \cdots < a_N < b_N$ and on each of the intervals $[a_j,b_j]$ in $S(\mu_1)$ there is a density
\[
\frac{\mathrm d \mu_1}{\mathrm d x}=\rho_1(x)= \begin{cases}
 \frac{1}{\pi}g_j(x) \sqrt{(b_j-x)(x-a_j)},  &  x \in [a_j,b_j], \quad \text{if }a_j>0, \\
 \frac{1}{\pi}g_1(x) \sqrt{(b_1-x)x^{-1}},   &  x \in [0,b_1], \quad \text{if }a_1=0,
 \end{cases}
\]
where $g_j$ is nonnegative and real analytic on $[a_j,b_j]$.
\end{theorem}
\begin{proof}
We claim that
\begin{equation}\label{equil:problems:relation}
(\ud\mu_1(x),\ud\mu_2(x),\ud\mu_3(x)):=(2\ud\mu_1^{\DKM}(\sqrt{x}),2\ud\mu_2^{\DKM}(\sqrt{x}),2\ud\mu_3^{\DKM}(\sqrt{x})),
\end{equation}
where the right hand side denotes the solution to the equilibrium problem
\eqref{eq:energy functional DKM}--\eqref{upper:constaint:DKM} with the same
parameters $\tau,\alpha$. To prove \eqref{equil:problems:relation}, let us
denote for any measure $\mu$ symmetric with respect to the origin, the squared
measure $\widehat \mu$ by the rule $\ud \what \mu(x) = 2\ud\mu(\sqrt{x})$.
Since $\mu$ is symmetric we do not have to precise the branch cut of the square
root.
Thus if $\mu$ has a density  $\ud\mu(x) = \rho(x)\ud x$ then
$$\ud \what \mu(x) := 2\rho(\sqrt{x})\ud \sqrt{x} = \frac{\rho(\sqrt{x})}{\sqrt{x}}\ud x.$$
Now the mutual energies of two measures $\nu_1,\nu_2$ and the corresponding squared
measures $\what \nu_1,\what \nu_2$ are related by
$$ I(\what \nu_1,\what \nu_2) = 2I(\nu_1,\nu_2),
$$
see e.g.\ \cite[Th.~IV.1.10(f)]{SaffTotikBook}. With the help of
this relation, and using \eqref{eq:data 2}, \eqref{potential:DKM}
and \eqref{eq:quadratic potential W}, the claimed relation
\eqref{equil:problems:relation} follows.

Given this, the theorem is a corollary of \cite[Theorem 1.1]{DKM}.
\end{proof}

\begin{remark}
If $a_1=0$ the density of $\mu_1$ blows up like an inverse square root at the origin as is the case for the Marchenko-Pastur density \cite{MaPa}.
\end{remark}

\subsection{Classification into cases} \label{sec: classification}
In this model we distinguish a number of regular and singular cases depending
on the supports and densities of the measures $\mu_1$, $\sigma_2-\mu_2$, and
$\mu_3$. The classification in our chiral setting is inherited from the
classification for the associated non-chiral models given in \cite[Section
1.5]{DKM}, e.g. we say that our chiral model belongs to Case~I if the
associated non-chiral model is in Case~I according to \cite[Section 1.5]{DKM}.
This leads us to distinguish five generic cases and 8 singular cases. The
classification of the five generic phases depends on whether $0$ is in the
support of the measures $\mu_1$, $\sigma_2-\mu_2$, and $\mu_3$ or not. We have
the following cases with regular behavior of the supports at zero:
\begin{description}
\item[Case I] $0\in S(\mu_1)$, $0 \notin S(\sigma_2-\mu_2)$ and $0\in S(\mu_3)$,
\item[Case II] $0\notin S(\mu_1)$, $0 \notin S(\sigma_2-\mu_2)$ and $0\in S(\mu_3)$,
\item[Case III] $0\notin S(\mu_1)$, $0 \in S(\sigma_2-\mu_2)$ and $0\notin S(\mu_3)$,
\item[Case IV] $0\in S(\mu_1)$, $0 \notin S(\sigma_2-\mu_2)$ and $0\notin S(\mu_3)$,
\item[Case V] $0\notin S(\mu_1)$, $0 \notin S(\sigma_2-\mu_2)$ and $0\notin S(\mu_3)$.
\end{description}

We discuss only three of the singular cases. For the remaining five
singular cases we refer to~\cite{DKM}. The critical phenomena
corresponding to these remaining five cases can already be found in
the one-matrix model, so we will not discuss them here.
\begin{description}
\item[Singular supports I] $0\in S(\mu_1)\cap S(\sigma_2-\mu_2)$, $0\notin
S(\mu_3)$.
\item[Singular supports II] $0\notin S(\mu_1)$, $0 \in
S(\sigma_2-\mu_2)\cap S(\mu_3)$.
\item[Singular supports III] $0 \in S(\mu_1)\cap S(\sigma_2-\mu_2)\cap
S(\mu_3)$.
\end{description}
The last one is a multi-singular case.

Except for Case V, all these phenomena already occur in the simplest version of
the model in which the potential $V(x)=x$. For this situation we will establish
a phase diagram in Section~\ref{sec: phase diagram}.

The next theorem will only be proved in the regular cases that we define as follows.
\begin{definition} \label{def: regular cases}
The triplet $(V,W,\tau)$ is regular if the associated triplet
$(V^\DKM,W^\DKM,\tau)$ is regular in the sense of \cite[Definition
1.3]{DKM}, with $ V^\DKM$ and $W^\DKM$ as in \eqref{eq: relation
chiral/hermitian potentials}.
\end{definition}
In particular for a regular triplet it holds that $S(\mu_1) \cap
S(\sigma_2-\mu_2)=\emptyset$ and $S(\mu_3) \cap
S(\sigma_2-\mu_2)=\emptyset$. Moreover the functions $g_j$,
$j=1,\ldots,N$, in Theorem \ref{th: equilibrium problem} are nonzero
for all $x \in [a_j,b_j]$, i.e. the density of the measure $\mu_1$
does not vanish in the interior of the support and behaves like a
square root at the nonzero endpoints of the support. The same holds
for the densities of $\sigma_2-\mu_2$ and $\mu_3$. There is an extra
condition on the variational inequality for $\mu_1$ which guarantees
that no extra interval emerges in the support of $\mu_1$ when
continuously varying the potentials.

\subsection{Limiting mean singular value distribution}

The measure $\mu_1$ is the limiting mean squared singular value
distribution of the matrix $\Phi_1$ in the chiral two-matrix model
as $n \to \infty$. This statement holds for $W$ as in
\eqref{eq:quadratic potential W} and general $V$ and $\tau$, but we
will only prove it for the regular cases, see Definition~\ref{def:
regular cases}

\begin{theorem} \label{th: noncrit limit}
Suppose $(V,W,\tau)$ is regular. Let $\mu_1$ be the first component
of the minimizer $(\mu_1,\mu_2,\mu_3)$ of the vector equilibrium
problem \eqref{eq:energy functional}--\eqref{eq:upper constraint}.
Then $\mu_1$ is the limiting mean particle distribution of the
determinantal point process with correlation kernel
\eqref{kernel:Kn} as $n \to \infty$ with $n \equiv 0 \mod 3$.
\end{theorem}
This theorem is equivalent to the statement
\[
\lim_{n \to \infty}\frac{1}{n}K_n(x,x)=\rho_1(x)=\frac{\mathrm d
\mu_1}{\mathrm d x}(x), \qquad x>0, \qquad n \equiv 0 \mod 3.
\]
The proof is given in Section \ref{sec: proof noncrit theorem}. It
is based on a lengthy Deift/Zhou steepest descent analysis carried
out in Section \ref{section:steepestdescent:quadratic}. Without too
much extra effort one could also obtain the universal scaling limits
that are typical for unitary random matrix ensembles. More
precisely, if $\rho_1(x^*)>0$, where $\rho_1$ denotes the density of
the measure $\mu_1$, we retrieve the sine kernel as a scaling limit
\[
\lim_{n \to \infty} K_n\left(x^*+\frac{x}{n\rho_1(x^*)},x^*+\frac{y}{n\rho_1(x^*)}\right)=\frac{\sin \pi (x-y)}{\pi (x-y)}.
\]
If $x^*$ is a nonzero endpoint of $S(\mu_1)$, i.e. $x^* \in \{a_1,b_1,\ldots,a_N,b_N\}\setminus \{0\}$, then
\[
\lim_{n \to \infty} \frac{1}{(cn)^{2/3}}K_n\left(x^*\pm\frac{x}{(cn)^{2/3}},x^*\pm\frac{y}{(cn)^{2/3}}\right)=\frac{\Ai(x)\Ai'(y)-\Ai'(x)\Ai(y)}{x-y},
\]
where we choose the $+$-sign ($-$-sign) if $x^*$ is a left (right)
endpoint of the support of $\mu_1$. Recall that we are in the
regular case so that the density vanishes like a square root at $x^*
> 0$. However, if $x^*=a_1=0$ then the density blows up as an
inverse square root and we would obtain the Bessel kernel of order
$\nu$ as a scaling limit
\[
\lim_{n \to \infty} \frac{1}{(cn)^{2}}K_n\left(\frac{x}{(cn)^{2}},\frac{y}{(cn)^{2}}\right)=\frac{J_\nu(\sqrt x)\sqrt y J'_\nu (\sqrt y)-\sqrt x J_\nu'(\sqrt x)J_\nu(\sqrt y)}{2(x-y)},
\]
for $x,y>0$ and a suitable constant $c>0$. We will not discuss this any further.

\subsection{Phase diagram in the quadratic/linear case} \label{sec: phase diagram}

From here we restrict ourselves to the very specific model of quadratic and linear potentials
\begin{equation} \label{eq: quadratic/linear potentials}
V(x)=x, \qquad W(y)=\frac{y^2}{2}+\alpha y, \quad \alpha \in \R.
\end{equation}
For this concrete model we can construct a phase diagram, i.e. determine which values of $(\alpha,\tau)$ correspond to which case of the classification in Section \ref{sec: classification}. It turns out that the case `Singular supports III' occurs and this will get most of our attention. To establish the phase diagram, we first discuss the behavior of the supports of the measures $\mu_1$, $\sigma_2-\mu_2$, and $\mu_3$
and how they depend on $\alpha$ and $\tau$. It follows from
\eqref{equil:problems:relation} and \cite{DGK} that the supports of the
measures $\mu_1$, $\sigma_2-\mu_2$, and $\mu_3$ have the following form
\begin{align*}
    \supp (\mu_1) & =[\beta_1, \beta_0], \\
    \supp(\sigma_2-\mu_2) & = (-\infty,-\beta_2], \\
    \supp(\mu_3) &= [\beta_3,\infty),
\end{align*}
for some $\beta_0 > \beta_1 \geq 0$, $\beta_2,\beta_3 \geq 0$ that all depend
on the values of $\alpha \in \mathbb R$ and $\tau > 0$. We distinguish a number
of cases, depending on whether $\beta_1$, $\beta_2$, or $\beta_3$ are equal to
zero, or not. At least one of these is zero, and generically, no two
consecutive ones are zero. According to the classification in Section \ref{sec: classification} we have
\begin{description}
\item[Case I:] $\beta_1=0$, $\beta_2>0$, and $\beta_3=0$.
\item[Case II:] $\beta_1>0$, $\beta_2 > 0$, and $\beta_3=0$.
\item[Case III:] $\beta_1>0$, $\beta_2 = 0$, and $ \beta_3>0$.
\item[Case IV:] $\beta_1=0$, $\beta_2>0$, and $\beta_3>0$.
\end{description}
Case V does not occur in this specific model.

The phase diagram in Figure \ref{fig: phase diagram} shows which
values of $(\alpha,\tau)$ correspond to these four cases. The different
cases are separated by the curves given by the equations
\begin{align*}
   \tau = \sqrt{\alpha + 2}, \quad -2 \leq \alpha < \infty, \quad \text{and} \quad
   \tau = \sqrt{- \frac{1}{\alpha}}, \quad -\infty < \alpha < 0. \end{align*}
On these critical curves two of the numbers $\beta_1, \beta_2$ and $\beta_3$ are
equal to zero. For example, on the curve between Case III and Case IV, we have
$\beta_1=\beta_2=0$, while $\beta_3>0$. Finally, note the multi-critical point
$(\alpha,\tau)=(-1,1)$ in the phase diagram, where $\beta_1=\beta_2=\beta_3 =
0$. All four cases come together at this point in the $(\alpha,\tau)$-plane.
The nature of this multi-critical point is discussed in the next section.

We do not study the other types of critical behavior in this paper,
but we shortly comment on them now. The geometry of the supports of
$\mu_1$, $\sigma_2-\mu_2$, and $\mu_3$ suggests that the transition
on the curve $\tau=\sqrt{\alpha+2}$, $\alpha>-1$, is described by
the inhomogeneous Painlev\'e II kernel as in \cite{Claeys2}. The
transition on the curve $\tau=\sqrt{-1/\alpha}$, $\alpha<-1$, on the
other hand is described by the chiral version of the Pearcey kernel
as in \cite{KMFW2}. The transitions on the dashed lines in Figure
\ref{fig: phase diagram} are situated on the non-physical sheets of
the underlying Riemann surface and, hence, do not affect the local
correlations of the singular values of $\Phi_1$.

The curve $\tau=\sqrt{\alpha+2}$, $\alpha>-2$ is also critical for
the singular values of $\Phi_2$. The other curve is not critical in
that context.
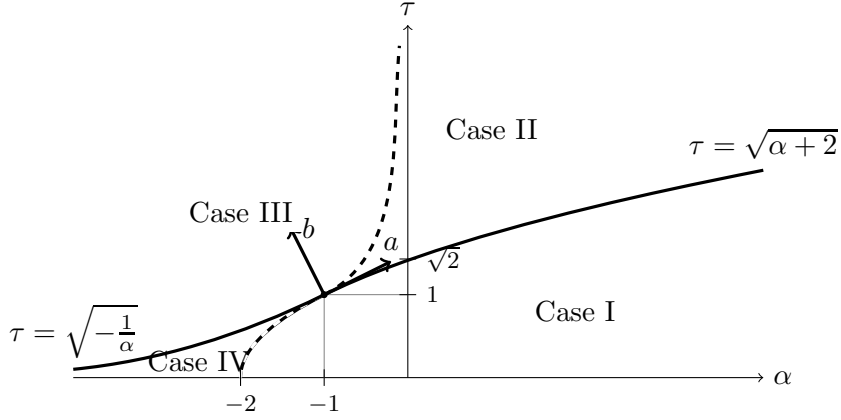
\begin{figure}[t]
\begin{center}
\begin{tikzpicture}[scale=1.1]
\draw[->](0,0)--(0,4.25) node[above]{$\tau$}; \draw[->](-4,0)--(4.25,0)
node[right]{$\alpha$}; \draw[help lines] (-1,0)--(-1,1)--(0,1); \draw[very
thick,rotate around={-90:(-2,0)}] (-2,0) parabola (-4.5,6.25)
node[above]{$\tau=\sqrt{\alpha+2}$}; \draw[very thick,dashed,rotate
around={-90:(-2,0)},color=white] (-2,0) parabola (-3,1) ; \draw[very
thick,dashed] (-1,1)..controls (0,1.5) and (-0.2,3).. (-0.1,4); \draw[very
thick]              (-1,1)..controls (-2,0.5) and (-3,0.2).. (-4,0.1)
node[above]{$\tau=\sqrt{-\frac{1}{\alpha}}$}; \filldraw  (-1,1) circle (1pt);
\draw[very thick,->] (-1,1)--(-0.2,1.4) node[above]{$a$}; \draw[very thick,->]
(-1,1)--(-1.4,1.8) node[right]{$b$}; \draw (0.1,1)
node[font=\footnotesize,right]{$1$}--(-0.1,1); \draw (-1,0.1)--(-1,-0.1)
node[font=\footnotesize,below]{$-1$}; \draw (-2,0.1)--(-2,-0.1)
node[font=\footnotesize,below]{$-2$}; \draw (0.1,1.43) node
[font=\footnotesize,right]{$\sqrt 2$}--(-0.1,1.43); \draw[very thick] (2,0.8)
node[fill=white]{Case I}
                  (-2.5,0.2) node{Case IV}
                  (-2,2) node[fill=white]{Case III}
                  (1,3) node[fill=white]{Case II};
\end{tikzpicture}
\end{center}
\caption{The phase diagram in the $\alpha \tau$-plane: the critical curves $\tau=\sqrt{\alpha+2}$ and $\tau=\sqrt{-\frac{1}{\alpha}}$ separate the four cases. The cases are distinguished by the fact whether $0$ is in the support of the measures $\mu_1$, $\sigma_2-\mu_2$, and $\mu_3$, or not.}
\label{fig: phase diagram}
\end{figure}

Figure \ref{fig: phase diagram cm revisited} shows the phase diagram again together with the shape of the limiting mean squared singular value densities above, below, and on the critical curve.

\begin{figure}[t]
\begin{center}
\begin{tikzpicture}[scale=1]
\draw[white] (-8,-3.5)rectangle (9.5,5);
\begin{scope}[shift={(2.5,0.7)}, scale=0.4]
\draw[->] (0,3)--(0,0)--(6,0);
\draw[very thick] (0.1,3) .. controls (0.3,0) and (5,2) .. (5,0);
\draw (0.1,3) .. controls (0.3,0) and (5,2) .. (5,0);
\draw (0,0) circle (0.5cm) node[below]{$\sim x^{-1/2}$};
\end{scope}
\begin{scope}[shift={(-3.5,2)}, scale=0.4]
\draw[->] (0,3)--(0,0)--(6,0);
\draw[very thick]  (1,0) .. controls (1,2.5) and (2,2) .. (3,2)
                          (3,2) .. controls (3.5,2) and (5,2) .. (5,0);

\draw (1,0) .. controls (1,2.5) and (2,2) .. (3,2)
                          (3,2) .. controls (3.5,2) and (5,2) .. (5,0);

\end{scope}
\begin{scope}[shift={(5,2.3)}, scale=0.3]
\draw[->] (0,3)--(0,0)--(6,0);
\draw[very thick] (0,0) .. controls (0,2.5) and (2,2) .. (3,2)
                              (3,2) .. controls (3.5,2) and (5,2) .. (5,0);
\draw (0,0) circle (0.5cm) node[below]{$\sim x^{1/2}$};
\draw (0,0) .. controls (0,2.5) and (2,2) .. (3,2)
          (3,2) .. controls (3.5,2) and (5,2) .. (5,0);
\end{scope}
\begin{scope}[shift={(-6,-0.3)}, scale=0.3]
\draw[->] (0,3)--(0,0)--(6,0);
\draw[very thick] (0.1,3) .. controls (0.3,0) and (5,2) .. (5,0);
\draw (0.1,3) .. controls (0.3,0) and (5,2) .. (5,0);
\draw (0,0) circle (0.5cm) node[below]{$\sim x^{-1/3}$};
\end{scope}
\begin{scope}[shift={(-2,-2.5)}, scale=0.5]
\draw[->] (0,3)--(0,0)--(6,0);
\draw[very thick] (0.1,3) .. controls (0.3,0) and (5,2) .. (5,0);
\draw(0,0) circle (0.5cm) node[below]{$\sim x^{-1/4}$};
\end{scope}
\draw[->](0,0)--(0,4.25) node[above]{$\tau$};
\draw[->](-4,0)--(4.25,0) node[right]{$\alpha$};
\draw[help lines] (-1,0)--(-1,1)--(0,1);
\draw[very thick,dashed] (-1,1)..controls (-2,0.5) and (-3,0.2).. (-4,0.1);
\begin{scope}
\clip (-1,1) rectangle (7,4);
\draw[very thick,dotted,rotate around={-90:(-2,0)}] (-2,0) parabola (-4.5,6.25);
\end{scope}
\draw (0.1,1) node[font=\footnotesize,right]{$1$}--(-0.1,1);
\draw (-1,0.1)--(-1,-0.1) node[font=\footnotesize,below]{$-1$};
\filldraw  (-1,1) circle (2pt);
\draw[help lines, ->] (-0.8,0.8) arc (20:-20:3.5);
\end{tikzpicture}\end{center}
\vspace{-5mm}\caption[$\alpha\tau$-phase diagram with density plots]{The phase diagram in the $\alpha\tau$-plane with density plots. Below the critical curve the limiting mean squared singular value density of $\Phi_1$ is supported on one interval touching the origin. At the origin the density blows up like an inverse square root. Above this line the singular values cluster on an interval away from zero. On the critical curve the transition between these two regimes occurs. There the density is supported on an interval including zero. On the part of critical curve to the right of the point $(-1,1)$ (dotted) the density vanishes like a square root. To the left of this point (dashed) the exponent at the origin is $-1/3$. At the point $(-1,1)$ we have a transition between these two types of critical behavior. There the density blows up with an exponent $-1/4$.}
\label{fig: phase diagram cm revisited}
\end{figure}
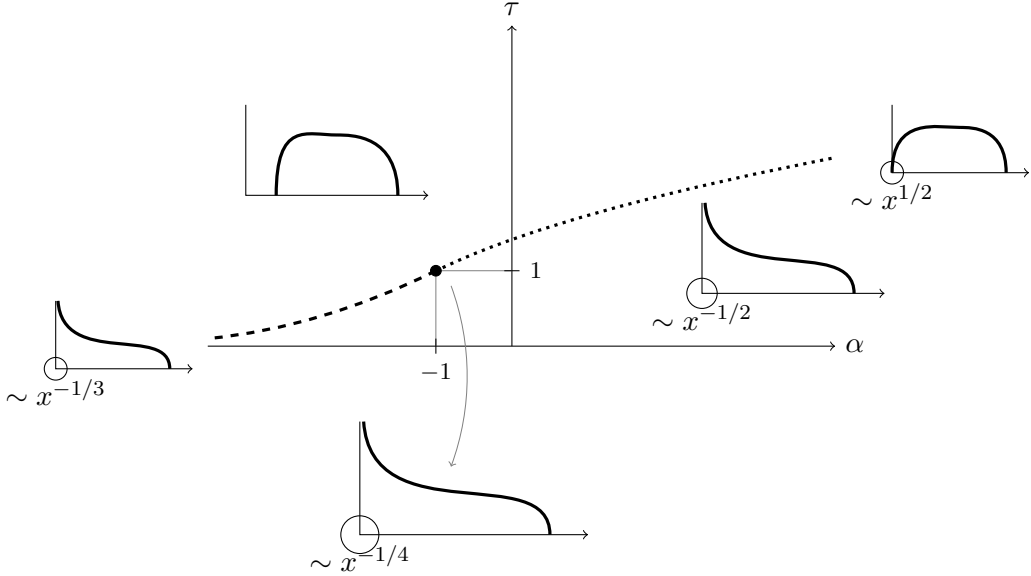

\subsection{A triple scaling limit}

Let us now focus on the squared singular value process of $\Phi_1$ near the critical parameters $\tau=1$ and $\alpha=-1$, by means of a triple scaling limit. To this end, we rescale $\alpha$ and $\tau$ near the critical values in the following way
\begin{equation} \label{eq: scaling alpha tau}
\begin{pmatrix} \alpha \\ \tau \end{pmatrix} = \begin{pmatrix} -1 \\ 1 \end{pmatrix} + a n^{-1/3}\begin{pmatrix} 2 \\ 1 \end{pmatrix} + b n^{-2/3}\begin{pmatrix} -1 \\ 2 \end{pmatrix},
\end{equation}
for  $a,b\in \R$.  We  also scale the space variables with
\begin{align*}
x=u n^{-4/3}, \qquad \text{and} \qquad y=v n^{-4/3}, \qquad u,v>0
\end{align*}
and compute the limiting behavior of $K_n(x,y)$ as $n\to \infty$.

The limiting kernel is characterized by the solution to the
following RH problem introduced in \cite{Del3}. The RH problem has
jumps on a contour in the complex plane consisting of $10$ rays
emanating from the origin. More precisely, we fix two numbers
$\varphi_1,\varphi_2$ such that $0<\varphi_1<\varphi_2<\pi/2$ and
define the half-lines $\Gamma_k$, $k=0,\ldots,9$, by
\begin{equation}\label{def:rays1}
\Gamma_0=\er^+,\quad \Gamma_1=e^{i\varphi_1}\er^+,\quad
\Gamma_2=e^{i\varphi_2}\er^+,\quad \Gamma_3=e^{i(\pi-\varphi_2)}\er^+,\quad
\Gamma_4=e^{i(\pi-\varphi_1)}\er^+,
\end{equation}
and
\begin{equation}\label{def:rays2}
\Gamma_{5+k}=-\Gamma_k,\qquad k=0,\ldots,4.
\end{equation}
All rays $\Gamma_k$, $k=0,\ldots,9$, are oriented towards infinity,
as shown in Figure~\ref{fig:modelRHP}. We also denote by $\Omega_k$
the region in $\cee$ which lies between the rays $\Gamma_k$ and
$\Gamma_{k+1}$, for $k=0,\ldots,9$, where we identify
$\Gamma_{10}:=\Gamma_0$. Now we can state the RH problem.

\begin{rhp}\label{rhp:modelM}
We look for a $4\times 4$ matrix-valued function
$M:\cee\setminus\left(\bigcup_{k=0}^{9}\right)\rightarrow \cee^{4
\times 4}$ (which also depends parametrically on $\til\nu>-1/2$ and
on the complex parameters $r_1,r_2,s,t \in\cee$) satisfying the
following conditions.
\begin{itemize}
\item[(1)] $M(\zeta)$ is analytic (entrywise) for $\zeta\in\cee\setminus\left(\bigcup_{k=0}^{9}
\Gamma_k\right)$.
\item[(2)] For $\zeta\in\Gamma_k$, the limiting values
\[ M_+(\zeta) = \lim_{\substack{z \to \zeta \\ z \textrm{ on $+$-side of }\Gamma_k}} M(z), \qquad
    M_-(\zeta) = \lim_{\substack{z \to \zeta \\ z \textrm{ on $-$-side of }\Gamma_k}} M(z) \]
exist, where the $+$-side and $-$-side of $\Gamma_k$ are the sides which lie on
the left and right of $\Gamma_k$, respectively, when traversing $\Gamma_k$
according to its orientation. These limiting values satisfy the jump relation
\begin{equation*}
M_{+}(\zeta) = M_{-}(\zeta)J_k(\zeta),\qquad k=0,\ldots,9,
\end{equation*}
where the jump matrix $J_k(\zeta)$ for each ray $\Gamma_k$ is shown in
Figure~\ref{fig:modelRHP}.
\item[(3)] As $\zeta\to\infty$ we have
\begin{multline*}
 M(\zeta) =
\left(I+\frac{M_1}{\zeta}+\frac{M_2}{\zeta^2}+O\left(\frac{1}{\zeta^3}\right)\right)
\diag((-\zeta)^{-1/4},\zeta^{-1/4},(-\zeta)^{1/4},\zeta^{1/4})
\\
\times
A\diag\left(e^{-\psi_2(\zeta)+t\zeta},e^{-\psi_1(\zeta)-t\zeta},e^{\psi_2(\zeta)+t\zeta},e^{\psi_1(\zeta)-t\zeta}\right)
\end{multline*}
where the coefficient matrices $M_1,M_2,\ldots$ depend on the parameters
$\til\nu$, $r_1,r_2$, $s$ and $t$, but not on $\zeta$, and where we define
\begin{equation}\label{mixing:matrix}
A:=\frac{1}{\sqrt{2}}\begin{pmatrix} 1 & 0 & -i & 0 \\
0 & 1 & 0 & i \\
-i & 0 & 1 & 0 \\
0 & i & 0 & 1 \\
\end{pmatrix},
\end{equation}
\begin{equation}\label{def:theta1}
\psi_1(\zeta) = \frac{2}{3}r_1\zeta^{3/2}+2s\zeta^{1/2},\qquad \psi_2(\zeta) =
\frac{2}{3}r_2(-\zeta)^{3/2}+2s(-\zeta)^{1/2}.
\end{equation}
\item[(4)] $M(\zeta)$ behaves for $\zeta\to 0$ as
\begin{equation*} M(\zeta) = \mathcal O(\zeta^{\til\nu}),\quad M^{-1}(\zeta) =
\mathcal O(\zeta^{\til\nu}),\qquad \hbox{if $\til\nu\leq 0$},
\end{equation*}
and \begin{align*} \left\{\begin{array}{ll}
M(\zeta)\diag(\zeta^{-\til\nu},\zeta^{\til\nu},\zeta^{\til\nu},\zeta^{-\til\nu})=\mathcal
O(1),& \quad
\zeta\in \Omega_1\cup\Omega_8,\\
M(\zeta)\diag(\zeta^{\til\nu},\zeta^{-\til\nu},\zeta^{-\til\nu},\zeta^{\til\nu})=\mathcal
O(1),&\quad
\zeta\in \Omega_3\cup\Omega_6,\\
M(\zeta)=\mathcal O(\zeta^{-\til\nu}),&\quad\zeta\not\in
(\Omega_1\cup\Omega_3\cup\Omega_6\cup\Omega_8),\end{array}\right. \hbox{if
$\til\nu\geq 0$}.
\end{align*}
\end{itemize}
\end{rhp}

\begin{figure}[t]
\vspace{14mm}
\begin{center}
   \setlength{\unitlength}{1truemm}
   \begin{picture}(100,70)(-5,2)
       \put(40,40){\line(1,0){40}}
       \put(40,40){\line(-1,0){40}}
       \put(40,40){\line(2,1){30}}
       \put(40,40){\line(2,-1){30}}
       \put(40,40){\line(-2,1){30}}
       \put(40,40){\line(-2,-1){30}}
       \put(40,40){\line(2,3){18.5}}
       \put(40,40){\line(2,-3){17}}
       \put(40,40){\line(-2,3){18.5}}
       \put(40,40){\line(-2,-3){17}}
       \put(40,40){\thicklines\circle*{1}}
       \put(39.3,36){$0$}
       \put(60,40){\thicklines\vector(1,0){.0001}}
       \put(20,40){\thicklines\vector(-1,0){.0001}}
       \put(60,50){\thicklines\vector(2,1){.0001}}
       \put(60,30){\thicklines\vector(2,-1){.0001}}
       \put(20,50){\thicklines\vector(-2,1){.0001}}
       \put(20,30){\thicklines\vector(-2,-1){.0001}}
       \put(50,55){\thicklines\vector(2,3){.0001}}
       \put(50,25){\thicklines\vector(2,-3){.0001}}
       \put(30,55){\thicklines\vector(-2,3){.0001}}
       \put(30,25){\thicklines\vector(-2,-3){.0001}}

       \put(60,41){$\Gamma_0$}
       \put(60,52.5){$\Gamma_1$}
       \put(47,57){$\Gamma_2$}
       \put(29,57){$\Gamma_3$}
       \put(16,52.5){$\Gamma_4$}
       \put(16,41){$\Gamma_5$}
       \put(16,30.5){$\Gamma_6$}
       \put(29,20){$\Gamma_7$}
       \put(48,20){$\Gamma_8$}
       \put(60,30){$\Gamma_{9}$}

       \put(65,45){$\small{\Omega_0}$}
       \put(58,58){$\small{\Omega_1}$}
       \put(37,62){$\small{\Omega_2}$}
       \put(17.5,58){$\small{\Omega_3}$}
       \put(10,46){$\small{\Omega_4}$}
       \put(10,34){$\small{\Omega_5}$}
       \put(17.5,22){$\small{\Omega_6}$}
       \put(38,17){$\small{\Omega_7}$}
       \put(58,22){$\small{\Omega_8}$}
       \put(65,33){$\small{\Omega_{9}}$}

       \put(78.5,40){$\small{\begin{pmatrix}0&0&1&0\\ 0&1&0&0\\ -1&0&0&0\\ 0&0&0&1 \end{pmatrix}}$}
       \put(69,57){$\small{\begin{pmatrix}1&0&0&0\\ 0&1&0&0\\ 1&0&1&0\\ 0&0&0&1 \end{pmatrix}}$}
       \put(45,74){$\small{\begin{pmatrix}1&0&0&0\\ -e^{\til\nu\pi i}&1&0&0\\ 0&0&1&e^{\til\nu\pi i}\\ 0&0&0&1 \end{pmatrix}}$}
       \put(4,74){$\small{\begin{pmatrix}1&e^{-\til\nu\pi i}&0&0\\ 0&1&0&0\\ 0&0&1&0\\ 0&0&-e^{-\til\nu\pi i}&1 \end{pmatrix}}$}
       \put(-16,57){$\small{\begin{pmatrix}1&0&0&0\\ 0&1&0&0\\ 0&0&1&0\\ 0&-1&0&1\end{pmatrix}}$}
       \put(-26,40){$\small{\begin{pmatrix}1&0&0&0\\ 0&0&0&-1\\ 0&0&1&0\\ 0&1&0&0 \end{pmatrix}}$}
       \put(-16,23){$\small{\begin{pmatrix}1&0&0&0\\ 0&1&0&0\\ 0&0&1&0\\ 0&-1&0&1 \end{pmatrix}}$}
       \put(4,6){$\small{\begin{pmatrix}1&-e^{\til\nu\pi i}&0&0\\ 0&1&0&0\\ 0&0&1&0\\ 0&0&e^{\til\nu\pi i}&1 \end{pmatrix}}$}
       \put(45,6){$\small{\begin{pmatrix}1&0&0&0\\ e^{-\til\nu\pi i}&1&0&0\\ 0&0&1&-e^{-\til\nu\pi i}\\ 0&0&0&1 \end{pmatrix}}$}
       \put(69,23){$\small{\begin{pmatrix}1&0&0&0\\ 0&1&0&0\\ 1&0&1&0\\ 0&0&0&1\end{pmatrix}}$}
  \end{picture}
   \vspace{0mm}
   \caption{The figure shows the jump contours $\Gamma_k$ in the complex $\zeta$-plane and the corresponding jump
   matrix
   $J_k$ on $\Gamma_k$, $k=0,\ldots,9$, in the RH problem for  $M = M(\zeta)$. We denote by
   $\Omega_k$ the region between
    the rays $\Gamma_k$ and $\Gamma_{k+1}$.}
   \label{fig:modelRHP}
\end{center}
\end{figure}
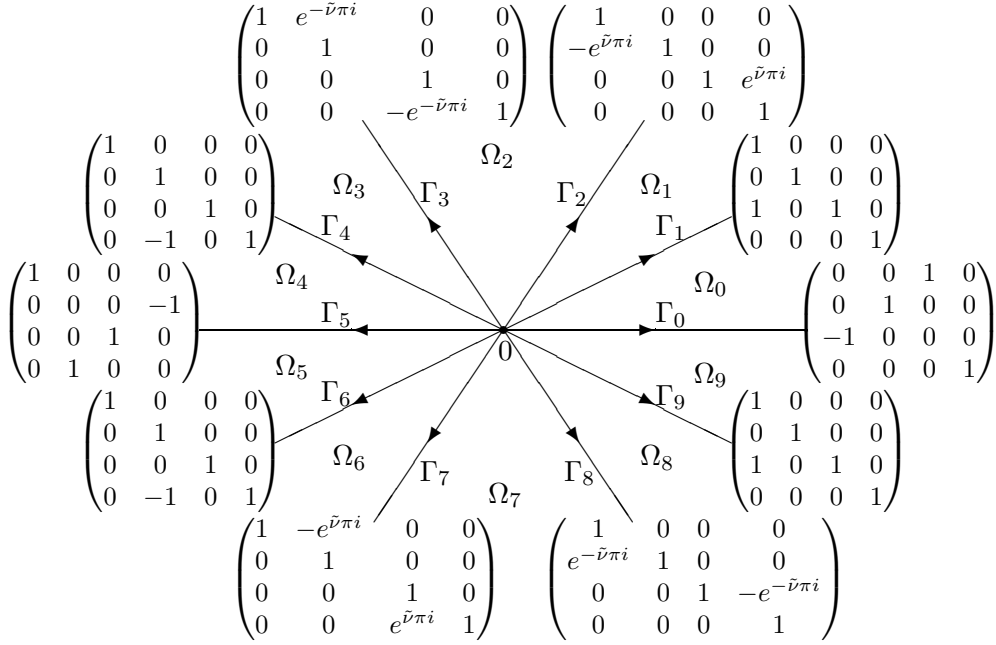

It was proved in \cite{Del3} that this RH problem is solvable for
$r_1=r_2>0$ and $s,t \in \R$. In the same paper this RH problem is
shown to be related to the Hastings-McLeod solution of the
inhomogeneous Painlev\'e II equation $q''(x) = xq(x)+2q^3(x)-\til
\nu-1/2$.

We transform the RH matrix $M(z)$ into a new matrix $\what M(z)$ as follows
\begin{equation}\label{Mhat}
\what M(z) := \diag(z^{1/4},z^{-1/4},z^{1/4},z^{-1/4})\diag\left(\begin{pmatrix} 1 & -1 \\
1 & 1\end{pmatrix},\begin{pmatrix} 1 & 1 \\
-1 & 1\end{pmatrix}\right)M(i z^{1/2}).
\end{equation}
The transformed matrix $\what M(z)$ depends on the same parameters $r_1,r_2,s,t,\til\nu$ as $M(z)$.
The matrix $\what M(z)$ satisfies a RH problem by itself but we will not state
it here.

For $u,v>0$, we now define the critical kernel $\Kcr(u,v;s,t,\nu)$ by
\begin{multline} \label{eq: critical kernel}
\Kcr(u,v;s,t,\nu)=\frac{1}{2\pi i (u-v)}\begin{pmatrix} 0 & 0 & -ie^{-\nu \pi
i} & 1 \end{pmatrix}
\what M\left(v;s,t, \til \nu \right)^{-1}\\ \times \what M\left(u;s,t,\til \nu \right)
 \begin{pmatrix} 0 & 0 & -ie^{\nu \pi i} & 1
 \end{pmatrix}^T,
\end{multline}
where $\what M(z;s,t, \til \nu)$ is defined by \eqref{Mhat} with $M(\zeta)=M(\zeta;s,t, \til \nu)$ the unique solution to RH problem
\ref{rhp:modelM} with parameters
\begin{equation}\label{eq: parameters}
\begin{cases}
r_1=r_2=2, \\
s,t\in \R,\\
\til\nu=\nu+1/2.
\end{cases}
\end{equation}
It can be shown as in \cite[Appendix A]{DG} that this kernel is different from the hard-edge tacnode kernel discovered in \cite{Del3}.

We can then state our final main result.

\begin{theorem} \label{th: triple scaling limit}
Assume $\nu>-1$ and let $K_{n}$ be the kernel in \eqref{kernel:Kn} with $V$ and $W$ as in \eqref{eq:
quadratic/linear potentials}. If $\nu$ is integer this kernel describes the squared singular
values of $M_1$ when averaged over $M_2$. Set
\begin{equation*}
\begin{pmatrix} \alpha \\ \tau \end{pmatrix} = \begin{pmatrix} -1 \\ 1 \end{pmatrix} + a n^{-1/3}\begin{pmatrix} 2 \\ 1 \end{pmatrix} + b n^{-2/3}\begin{pmatrix} -1 \\ 2 \end{pmatrix},
\end{equation*}
for  $a,b\in \R$. Then for $n\to \infty$ and $n\equiv 0 \mod 3$
\begin{equation}\label{kernel:crit:thm}
\lim_{n \to
\infty}\frac{1}{n^{4/3}}K_n\left(\frac{u}{n^{4/3}},\frac{v}{n^{4/3}}\right)
=\frac{v^{\nu/2}}{u^{\nu/2}}\Kcr\left(u,v;\tfrac12
(a^2-5b),2a,\nu\right),
\end{equation}
uniformly for $u,v$ in compact
subsets of $\R^+$.
\end{theorem}
This theorem will be proved in Section \ref{sec: proof of triple
scaling limit}. The prefactor $v^{\nu/2}/u^{\nu/2}$ has no influence
on the singular value correlations. The phase diagram in Figure
\ref{fig: phase diagram} suggests that by taking proper scaling
limits, this critical kernel converges to the inhomogeneous
Painlev\'{e} II kernel or to the kernel in \cite{KMFW2}. The
non-chiral version of this statement has been shown in \cite{DG,GZ}.

\subsection{Organization of the paper}

The rest of this paper is organized as follows. In
Section~\ref{section:RHP} we prove Theorem~\ref{theorem:MOP} on the
multiple orthogonality relations for the biorthogonal polynomials
$P_{j,n}(x)$ and $Q_{k,n}(y)$. In
Section~\ref{section:thirdorderdiff} we obtain and analyze a
third-order differential equation related to the chiral two-matrix
model in the case of quadratic $W(y)=y^2/2+\alpha y$. In
Section~\ref{section:steepestdescent:quadratic} we perform a
steepest descent analysis on RH problem \ref{rhp:Y} for quadratic
$W$ leading to the proof of Theorem \ref{th: noncrit limit}. In
Section~\ref{section:triple:scaling} we adapt the steepest descent
analysis of the preceding section to prove Theorem~\ref{th: triple
scaling limit}. The main difference is in the local parametrix at
the origin which is now built using the solution to RH problem
\ref{rhp:modelM}.

\section{Proof of Theorem~\ref{theorem:MOP}}
\label{section:RHP}

\subsection{Preliminary lemma}
\label{subsection:MOP1}

\begin{lemma}
The function $f_n(x)$ in \eqref{def:f} satisfies the differential
equation
\begin{equation}\label{f:diffeq}
xf_n''(x)-(\nu-1)f_n'(x) = (\tau n)^2  f_n(x).
\end{equation}
\end{lemma}

\begin{proof}
This is immediate from the definition \eqref{def:f} and the fact that $I_\nu(z)$ satisfies the modified Bessel equation \cite{DLMF}
\begin{equation} \label{eq: modified Bessel equation}
z^2I_\nu''(z)+zI_\nu'(z)-(z^2+\nu^2)I_\nu(z)=0. \qedhere
\end{equation}
\end{proof}

\subsection{Proof of Theorem~\ref{theorem:MOP}(a)}
We claim that for any polynomial $P$ and any nonnegative integer $l$, one has
\begin{multline} \label{eq: claim}
\int_0^\infty  x^k P(x) e^{-nV(x)}\int_0^\infty y^l f_n(xy)e^{-nW(y)}\ud y \ud
x  = \int_0^\infty \int_0^\infty P(x) \pi_{k,l}(y) w_n(x,y) \ud x \ud y,
\end{multline}
where $\pi_{k,l}$ is a polynomial of exact degree $k(2r+1)+l$.

We prove the claim by induction on $k$. The case $k=0$ is trivial.
Let us assume that the statement is valid up to $k$ and compute the
left-hand side of \eqref{eq: claim} for $k+1$
\begin{align*}
& \int_0^\infty x^{k+1} P(x) e^{-nV(x)}\int_0^\infty y^l f_n(xy)e^{-nW(y)} \ud y \ud x \\
& = \frac{1}{(\tau n)^2} \int_0^\infty  x^k P(x) e^{-nV(x)}  \int_0^\infty y^{l+1} x^2 f_n''(xy) e^{-nW(y)} \ud y \ud x \\
& \qquad \qquad-\frac{\nu-1}{(\tau n)^2} \int_0^\infty  x^k P(x) e^{-nV(x)}  \int_0^\infty y^l x f_n'(xy) e^{-nW(y)} \ud y \ud x\\
&  = \frac{1}{(\tau n)^2} \int_0^\infty  x^k P(x) e^{-nV(x)} \int_0^\infty \tilde \pi_{l+1+2r}(y) f_n(xy) e^{-nW(y)} \ud y \ud x \\
& \qquad \qquad +\frac{\nu-1}{(\tau n)^2}  \int_0^\infty  x^k P(x)
e^{-nV(x)} \int_0^\infty  \tilde \pi_{l+r}(y)  f_n(xy) e^{-nW(y)}
\ud y \ud x,
\end{align*}
where $\tilde \pi_i$ denotes a polynomial of degree $i$. We used \eqref{f:diffeq} for the first equality. The second
equality is based on integration by parts. Note that the integrated terms vanish for $\nu > 0$ and cancel each other out for $-1<\nu \leq 0$. Now
expanding
\[
\frac{1}{(\tau n)^2} \tilde \pi_{l+1+2r}(y)+\frac{\nu-1}{(\tau n)^2}\tilde
\pi_{l+r}(y)= \sum_{i=0}^{l+1+2r}a_iy^i,
\]
and applying the induction hypothesis yields
\begin{align*}
& \int_0^\infty x^{k+1} P(x) e^{-nV(x)}\int_0^\infty y^l f_n(xy)e^{-nW(y)} \ud y \ud x    \\
& = \sum_{i=0}^{l+1+2r} a_i \int_0^\infty x^k P(x) e^{-nV(x)} \int_0^\infty y^i f_n(xy)e^{-nW(y)} \ud y \ud x  \\
& = \sum_{i=0}^{l+1+2r} a_i \int_0^\infty \int_0^\infty P(x) \pi_{k,i}(y) w_n(x,y) \ud y \ud x \\
& = \int_0^\infty \int_0^\infty P(x) \pi_{k+1,l}(y) w_n(x,y) \ud y \ud
x.
\end{align*}
Here $\pi_{k+1,l}$ is defined as
\[
\sum_{i=0}^{l+1+2r} a_i \pi_{k,i},
\]
and is, therefore, of degree $(k+1)(2r+1)+l$ as it should be. This
proves the claim.

To show \eqref{MOP:1}, note that the set $A=\{\pi_{k,l} \mid k=0,1,2,3,\ldots,
l=0, \ldots, 2r \}$ spans the space of polynomials. If we take
$P(x)=P_{j,n}(x)$ then the right-hand side of \eqref{eq: claim} vanishes for
any $\pi_{k,l} \in A$ satisfying $k(2r+1)+l<j$.  This leads to vanishing of the
left-hand side of \eqref{eq: claim} for every $k=0, \ldots,
\left\lfloor\frac{j-l-1}{2r+1}\right\rfloor$, $l=0, \ldots 2r$. Hence,
$P_{j,n}$ is the $j$-th multiple orthogonal polynomial satisfying
\eqref{MOP:1}.  This proves Theorem~\ref{theorem:MOP}(a).

\subsection{Proof of Theorem~\ref{theorem:MOP}(b)}
\label{subsection:MOP2}

To prove \eqref{MOP:2}, it is sufficient to establish linear
relations between the functions $h_{j,n}$ and their derivatives.

\begin{lemma}\label{lemma:wdiff} The function $xh_{l,n}'(x)$, see \eqref{hl:def} for $h_{l,n}$, can be written as a linear combination of the functions $h_{l+j,n}(x)$, $j=0,\ldots,r+1$,
\begin{equation}\label{w:diff}
x h_{l,n}'(x) = -(l+1)h_{l,n}(x)+n\sum_{j=1}^{r+1}c_{j}h_{l+j,n}(x),
\end{equation}
for any $l\in\mathbb N\cup\{0\}$, where we write the polynomial $W$
in \eqref{Bessel:weight2} as
\begin{equation}\label{V2:sum}
W(y) = \sum_{j=1}^{r+1}\frac{c_j}{j}y^j,\end{equation} with
coefficients $c_k\in\R$, $c_{r+1}>0$. Here we assume without loss
of generality that the constant term of $W$ vanishes.
\end{lemma}

\begin{proof}
A straightforward calculation gives
\begin{align*} xh_{l,n}'(x) &:= x\int_0^{\infty}
y^l\frac{\partial }{\partial
x}(f_n(xy))e^{-nW(y)} \ud y\\
&= \int_0^{\infty}y^{l+1} \frac{\partial}{\partial
y}\left(f_n(xy)\right) e^{-nW(y)} \ud y\\
&= -\int_0^{\infty} f_n(xy)\frac{\partial}{\partial
y}\left(y^{l+1} e^{-nW(y)}\right) \ud y\\
&= -(l+1)h_{l,n}(x)+n\sum_{j=1}^{r+1}c_{j}h_{l+j,n}.
\end{align*}
Here the first step follows from the definition \eqref{hl:def}, the second step
is obvious by symmetry, the third step is integration by parts, and the last step uses \eqref{V2:sum}.
\end{proof}

The above lemma then implies the alternative system of weight functions for the
multiple orthogonal polynomials $P_{j,n}(x)$ in \eqref{MOP:2}. This proves
Theorem~\ref{theorem:MOP}(b).

\section{Differential equation for $h_{0,n}(x)$ with quadratic~$W$}
\label{section:thirdorderdiff}

From here we assume again that $W$ is quadratic and depends on the real parameter $\alpha$ as in \eqref{eq:quadratic potential W}. Then $r=\deg W-1 = 1$.

In this section we obtain a third-order differential equation for the function
\begin{equation}\label{w0tilde:def}
h_{0,n}(x) = \int_0^{\infty}f_n(xy)e^{-nW(y)} \ud y,
\end{equation}
see \eqref{hl:def} and \eqref{def:f}. We will study this differential equation in detail and construct a Wronskian matrix from three special solutions to the equation. We will show that this Wronskian matrix satisfies a certain $3\times 3$ RH problem, which should be regarded as a chiral analogue of the $3\times 3$ RH problem associated with the Pearcey differential equation used in \cite{Duits2,DKM}.

\subsection{Third-order differential equation}

In this section we prove the following result.

\begin{proposition}\label{proposition:thirdorderODE} (Third order differential equation)
Let $W$ be quadratic as in \eqref{eq:quadratic potential W}. Then the function
$p(x):=h_{0,n}(x)$ in \eqref{w0tilde:def} satisfies the third order differential
equation
\begin{equation}\label{thirdorderdiffeq:wtilde}
x^2 p'''(x)+(2-2\nu)xp''(x) +(\alpha n^2\tau^2 x+\nu^2 -\nu)p'(x)
-(\tau^4n^3x+\tau^2n^2\nu \alpha)p(x)=0.
\end{equation}
\end{proposition}

We note that a similar differential equation appears in a Mehler-Heine type formula in \cite{Deschout1}.
To prove Proposition~\ref{proposition:thirdorderODE}, we start with two lemmas.

\begin{lemma} The function $h_{l+1,n}$ in \eqref{hl:def} can be expressed in terms of $h_{l,n}$ as
\begin{equation}\label{Vquad:4}
(\tau n)^2 h_{l+1,n}(x) = x h_{l,n}''(x)-(\nu-1) h_{l,n}'(x),
\end{equation}
for any $l\in\mathbb N\cup\{0\}$.
\end{lemma}

\begin{proof} This is obvious from \eqref{f:diffeq} and \eqref{hl:def}.
\end{proof}

\begin{lemma} Let $W$ be quadratic as in \eqref{eq:quadratic potential W}.
Then the functions $h_{l,n}$, $l=0,1,2,3$, satisfy the relation
\begin{equation}\label{Vquad:rel3} h_{3,n}(x)+2\alpha h_{2,n}(x)+(\alpha^2-(\nu+2) n^{-1})h_{1,n}(x) -
((\nu+1)n^{-1}\alpha + \tau^2x)h_{0,n}(x)=0.
\end{equation}
\end{lemma}

\begin{proof}
We first assume that $\nu>0$ and calculate
\begin{equation}\label{tricky:integration:by:parts:1}
\begin{aligned}
h_{3,n}(x)+2\alpha h_{2,n}(x)+(\alpha^2-n^{-1})h_{1,n}(x) & = \int_0^{\infty} y\left((y+\alpha)^2-n^{-1}\right)f_n(xy)e^{-nW(y)} \ud y\\
& = n^{-2}\int_0^{\infty} y f_n(xy)\frac{\partial^2}{\partial y^2}\left(e^{-nW(y)}\right) \ud y\\
& = n^{-2}\int_0^{\infty} x\left( xyf_n''(xy)+2f_n'(xy)\right)e^{-nW(y)} \ud y,
\end{aligned}
\end{equation}
where the first step follows from the definition \eqref{hl:def}, the second
step uses \eqref{eq:quadratic potential W}, and the third step is integration by parts.
Similarly,
\begin{equation}\label{tricky:integration:by:parts:2}
\begin{aligned}
h_{1,n}(x)+\alpha h_{0,n}(x) &= \int_0^{\infty} (y+\alpha)f_n(xy)e^{-nW(y)} \ud y\\
                     &= -n^{-1}\int_0^{\infty} f_n(xy)\frac{\partial}{\partial y}\left(e^{-nW(y)}\right) \ud y\\
                     &= n^{-1}\int_0^{\infty} xf'(xy)e^{-nW(y)} \ud y.
\end{aligned}
\end{equation}
The lemma then follows from \eqref{f:diffeq} by a little calculation. Finally, note that the assumption $\nu> 0$ was needed to make sure that the integration by parts in \eqref{tricky:integration:by:parts:1}--\eqref{tricky:integration:by:parts:2}  creates no integrated terms. However, if $\nu\leq 0$, integrated terms are present, but cancel each other out so that \eqref{Vquad:rel3} remains valid.
\end{proof}

Now we are ready to prove Proposition~\ref{proposition:thirdorderODE}.

\begin{proof}[Proof of Proposition~\ref{proposition:thirdorderODE}]
First observe that \eqref{w:diff} and \eqref{eq:quadratic potential W} yield the
relations
\begin{align}
\label{Vquad:rel1} xh_{0,n}'(x) = -h_{0,n}(x)+n\alpha h_{1,n}(x)+n h_{2,n}(x),\\
\label{Vquad:rel2} x h_{1,n}'(x) = -2 h_{1,n}(x)+n\alpha h_{2,n}(x)+n h_{3,n}(x),
\end{align}
and so on. Multiplying \eqref{Vquad:rel1} with $\alpha n^{-1}$ and
\eqref{Vquad:rel2} with $n^{-1}$, and adding them up, we get
$$ h_{3,n}(x)+2\alpha h_{2,n}(x)+\alpha^2 h_{1,n}(x) = n^{-1}\left(\alpha x h_{0,n}'(x)+\alpha h_{0,n}(x)+x h_{1,n}'(x)+2 h_{1,n}(x)\right).
$$
Substituting this into \eqref{Vquad:rel3}, we find
\begin{equation}\label{diffeq:proof1} x h_{1,n}'(x)-\nu h_{1,n}(x)+\alpha x h_{0,n}'(x) -(\nu \alpha+ \tau^2nx)h_{0,n}(x)=0.
\end{equation}
On account of \eqref{Vquad:4} with $l=0$ this yields
\begin{align*}
(\tau n)^2h_{1,n}(x) = x h_{0,n}''(x)-(\nu-1)h_{0,n}'(x),\\
(\tau n)^2 h_{1,n}'(x) = x h_{0,n}'''(x)-(\nu-2)h_{0,n}''(x).
\end{align*}
Inserting this in \eqref{diffeq:proof1} we get
\begin{multline*}
x(\tau n)^{-2}(x h_{0,n}'''(x)-(\nu-2)h_{0,n}''(x))-\nu (\tau n)^{-2}(x h_{0,n}''(x)
-(\nu-1)h_{0,n}'(x))\\ +\alpha x h_{0,n}'(x) - (\nu \alpha + \tau^2n x)h_{0,n}(x)=0,
\end{multline*}
which with the notation $p(x):=h_{0,n}(x)$ is equivalent to
\eqref{thirdorderdiffeq:wtilde}. This proves
Proposition~\ref{proposition:thirdorderODE}.
\end{proof}

\begin{remark}
A little calculation shows that for any solution $p(x)$ to \eqref{thirdorderdiffeq:wtilde} the transformed function
$$q(x):=x^{(1-4\nu)/3}p(x^2)$$ satisfies the differential equation
\begin{multline*}
x^3 q'''(x) +\left(4\tau^2n^2\alpha x^3-\frac{(2\nu+1)^2}{3}x\right)q'(x) \\
+\left(\frac{16 \nu^3+60\nu^2+12\nu-7}{27} -8\tau^4n^3 x^4-\frac 43   \tau^2n^2\alpha x^2(2\nu+1)\right)q(x) = 0.
\end{multline*}
For the special case $\nu=-1/2$ this reduces to the third order differential equation
\begin{equation*} 
q'''(x)+4\tau^2n^2\alpha q'(x) = 8\tau^4n^3x q(x),
\end{equation*}
which appeared before in \cite[eqn. (5.3)]{DKM}. In that paper this ODE is solved using Pearcey integrals.
\end{remark}

\subsection{Three special solutions to \eqref{thirdorderdiffeq:wtilde}}

We now construct three contour integral solutions $p_0$, $p_1$, and $p_2$ to the third-order differential equation \eqref{thirdorderdiffeq:wtilde}. In the next section we will study the asymptotics of these functions as $z\to\infty$.

The function $p_0$ is defined as
\begin{equation}\label{def:p0}
p_0(z):= h_{0,n}(z) = \int_0^{\infty} f_n(zy)e^{-nW(y)} \ud y,
\end{equation}
where we recall the definition of $f_n$ \eqref{def:f}. Then
$p_0$ is analytic in $\cee\setminus\er^-$.

The second solution $p_1$ to \eqref{thirdorderdiffeq:wtilde} is defined as
\begin{equation}\label{def:p1} 
p_1(z):= \frac{i}{\pi}\int_{-\infty}^{\infty}
(zy)^{\nu/2}K_\nu(2\tau n\sqrt{zy})e^{-nW(y)} \ud y,
\end{equation}
where
\begin{equation}\label{def:Knu}
K_\nu(z)=\frac{1}{2}\pi\frac{I_{-\nu}(z)-I_\nu(z)}{\sin(\nu\pi)},
\qquad z\in\mathbb{C}\setminus(-\infty,0]
\end{equation}
is the modified Bessel function of the second kind. The right hand
side of \eqref{def:Knu} is replaced by its limiting value if $\nu$
is an integer or zero. Then $p_1$ is analytic in $\cee\setminus\er$.

Finally, the third solution to \eqref{thirdorderdiffeq:wtilde} is defined as
\begin{equation}\label{def:p2}
p_2(z) := \int_{-\infty}^{0} f_n(zy)e^{-nW(y)}\ \ud y.
\end{equation}
This function is analytic in $\cee\setminus\er^+$.

The functions $p_0,p_1,p_2$ all satisfy the differential equation
\eqref{thirdorderdiffeq:wtilde}. For $p_0$ this follows directly
from Proposition~\ref{proposition:thirdorderODE}. The same proof
works for $p_2$. Finally, for $p_1$ we first observe that
$(x)^{\nu/2}K_{\nu}(2\tau n\sqrt{x})$ satisfies the same
differential equation \eqref{f:diffeq} as $f_n(x)$ since also
$K_\nu$ solves the modified Bessel equation \eqref{eq: modified
Bessel equation}. This then allows the proof of
Proposition~\ref{proposition:thirdorderODE} to be applied word for
word.

\begin{lemma}\label{lemma:jumpsp123}
The functions $p_0,p_1,p_2$ satisfy the jump relations
\begin{align*}
p_{0,+}(x) &= e^{2\nu \pi i}p_{0,-}(x),\qquad x\in\er^-,\\
p_{1,+}(x) &= p_{1,-}(x)- e^{-\nu\pi i}p_{2,-}(x),\qquad x\in\er^+,\\
p_{1,+}(x) &= p_{1,-}(x)+ e^{\nu\pi i}p_{0,-}(x),\qquad x\in\er^-,\\
p_{2,+}(x) &= e^{-2\nu \pi i}p_{2,-}(x),\qquad x\in\er^+,
\end{align*}
where the subscripts $+$ and $-$ denote the boundary values obtained from the
upper or lower half plane of $\cee$, respectively.
\end{lemma}

\begin{proof}
This follows from the definitions \eqref{def:p0}, \eqref{def:p1}, and \eqref{def:p2}; the connection formula \eqref{def:Knu}; and the fact that ${I_\nu}_+(x)={I_\nu}_-(x)e^{2\nu \pi i}$ for $x<0$ which is immediate from \eqref{def:Inu}.
\end{proof}

\subsection{Asymptotics for the solutions to \eqref{thirdorderdiffeq:wtilde}}

Our next goal is to find the large $z$ asymptotics of $p_0$, $p_1$, and $p_2$.
In analogy with \cite{DKM} we first introduce three functions $\theta_j(z)$,
$j=1,2,3$, that will appear in these asymptotics.

First we define the constants $x^*(\alpha)$ and $y^*(\alpha)$ that
depend on $\alpha$ and $\tau$.  We put
\begin{equation} \label{eq: x y star}
x^*(\alpha)=\left\{
         \begin{array}{ll}
           0, & \hbox{$\alpha\geq 0$,} \\
           -\frac{4}{\tau^2}(\frac{\alpha}{3})^3, & \hbox{$\alpha<0$},
         \end{array}
       \right. \qquad
y^*(\alpha)=-x^*(-\alpha).
\end{equation}
So if $\alpha>0$ we have that $x^*(\alpha)=0$ and
$y^*(\alpha)<0$ whereas in the case $\alpha\leq 0$ it holds that
$x^*(\alpha) \geq 0$ and $y^*(\alpha)=0$.

\begin{lemma}\label{lem:theta function}
For every $\alpha \in \R$ and $\tau>0$ there exist analytic functions
$\theta_j: \C \setminus \R \to \C$, $j=1,2,3$, with the following properties.
\begin{itemize}
\item[(a)] The jumps of the functions $\theta_j$ are taken together in terms of the jumps for the diagonal matrix
\begin{equation} \label{Theta}
\Theta(z)= \diag\left(\theta_1(z),\theta_2(z),\theta_3(z)\right), \qquad z \in
\C \setminus \R.
\end{equation}
We have
\[ \begin{cases}
\Theta_+(x)= \begin{pmatrix}1&0&0 \\0&0&1 \\ 0&1&0 \end{pmatrix}
\Theta_-(x) \begin{pmatrix}1&0&0\\0&0&1 \\ 0&1&0 \end{pmatrix}, &x>x^*(\alpha), \\
\Theta_+(x)=\Theta_-(x), & y^*(\alpha)<x<x^*(\alpha), \\
\Theta_+(x)= \begin{pmatrix}0&1&0 \\ 1&0&0 \\ 0&0&1 \end{pmatrix}
\Theta_-(x) \begin{pmatrix}0&1&0 \\ 1&0&0 \\ 0&0&1 \end{pmatrix}, &x<y^*(\alpha), \\
\end{cases} \]
with $x^*(\alpha)$ and $y^*(\alpha)$ as in \eqref{eq: x y star}.
\item[(b)] We have as $z \to \infty$ within $\C^+$
\begin{multline*}
\theta_j(z)=\frac32 \omega^{j-1} \tau^{4/3}z^{2/3}-\alpha \omega^{4-j} \tau^{2/3}z^{1/3}+ \frac{\alpha^2}{3}\\
-\frac{\alpha^3}{27}\omega^{j-1} \tau^{-2/3}z^{-1/3} + D\omega^{4-j} \tau^{-4/3}z^{-2/3} +\mathcal O \left( z^{-1} \right),
\end{multline*}
for a constant $D \in \R$ and $j=1,2,3$. Here $\om:=\exp(2\pi i/3)$. The behavior in $\C^-$ follows from the relation $\theta_j(\overline z)= \overline{\theta_j(z)}$, $j=1,2,3$.
\end{itemize}
\end{lemma}
\begin{proof}
Define
\begin{equation}\label{eq:theta function}
\theta_j(z)=2 \theta_j^{\DKM}(\sqrt z), \qquad j=1,2,3,
\end{equation}
where the $\theta_j^{\DKM}$ refer to the $\theta$-functions
used in \cite{DKM} with potentials given in \eqref{eq: relation
chiral/hermitian potentials}. From (2.5) and (2.8) in \cite{DKM}, it
is readily seen that
\begin{equation}\label{eq:even of theta DKM}
\theta_j^{\DKM}(-z)=\theta_j^{\DKM}(z).
\end{equation}
This, together with \cite[Corollary 2.2 and Lemma 2.4]{DKM}, implies our lemma.
\end{proof}

We are now ready to investigate the large $z$ asymptotics of $p_0$,
$p_1$, and $p_2$. Let us first give a heuristic argument.

\begin{lemma}\label{lemma:asygeneric}
For any solution $p(z)$ to \eqref{thirdorderdiffeq:wtilde}, there exist $j\in\{0,1,2\}$ and $C\in\cee\setminus\{0\}$ such that
\begin{equation}\label{asy:pz:generic}
p(z)=Cz^{\frac{2\nu-1}{3}}e^{n\theta_{j+1}(z)}\left(1-\frac{\alpha\nu}{3\tau^{2/3}n}
\omega^jz^{-1/3}+\tilde D \omega^{2j}z^{-2/3}+\mathcal
O(z^{-1})\right),
\end{equation}
as $z\to\infty$ in closed sectors of the first quadrant
$\{z\in\cee\mid\Re>0,~ \Im z>0\}$. Here $\tilde D \in \R$ is a
constant, and $\theta_j$ is defined in Lemma \ref{lem:theta
function}.
\end{lemma}

\begin{proof}
Substituting $p(z)=e^{nF(z)}$ in \eqref{thirdorderdiffeq:wtilde} we get the
following nonlinear ODE for $f(z):=F'(z)$
\begin{multline*}
\left(z^2f^3(z)-\tau^4z+\alpha \tau^2zf(z)\right)+n^{-1}\left(3z^2f(z)f'(z)-2(\nu-1)zf^2(z)-\nu \tau^2 \alpha\right)\\
+n^{-2}\left(z^2f''(z)-2(\nu-1)zf'(z)+\nu(\nu-1)f(z)\right)=0.
\end{multline*}
This ODE has solutions $f$ with expansions
\begin{multline*}
\tau^{-2} f(z/\tau^2)=\om^jz^{-1/3}-\frac
\alpha3\om^{2j}z^{-2/3}+\frac{2\nu-1}{3n}z^{-1}
+\frac{\alpha(9\nu+n\alpha^2)}{81n}\om^j
z^{-4/3}\\+\frac{\omega^{2j}}{243n^2}\left(27\nu(\nu+1)-9+9(\nu+1)\alpha^2n+\alpha^4n^2
\right)z^{-5/3}+\mathcal O(z^{-2}),
\end{multline*}
as $z\to \infty$, for any $j\in\{0,1,2\}$. If $\Im z >0$, after integration this yields
\begin{equation*}
\begin{aligned}
F(z)&=\theta_{j+1}(z)+\frac{2\nu-1}{3n}\log z+\frac{1}{n}\log
C-\frac{\alpha\nu}{3\tau^{2/3}n}\om^j
z^{-1/3}+D'\omega^{2j}z^{-2/3}+\mathcal O(z^{-1}),
\end{aligned}
\end{equation*}
where $C$ is the integration constant and $D' \in \R$. Then $p(z)=e^{nF(z)}$ satisfies \eqref{asy:pz:generic}.
\end{proof}

Next we specialize Lemma~\ref{lemma:asygeneric} to the functions $p_0$, $p_1$, and $p_2$.

\begin{lemma}\label{lemma:asyp123} For each $j=0,1,2$, we have
\begin{equation*}
p_{j}(z)= \frac{C_j}{\sqrt 3 n}\tau^{\frac{\nu-2}{3}}e^{\alpha^2
n/3} \omega^{2j} z^{\frac{2\nu-1}{3}}
e^{n\theta_{j+1}(z)}\left(1-\frac{\alpha\nu}{3\tau^{2/3}n}\omega^{j}z^{-1/3}+\tilde
D \omega^{2j}z^{-2/3}+\mathcal O(z^{-1})\right),
\end{equation*}
as $z\to\infty$ in closed sectors of the upper half plane $\{z\in\cee\mid\Im
z>0\}$, where
\begin{equation}
C_0=1,\quad C_1=-ie^{-\frac{\pi i}{6}(2\nu+1)},\quad C_2=e^{-\frac{\pi
i}{3}(2\nu+1)},
\end{equation}
and $\tilde D$ is the constant in \eqref{asy:pz:generic}. In the lower half
plane we have the same expansions but with $\omega$ replaced by $\omega^{-1}$
and $C_0,C_1,C_2$ by $\overline{C_0},-\overline{C_1},\overline{C_2}$
respectively.
\end{lemma}

\begin{proof}
First assume $\tau=n=1$. Consider the ODE
\begin{equation} \label{eq: ODE KMW}
zq'''(z)+\nu q''(z)+\alpha q'(z)-q(z)=0.
\end{equation}
As observed in \cite{Deschout1}, for any solution $q(z)$ of \eqref{eq: ODE KMW}
\[
p(z)=z^\nu q''(z)
\]
solves \eqref{thirdorderdiffeq:wtilde}. In \cite{KMFW2} solutions to \eqref{eq: ODE KMW} are studied.
The authors consider four solutions of the form
\begin{equation} \label{eq: def q}
q_j(z)=\int_{\Gamma_j}t^{\nu-3}e^{\frac{1}{2t^2}-\frac{\alpha}{t}+zt}\ud t,
\qquad j=1,2,3,4,
\end{equation}
where the contours $\Gamma_j$ are shown in Figure \ref{fig: Gammaj}.
\begin{figure}
\begin{center}
\begin{tikzpicture}[scale=2]
\begin{scope}[decoration={markings,mark= at position 0.99 with {\arrow{stealth}}}]
\draw[very thick,postaction=decorate](0,0).. controls (0,-1) and (1,-0.8)..(1,0);
\draw[very thick](0,0).. controls (0,1) and (1,0.8)..(1,0) node[right]{$\Gamma_1$};
\draw[very thick,postaction=decorate](0,0).. controls (0,-1) and (-1,-0.8)..(-1,0);
\draw[very thick](0,0).. controls (0,1) and (-1,0.8)..(-1,0) node[left]{$\Gamma_2$};
\end{scope}
\begin{scope}[decoration={markings,mark= at position 0.3 with {\arrow{stealth}}}]
\draw[very thick,postaction=decorate](-0.3,1.2)..node[near start, right]{$\Gamma_3$} controls (-0.1,1) and (0,0.6)..(0,0) ;
\draw[very thick,postaction=decorate](-0.3,-1.2)..node[near start, right]{$\Gamma_4$} controls (-0.1,-1) and (0,-0.6)..(0,0) ;
\end{scope}
\filldraw (0,0) circle (.2mm) node[below right]{0};
\end{tikzpicture}
\end{center}
\caption{The contours of integration $\Gamma_j$, $j=1,2,3,4$, used in the definitions \eqref{eq: def q} of the functions $q_j$.}
\label{fig: Gammaj}
\end{figure}
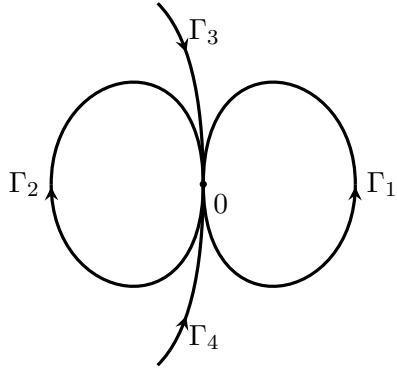

The branch cuts for $t^{\nu-3}$ are specified as follows: for $j=1$
we take $-\pi/2<\arg t < \pi/2$, for $j=2$ we have $\pi/2<\arg t <
3\pi/2$, for $j=3$ we choose $0<\arg t < \pi$, and for $j=4$ we put
$-\pi<\arg t <0$. Note that our definition differs from the one in
\cite{KMFW2} by a multiplicative constant.

We have the following relations
\begin{align}
q_2(z)&=q_3(z)-e^{2\pi \nu i}q_4(z), & \text{for }\Re z > 0,  \label{eq: connection formula 1}\\
q_1(z) &= q_3(z)-q_4(z), & \text{for }\Re z < 0. \label{eq:
connection formula 2}
\end{align}
It is then clear that the functions
\[
z^\nu q''_j(z)=z^\nu
\int_{\Gamma_j}t^{\nu-1}e^{\frac{1}{2t^2}-\frac{\alpha}{t}+zt}\ud t, \qquad
j=1,2,3,4,
\]
generate the three dimensional solution space of
\eqref{thirdorderdiffeq:wtilde}.
Moreover \cite{KMFW2} provides detailed
asymptotics of these functions, e.g. as $z \to \infty$ such that $0<\arg z <
\pi/4$ we have
\begin{align}
z^{\nu}q_1''(z) &= \sqrt{2\pi}e^{-\alpha^2/6}
\tfrac{i}{\sqrt 3}z^{\frac{2\nu-1}{3}}e^{\frac32z^{2/3}-\alpha z^{1/3}}\left(1+\mathcal O(z^{-1/3})\right), \\
z^{\nu}q_3''(z) &= \sqrt{2\pi}e^{-\alpha^2/6}e^{\frac{2\nu\pi i}{3}}\tfrac{i}{\sqrt 3}\omega z^{\frac{2\nu-1}{3}}e^{\frac32\omega z^{2/3}-\alpha\omega^2 z^{1/3}}\left(1+\mathcal O(z^{-1/3})\right), \\
z^{\nu}q_4''(z) &= -\sqrt{2\pi}e^{-\alpha^2/6}e^{-\frac{2\nu\pi
i}{3}}\tfrac{i}{\sqrt 3}\omega^2
z^{\frac{2\nu-1}{3}}e^{\frac32\omega^2 z^{2/3}-\alpha\omega
z^{1/3}}\left(1+\mathcal O(z^{-1/3})\right).
\end{align}

The asymptotic behavior of $z^{\nu}q_2''(z)$ follows using the
connection formula \eqref{eq: connection formula 1}. The idea is now
to express $p_0$, $p_1$, and $p_2$ as linear combinations of the
functions $z^{\nu}q_j''(z)$ to obtain the asymptotic behavior. This
is done in Lemma \ref{lemma: relation pq} below. It leads to the
following results
\begin{align}
p_0(z) &= \tfrac{1}{\sqrt 3} e^{\alpha^2/3}z^{\frac{2\nu-1}{3}}e^{\frac32z^{2/3}-\alpha z^{1/3}}\left(1+\mathcal O(z^{-1/3})\right), \\
p_1(z) &= -\tfrac{i}{\sqrt 3} e^{\alpha^2/3}e^{-\frac{\pi i}{6}(2\nu+1)}\omega^2 z^{\frac{2\nu-1}{3}}e^{\frac32\omega z^{2/3}-\alpha\omega^2 z^{1/3}}\left(1+\mathcal O(z^{-1/3})\right), \\
p_2(z) &= \tfrac{1}{\sqrt 3} e^{\alpha^2/3}e^{-\frac{\pi
i}{3}(2\nu+1)}\omega z^{\frac{2\nu-1}{3}}e^{\frac32\omega^2
z^{2/3}-\alpha\omega z^{1/3}}\left(1+\mathcal O(z^{-1/3})\right),
\end{align}
as $z \to \infty$ such that $0<
\arg z < \pi/4$. This proves the lemma for $\tau=n=1$ and in the
sector $0< \arg z < \pi/4$. The results in the remaining sectors are
proved similarly. The lemma for general values of $\tau$ and $n$
follows by the rescaling
\begin{equation} \label{eq: rescaling}
p_{j}^{\alpha,n,\tau,\nu}(z)=\frac{1}{(\tau n)^\nu \sqrt
n}p_j^{\sqrt n\alpha,n=\tau=1,\nu}\left(n^\frac32\tau^2z \right),
\qquad j=0,1,2.
\end{equation}
This equality follows from the definitions \eqref{def:p0}--\eqref{def:p2} changing variables in the integrals. Note that this rescaling is consistent with the ODE \eqref{thirdorderdiffeq:wtilde}.
\end{proof}

\begin{lemma} \label{lemma: relation pq}
Assuming $\tau=n=1$, the following relations between $p_j$, $j=0,1,2$, and $q_j$, $j=1,2,3$, hold
\begin{align}
p_0(z)&=\tfrac{1}{\sqrt{2\pi}i}e^{\alpha^2/2} z^\nu q''_1(z), \qquad z \in \C \setminus \R^-, \label{eq: relation p0q1} \\
p_1(z)&=\begin{cases} \frac{1}{\sqrt{2\pi}i}e^{-\nu \pi i}e^{\alpha^2/2} z^\nu q_3''(z), & \Im z>0,  \\
                      \frac{1}{\sqrt{2\pi}i}e^{\nu \pi i}e^{\alpha^2/2} z^\nu q_4''(z),  & \Im z<0,
        \end{cases} \label{eq: relation p1q3} \\
p_2(z)&=-\tfrac{1}{\sqrt{2\pi}i}e^{-\nu\pi i}e^{\alpha^2/2}(-z)^\nu q_2''(z)
\qquad z \in \C \setminus \R^+. \label{eq: relation p2q2}
\end{align}
\end{lemma}
\begin{proof}
We will prove \eqref{eq: relation p0q1}--\eqref{eq: relation p2q2} for $\nu\in(-1,\infty)\setminus \mathbb{Z}$. When $\nu\in\mathbb{Z}$ the equalities still hold true by continuity in $\nu$.

We start with the proof of \eqref{eq: relation p0q1} using the convergent
series representations
\begin{align*}
p_0(z) &=z^\nu \sum_{k=0}^\infty\frac{z^k}{k!\Gamma(k+\nu+1)}\int_0^\infty y^{k+\nu} e^{-W(y)} \ud y, \\
z^\nu q_1''(z) &= z^\nu \sum_{k=0}^\infty\frac{z^k}{k!}\int_{\Gamma_1}
t^{\nu-1+k}e^{\frac{1}{2t^2}-\frac{\alpha}{t}}\ud t.
\end{align*}
Recall the contour integral formula for the reciprocal of the Gamma function
\cite{DLMF}
\begin{equation} \label{eq: reciprocal Gamma}
\frac{1}{\Gamma(z)}=\frac{1}{2\pi i} \int e^t t^{-z} \ud t,
\end{equation}
where the integration is over a Hankel contour as shown in Figure \ref{fig: Hankel}. The Cauchy theorem allows us to deform this contour of integration into a contour coming from infinity under the angle $-\phi$, avoiding the negative real line, and tending to infinity under the angle $\phi$, for $\pi/2<\phi < 3\pi/4$, see Figure \ref{fig: Hankel}.
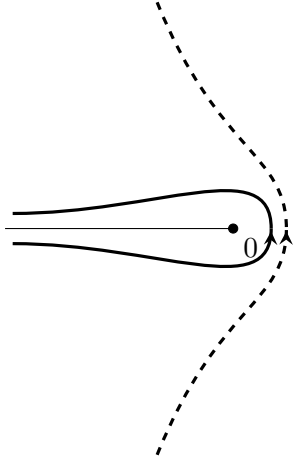
\begin{figure}
\begin{center}
\begin{tikzpicture}
\draw (-3,0)--(0,0);
\begin{scope}[decoration={markings,mark= at position 0.99 with {\arrow{stealth}}}]
\draw[very thick,postaction=decorate](-2.9,-0.2).. controls (-1,-0.2) and (0.5,-1)..(0.5,0);
\draw[very thick](-2.9,0.2).. controls (-1,0.2) and (0.5,1)..(0.5,0);
\draw[dashed,very thick,postaction=decorate](-1,-3).. controls (-0.2,-1) and (0.7,-1)..(0.7,0);
\draw[dashed,very thick](-1,3).. controls (-0.2,1) and (0.7,1)..(0.7,0);
\end{scope}
\filldraw (0,0) circle (.6mm) node[below right]{0};
\end{tikzpicture}
\end{center}
\caption{Hankel contour (solid) used for integration in \eqref{eq: reciprocal Gamma} and its deformation (dashed).}
\label{fig: Hankel}
\end{figure}
Then
\[
\frac{1}{\Gamma(\nu+k+1)}\int_0^\infty y^{\nu+k}e^{-\frac{y^2}{2}-\alpha y} \ud
y=\frac{1}{2\pi i} \int_0^\infty y^{\nu+k} \int t^{-\nu-1-k}
e^{-\frac{y^2}{2}-\alpha y+t} \ud t \ud y,
\]
where the inner integral is taken over the deformed Hankel contour.
In the inner integral (thus, for fixed $y>0$) we change variables $t=y/s$
\[
\frac{1}{\Gamma(\nu+k+1)}\int_0^\infty y^{\nu+k}e^{-\frac{y^2}{2}-\alpha y} \ud
y=\frac{1}{2\pi i} \int_0^\infty  \int_{\Gamma_1^*} s^{\nu-1+k}
e^{-\frac{y^2}{2}-\alpha y+\frac{y}{s}} \ud s \ud y.
\]
The new contour $\Gamma_1^*$ of integration emerges from zero under the angle $-\phi$ and ends in zero under the angle $\phi$, where $\pi/2<\phi < 3\pi/4$. It has counterclockwise orientation, see Figure \ref{fig: gamma1star}.
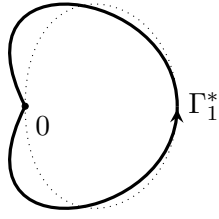
\begin{figure}
\begin{center}
\begin{tikzpicture}[scale=2]
\draw[dotted](0,0).. controls (0,-1) and (1,-0.8)..(1,0);
\draw[dotted](0,0).. controls (0,1) and (1,0.8)..(1,0);
\begin{scope}[decoration={markings,mark= at position 0.99 with {\arrow{stealth}}}]
\draw[very thick,postaction=decorate](0,0).. controls (-0.5,-1) and (1,-0.8)..(1,0);
\draw[very thick](0,0).. controls (-0.5,1) and (1,0.8)..(1,0) node[right]{$\Gamma_1^*$};
\end{scope}
\filldraw (0,0) circle (.2mm) node[below right]{0};
\end{tikzpicture}
\end{center}
\caption{Contour of integration $\Gamma_1^*$ (solid) used in the proof of Lemma \ref{lemma: relation pq} and the contour $\Gamma_1$ (dotted).}
\label{fig: gamma1star}
\end{figure}
Using the Fubini theorem we can change the order of integration to become
\begin{multline} \label{eq: proof asymp1}
\frac{1}{\Gamma(\nu+k+1)}\int_0^\infty y^{\nu+k}e^{-\frac{y^2}{2}-\alpha y} \ud y=\frac{1}{2\pi i}   \int_{\Gamma_1^*} s^{\nu-1+k} \int_{-\infty}^\infty e^{-\frac{y^2}{2}-\alpha y+\frac{y}{s}} \ud y \ud s \\
-\frac{1}{2\pi i}   \int_{\Gamma_1^*} s^{\nu-1+k} \int_{-\infty}^0
e^{-\frac{y^2}{2}-\alpha y+\frac{y}{s}} \ud y \ud s.
\end{multline}
We complete the square in the first term of the right-hand side of
\eqref{eq: proof asymp1} and evaluate the Gaussian integral. Then we
can deform the contour $\Gamma_1^*$ into $\Gamma_1$. Hence the first
term of the right-hand side of \eqref{eq: proof asymp1} equals
\[
\frac{1}{\sqrt{2\pi}i}e^{\frac{\alpha^2}{2}}\int_{\Gamma_1}
t^{\nu-1+k}e^{\frac{1}{2t^2}-\frac{\alpha}{t}}\ud t.
\]
The second term in the right-hand side of \eqref{eq: proof asymp1}
vanishes, since we can contract the contour of integration thanks to
Lemma \ref{lemma: contour contraction}. This completes the proof of
\eqref{eq: relation p0q1}.

Next we prove \eqref{eq: relation p2q2}. We claim that the following
chain of equalities holds
\begin{equation} \label{eq: chain of equalities}
p_2(z)=\tilde p_0(-z)=\tfrac{1}{\sqrt{2\pi}i}e^{\alpha^2/2}(-z)^\nu \tilde
q_1''(-z)=\tfrac{1}{\sqrt{2\pi}i}e^{\pi i (1-\nu)}e^{\alpha^2/2}(-z)^\nu
q_2''(z),
\end{equation}
for $z \in \C \setminus \R^+$. Here, a tilde means that the
parameter $\alpha$ on which the quantity depends has to be replaced
by $-\alpha$. Then \eqref{eq: relation p2q2} is immediate from the
claim. We now discuss the equalities in \eqref{eq: chain of
equalities}. The first one follows from \eqref{def:p0} and
\eqref{def:p2}. The second equality relies on \eqref{eq: relation
p0q1}. The last equality is a consequence of \eqref{eq: def q} for
$j=1,2$, where one has to take special care about the position of
the branch cut.

Finally we prove \eqref{eq: relation p1q3}. Using
\eqref{def:p0}--\eqref{def:p2} and \eqref{def:Inu} we find
\[
-2i\sin(\nu \pi) p_1(z)+p_0(z)+p_2(z)=\int_{-\infty}^\infty \sum_{k=0}^\infty
\frac{1}{k!\Gamma(k-\nu+1)}(zy)^k e^{-y^2/2-\alpha y} \ud y
\]
This function belongs to the one-parameter family of entire solutions of the ODE \eqref{thirdorderdiffeq:wtilde}. Therefore it is characterized by its value at zero
\[
\frac{\sqrt{2\pi}}{\Gamma(1-\nu)}e^{\alpha^2/2}.
\]
The idea is now to build an entire function as a linear combination of the $z^\nu q_j''(z)$, $j=1,2,3,4$ and find its value at zero to obtain the desired equality.

When $x>0$ we have
\begin{equation} \label{eq: analytic combination q 1}
x^\nu(q_4''(x)+q_1''(x)-q_3''(x))=\int
s^{\nu-1}e^{\frac{x^2}{2s^2}-\alpha\frac{x}{s}+s} \ud s,
\end{equation}
where we used \eqref{eq: def q} and made the change of variables
$s=xt$. The integral is taken over the dashed contour in Figure
\ref{fig: Hankel}. Clearly the right-hand side of this expression is
entire in $x$. Therefore the equality is valid in the right half
plane $\{z\in\mathbb{C} \mid \Re z >0 \}$. The analytic continuation
to the left half plane takes the form
\[
-z^\nu(e^{2\pi i \nu}q_4''(z)+q_2''(z)-q_3''(z))=\int
s^{\nu-1}e^{\frac{z^2}{2s^2}-\alpha\frac{z}{s}+s} \ud s, \qquad \Re
z <0,
\]
where the integration is over the same deformed Hankel contour and
the argument of $z$ is chosen in the interval $(-3\pi/2,-\pi/2)$.

We can use the right-hand side of \eqref{eq: analytic combination q
1} to determine its value at zero
\[
\frac{2\pi i}{\Gamma(1-\nu)},
\]
where we also used \eqref{eq: reciprocal Gamma}. Hence we find the equality
\[
-2i\sin(\nu
\pi)p_1(z)+p_0(z)+p_2(z)=\frac{1}{\sqrt{2\pi}i}e^\frac{\alpha^2}{2}z^\nu(q_4''(z)+q_1''(z)-q_3''(z)),
\qquad \Re z >0.
\]
Now apply \eqref{eq: relation p0q1}, \eqref{eq: relation p2q2}
(assuming $\Im z >0$), and \eqref{eq: connection formula 1} to obtain
\eqref{eq: relation p1q3}. In the other quadrants similar relations
can be found.
\end{proof}

\begin{lemma} \label{lemma: contour contraction}
The function
\[
f_\alpha(s)=\int_{-\infty}^0 e^{-\frac{y^2}{2}-\alpha y+\frac{y}{s}}
\ud y
\]
is analytic in $\C \setminus \{0\}$ and satisfies $f_\alpha(s) =
\mathcal O(s)$ as $s \to 0$ uniformly within the sector $|\arg s|
\leq 3\pi/4-\delta$, for any small $\delta>0$.
\end{lemma}
\begin{proof}
We rewrite $f_\alpha(s)$ as
\[
f_\alpha(s)=e^{\frac{t^2}{2}}\sqrt{\frac{\pi}{2}} \erfc \left(
\frac{t}{\sqrt 2} \right), \qquad t=\frac{1}{s}-\alpha,
\]
where the complementary error function is defined as
\[
\erfc(z)=\frac{2}{\sqrt{\pi}}\int_z^\infty e^{-y^2} \ud y.
\]
This function has uniform asymptotics
\[
\erfc(z)=\frac{e^{-z^2}}{\sqrt \pi z}\left(1+\mathcal O \left(z^{-2}\right) \right),
\]
as $z \to \infty$ such that $|\arg z| \leq 3\pi/4-\delta$, for a
small $\delta>0$, see \cite{DLMF}. Then $f_\alpha(s)$ has the
asymptotics
\[
f_\alpha(s)=\frac{s}{1-\alpha s}\left(1+\mathcal
O\left(\frac{s^2}{(1-\alpha s)^2} \right)\right)=\mathcal O(s),
\]
as $s \to 0$ such that $|\arg s| \leq 3\pi/4-\delta$.
This proves the lemma.
\end{proof}

\subsection{Wronskian matrix}

From the three special solutions $p_j$, $j=1,2,3$, introduced in the
previous section we construct the Wronskian matrix
\begin{equation}\label{Wronskian}
W_n(z) = \begin{pmatrix} p_0(z) & p_1(z) & p_2(z)\\
p_0'(z) & p_1'(z) & p_2'(z)\\
p_0''(z) & p_1''(z) & p_2''(z)
\end{pmatrix}\times\diag(1,1,e^{\pm\nu\pi i}),\qquad\hbox{ $\pm\Im z>0$}.
\end{equation}
This Wronskian matrix will be used in the first transformation of
the RH problem for $Y(z)$, see Section~\ref{subsection:firsttransfo}.

\begin{proposition}\label{proposition:WronskianRHP}
The Wronskian matrix $W_n$ satisfies the following RH problem.
\begin{enumerate}
\item[\rm (1)] $W_n(z)$ is analytic for $z\in\cee\setminus\er$.
\item[\rm (2)] $W_n$ has jumps on $\er^+$ and $\er^-$ given by
\begin{align}\label{defjumpmatrixW}
W_{n,+}(x) = W_{n,-}(x)\times\left\{\begin{array}{ll}
\diag\left(1,\begin{pmatrix} 1 & 0 \\
-1 & 1\end{pmatrix}\right),& \qquad x \in \mathbb R^+,
\\
\diag\left(\begin{pmatrix} e^{2\nu\pi i} &  e^{\nu\pi i} \\
0 & 1\end{pmatrix},e^{2\nu\pi i}\right),& \qquad x \in \mathbb R^-.
\end{array}\right.
\end{align}
\item[\rm (3)] As $z\to\infty$, we have that
\begin{multline}\label{asymptoticconditionW1}
    W_n(z) = 3^{-\frac 12}\tau^\frac{\nu+2}{3}e^{\frac{\alpha^2}{3n}}z^{\frac{2\nu-2}{3}}\diag(n^{-1}\tau^{-4/3}z^{1/3},1,n \tau^{4/3}z^{-1/3})
    \\ \times\left(I+ \begin{pmatrix} 0&*&0 \\ 0&0&* \\ *&0&0 \end{pmatrix}z^{-1/3}+\begin{pmatrix} 0&0&* \\ *&0&0 \\
    0&*&0
    \end{pmatrix}z^{-2/3}+O\left(z^{-1}\right)\right)
\begin{pmatrix}1&\om^2&\om \\
    1&1&1\\ 1& \om&\om^2\end{pmatrix}\\
     \times\diag(1,-\sigma^{-1},\sigma)e^{n \Theta(z)},
\end{multline}
as $z\to\infty$ in closed sectors of the upper half plane $\{z\in\cee\mid\Im
z>0\}$. Here the $*$ denote certain constant real entries and
\begin{equation}\label{sigma:constant}\sigma:=e^{\frac{\pi i}{3}(\nu-1)}.\end{equation}
We have the same expression as $z\to\infty$ in closed sectors of the lower half
plane, but then with $\om$ replaced by $\om^{-1}$ and
$\diag(1,-\sigma^{-1},\sigma)$ by $\diag(1,\sigma,\sigma^{-1})$.
\end{enumerate}
\end{proposition}

\begin{proof} This follows from Lemmas~\ref{lemma:jumpsp123} and
\ref{lemma:asyp123}.
\end{proof}

Note that $W_n(z)$ is not uniquely determined from the above RH
problem, since we did not specify the behavior near the origin
$z=0$. This will be done in \eqref{W:zero} below.

\begin{lemma}\label{lem:asy of W}
The asymptotics of $W_n$ can be rewritten as
\begin{multline}\label{asy:W:bis}
    W_n(z) = -i\tau^\frac{\nu+2}{3}e^{\frac{\alpha^2}{3n}}z^{\frac{2\nu-2}{3}}
    \left[\begin{pmatrix} 1&*&* \\ 0&1&* \\ 0&0&1 \end{pmatrix}+\mathcal  O \left(z^{-1}\right)\right]
    \diag(n^{-1}\tau^{-4/3}z^{1/3},1,n \tau^{4/3}z^{-1/3})
   \\ \times A_{\pm}
     \diag(1,\sigma^{\mp 1},\sigma^{\pm 1})e^{n \Theta(z)},
\end{multline}
as $z\to\infty$ in closed sectors of the half plane
$\{z\in\cee\mid\pm\Im z>0\}$, with $\sigma$ as in
\eqref{sigma:constant} and
\begin{equation}
\label{def:A+}A_+ := \frac{i}{\sqrt3} \begin{pmatrix} 1&-\omega^2 & \omega \\
1&-1&1 \\ 1&-\omega&\omega^2 \end{pmatrix},\qquad
A_- := \frac{i}{\sqrt3} \begin{pmatrix} 1&\omega & \omega^2 \\
1&1&1 \\ 1&\omega^2&\omega \end{pmatrix}.
\end{equation}
\end{lemma}
The prefactors in \eqref{def:A+} are chosen such that both matrices have
determinant 1.

\begin{remark}
It is easily seen from \eqref{thirdorderdiffeq:wtilde} and
\eqref{Wronskian} that the determinant of $W_n$ satisfies the
following linear differential equation
\begin{equation}
z(\det W_n)'(z)=(2\nu-2)\det W_n(z), \qquad z\in \mathbb{C}\setminus
\mathbb{R}^-.
\end{equation}
Hence, we have for a certain nonzero constant $K$ that
\begin{equation}\label{eq:det of Wronskian}
\det W_n(z)=K z^{2\nu-2}.
\end{equation}
\end{remark}

\section{Steepest descent analysis for $Y(z)$ with quadratic potential $W$: regular cases}
\label{section:steepestdescent:quadratic}

In this section we will perform a Deift/Zhou steepest descent analysis for
$Y(z)$ in the situation that $W$ is quadratic \eqref{eq:quadratic potential W}
and the triplet $(V,W,\tau)$ is regular in the sense of Definition~\ref{def:
regular cases}. This analysis consists of a series of invertible
transformations and results in the proof of Theorem~\ref{th: noncrit limit}.

\subsection{First transformation $Y\mapsto X$}
\label{subsection:firsttransfo}

The idea behind the first transformation $Y\mapsto X$ is inspired by
\cite{Duits2}. We will multiply the RH matrix $Y$ on the right with the inverse
transpose of the Wronskian matrix $W_n$. Recall from \eqref{Wronskian} that
$W_n$ is constructed from the solutions to the third-order differential
equation satisfied by $h_{0,n}$.

In the analysis we will work with the alternative system of weight
functions in \eqref{MOP:2} rather than those in \eqref{MOP:1}.
Thanks to \eqref{vl:def}, \eqref{Vquad:4}, and \eqref{def:p0} these
weight functions are
\begin{equation} \label{eq: weight functions}
\begin{cases}
w_{0,n}(z)=e^{-nV(z)}p_0(z),\\
w_{1,n}(z)=e^{-nV(z)}zp_0'(z),\\
w_{2,n}(z)=e^{-nV(z)}(zp_0''(z)+(1-\nu)p_0'(z)). \end{cases}
\end{equation}
Note that $w_{1,n}$ and $w_{2,n}$ have been interchanged with
respect to \eqref{vl:def}. The labeling \eqref{eq: weight functions}
will be more convenient for us. Defining
\begin{equation} \label{eq: D def}
D(z)=\begin{pmatrix} 1&0&0\\ 0&z&1-\nu \\ 0&0&z\end{pmatrix},
\end{equation}
\eqref{eq: weight functions} leads to the equality
\begin{equation}\label{Wronskian:bis}
e^{-nV(z)}W_n^T(z)D(z)=\begin{pmatrix}
w_{0,n}(z) & w_{1,n}(z) & w_{2,n}(z) \\ *&*&* \\ *&*&* \end{pmatrix},
\end{equation}
where $*$ denotes unimportant matrix entries. This identity motivates the  transformation $Y \mapsto X$ given below.

\begin{definition} \label{def: X}
We define $X$ by
\begin{equation} \label{eq: def X}
X(z)=z^{\nu/2}P_n^{-1} Y(z) \diag\left(1,D^{-1}(z) W_n^{-T}(z) e^{n
\Theta(z)}\right),
\end{equation}
for $z \in \C \setminus \R$, where $\Theta(z)$ is defined in \eqref{Theta} and $D(z)$ in \eqref{eq: D def}. $P_n$ is an invertible matrix to be specified below \eqref{Pn:constant:def}.
\end{definition}

We claim that $X(z)$ is the unique solution to RH problem~\ref{rhp:X}
below. Note that the jump on the positive real axis is simplified at the cost
of creating a jump on the negative real axis.

\begin{rhp}\label{rhp:X}
The $4\times 4$ matrix-valued function $X : \mathbb C \setminus
\mathbb R \to \mathbb C^{4 \times 4}$ defined in \eqref{eq: def X}
satisfies the following conditions.
\begin{enumerate}
\item[\rm (1)] $X(z)$ is analytic for $z\in\mathbb C\setminus\mathbb R$.
\item[\rm (2)] We have
\[
X_+(x)=X_-(x) J_X(x), \qquad x \in \R,
\]
where the jump matrices $J_X$ are given by
\begin{align*}
J_X(x) = \diag \left( 1, e^{-n\Theta_-(x)} \right)
\begin{pmatrix} 1&e^{-nV(x)}&0&0 \\ 0&1&0&0 \\ 0&0&1&1 \\ 0&0&0&1  \end{pmatrix}
\diag \left(1,e^{n\Theta_+(x)}\right),
\end{align*}
if $x>0$,
\begin{align*}
J_X(x) = \diag \left( 1, e^{-n\Theta_-(x)} \right)
\begin{pmatrix} e^{\nu \pi i}&0&0&0 \\ 0&e^{-\nu \pi i} &0&0 \\ 0&-1&e^{\nu \pi i}&0 \\ 0&0&0&e^{-\nu \pi i} \end{pmatrix}
\diag \left(1,e^{n\Theta_+(x)}\right),
\end{align*}
if $x<0$. See \eqref{Theta} for the definition of $\Theta$.
\item[\rm (3)] We have as $z\to \infty$ with $\pm\Im z >0$,
\begin{multline}\label{asy:X}
X(z)= \left[ I+ \mathcal O(z^{-1})\right]
\diag \left(z^n, z^{-\frac{n-1}{3}},z^{-\frac{n+1}{3}},z^{-\frac{n}{3}}\right) \diag\left(z^{\nu/2},z^{-\nu/6} A_{\pm}^{-T} \right)  \\
\times \diag \left( 1,1,\sigma^{\pm 1},\sigma^{\mp 1} \right),
\end{multline}
where $\sigma$ and $A_{\pm}$ are defined in \eqref{sigma:constant}
and \eqref{def:A+}, respectively.

\item[\rm (4)] As $z\to 0$, $z\in\mathbb{C}\setminus\mathbb{R}$, we have
\begin{equation}\label{eq:zero behavior of X}
\left\{
       \begin{array}{lll}
         X(z)\diag(|z|^{-\nu/2},|z|^{\nu/2},|z|^{-\nu/2},|z|^{\nu/2}) = \mathcal O(1),& \hbox{if $\nu>0$,} \\
X(z)\diag(1,(\log|z|)^{-1},1,(\log|z|)^{-1})=\mathcal O(1),& \hbox{if $\nu=0$,} \\
         X(z)=\mathcal O(|z|^{\nu/2}),\qquad X^{-1}(z)=\mathcal O(|z|^{\nu/2}), & \hbox{if $-1<\nu<0$.}
       \end{array}
     \right.
\end{equation}
\end{enumerate}
\end{rhp}

\begin{proof}
The jumps of $X$ follow directly from \eqref{defjumpmatrixW},
\eqref{Wronskian:bis} and the definitions. Next we check the
asymptotics for $z\to\infty$. From \eqref{asymptoticconditionY} and
\eqref{eq: D def} we obtain
\begin{equation*}
Y(z) \diag\left(1,D^{-1}(z)\right) =(I + \mathcal O(z^{-1}))\diag \left(
z^n,z^{-n/3},z^{-n/3-1},z^{-n/3-1}\right),
\end{equation*}
as $z \to \infty$.
Moreover, Lemma \ref{lem:asy of W} yields
\begin{multline*}
z^{\nu/2} W_n^{-T}(z) e^{n \Theta(z)}=i\tau^{-\frac{\nu+2}{3}}
e^{-\frac{\alpha^2}{3n}} z^{\frac 23-\frac{\nu}{6}} \left[ \begin{pmatrix} 1&0&0 \\
*&1&0 \\ *&*&1
\end{pmatrix} +\mathcal O \left(z^{-1}\right) \right]\\ \times \diag
\left(\tau^{4/3}nz^{-1/3},1,\tau^{-4/3} n^{-1} z^{1/3} \right)A_{\pm}^{-T}\diag
\left( 1,\sigma^{\pm 1},\sigma^{\mp 1}\right),
\end{multline*}
as $z \to \infty$ in a sector of $\pm\Im z >0$. Combining this we get
\begin{multline}\label{Pn:constant:def}
z^{\nu/2} Y(z) \diag\left(1,D^{-1}(z) W_n^{-T}(z) e^{n \Theta(z)}\right) \\
=
\left[\begin{pmatrix} 1&0&0&0 \\ 0&i\tau^\frac{2-\nu}{3} e^{-\frac{\alpha^2}{3n}}n&*&* \\
0 & 0 &i\tau^{-\frac{2+\nu}{3}} e^{-\frac{\alpha^2}{3n}}& 0
\\ 0 & 0 &*&i\tau^{-\frac{6+\nu}{3}} e^{-\frac{\alpha^2}{3n}}n^{-1} \end{pmatrix}+ \mathcal O(z^{-1})\right]\\
\times \diag \left(z^n,
z^{-\frac{n-1}{3}},z^{-\frac{n+1}{3}},z^{-\frac{n}{3}}\right)
\diag\left(z^{\nu/2},z^{-\nu/6} A_{\pm}^{-T} \right)\diag \left( 1,
1,\sigma^{\pm 1},\sigma^{\mp 1} \right),
\end{multline}
as $z \to \infty$ with $\pm\Im z>0$, for certain constants $*$.
Finally, we define $P_n$ to be the constant matrix in the square
bracket in \eqref{Pn:constant:def}, which yields \eqref{asy:X}.

Next we show the behavior near the origin in \eqref{eq:zero behavior of X}.
Recall the three weight functions $e^{-nV(x)}h_{0,n}(x)$, $ e^{-nV(x)}x h_{0,n}'(x)$,
and $\ e^{-nV(x)}(x h_{0,n}''(x)+(1-\nu) h_{0,n}'(x))$ in the RH problem for $Y$. It
follows from the explicit form of $Y$ that
\begin{equation}\label{eq:asy of Y near 0}
\begin{cases}
Y(z)=\mathcal O(1),                                           & \hbox{if $\nu>0$,}   \\
Y(z)\diag(1,(\log |z|)^{-1},1,(\log |z|)^{-1})=\mathcal O(1), & \hbox{if $\nu=0$,}   \\
Y(z)\diag(1,|z|^{-\nu},|z|^{-\nu},|z|^{-\nu})=\mathcal O(1),  & \hbox{if $-1<\nu<0$,}\\
Y(z)^{-T}\diag(|z|^{-\nu},1,1,1)=\mathcal O(1),               &
\hbox{if $-1<\nu<0$,}
\end{cases}
\end{equation}
as $z\to 0$. On the other hand, we see from the definitions of $p_i$,
$i=0,1,2$, that there exist some constants $A,B,C,D$ that depend on $\nu$ such
that
\begin{equation}\label{W:zero}
\begin{aligned}
p_0(z) &=Az^{\nu}+\mathcal O(|z|^{\nu+1}), \\
p_1(z) &=\begin{cases}
  Bz^{\nu}+\mathcal O(|z|^{\nu+1})+C+\mathcal O(z), & \text{if $\nu\neq 0$,} \\
  B+ \mathcal O(z)+\log z\left(C+\mathcal O(z)\right), & \text{if $\nu= 0$,}
  \end{cases} \\
p_2(z) &= Dz^{\nu}+\mathcal O(|z|^{\nu+1}),
\end{aligned}
\end{equation}
as $z\to 0$. This, together with \eqref{eq: D def} and \eqref{Wronskian} gives
\begin{align*}
D^T(z)W_n(z) =
\begin{pmatrix} Az^{\nu}+\mathcal O(|z|^{\nu+1}) & \mathcal O(|z|^{\nu})+\mathcal O(1)+\mathcal O(\log |z|) & Dz^{\nu}+\mathcal O(|z|^{\nu+1})\\
\nu Az^{\nu}+\mathcal O(|z|^{\nu+1}) & \mathcal
O(|z|^{\nu})+\mathcal O(z) & \nu
Dz^{\nu}+\mathcal O(|z|^{\nu+1})\\
\mathcal O(|z|^{\nu}) & \mathcal O(|z|^{\nu})+\mathcal O(1)+\mathcal
O(\log |z|) & \mathcal O(|z|^{\nu})\end{pmatrix},
\end{align*}
as $z\to 0$, where we understand that the $\mathcal O(\log |z|)$
terms are absent if $\nu\neq 0$. Also in case $\nu=0$ the $(1,2)$
and the $(3,2)$ entry have to be replaced by $\mathcal O(z)$. Note
that the terms of order $\mathcal O(|z|^{\nu-1})$ in the last row
all cancel.

Using the cofactor formula for the inverse transpose and the fact
that
\[ \det(D^T(z)W_n(z))=Kz^{2\nu}, \]
see \eqref{eq:det of Wronskian}, we find
\begin{multline*}
D^{-1}(z)W_n^{-T}(z) \\ =\frac{1}{Kz^{2\nu}}\begin{pmatrix} \mathcal
O(|z|^{2\nu})+\mathcal O(|z|^{\nu}) & \mathcal O(|z|^{2\nu}) &
\mathcal O(|z|^{2\nu})+\mathcal O(|z|^{\nu})
\\
\mathcal O(|z|^{2\nu})+\mathcal O(|z|^{\nu})+\mathcal O(\log |z|) & \mathcal O(|z|^{2\nu})
& \mathcal O(|z|^{2\nu})+\mathcal O(|z|^{\nu})+\mathcal O(\log |z|) \\
\mathcal O(|z|^{2\nu})+\mathcal O(|z|^{\nu}) & \mathcal O(|z|^{2\nu}) & \mathcal O(|z|^{2\nu})+\mathcal O(|z|^{\nu}) \\
\end{pmatrix},
\end{multline*}
where again the $\mathcal O(\log |z|)$ terms are absent if $\nu\neq
0$. Using this with \eqref{eq:asy of Y near 0} and \eqref{eq: def
X}, we arrive at \eqref{eq:zero behavior of X} after a
straightforward calculation.
\end{proof}

\subsection{Second transformation $X\mapsto U$}
\label{subsection:secondtransfo}

The second transformation $X\mapsto U$ serves to (partly) normalize
the behavior at infinity. To do this we will use certain functions
related to the vector equilibrium problem. We call these functions
$\lambda$-functions and define them as a transformed version of the
$\lambda$-functions in \cite{DKM}. More precisely, denoting with
$\lam^\DKM_j(z)$, $j=1,2,3,4$, the  $\lambda$-functions used in
\cite{DKM}, we will work with the square root versions
$\lambda_j(z)=2\lam^\DKM_j(\sqrt{z})$.

\begin{lemma}\label{lemma: lambda functions}
There exist functions $\lam_j$, $j=1,2,3,4$, analytic on $\C \setminus \R$ that
satisfy the following conditions.
\begin{itemize}
\item[\rm (a)] As $z \to \infty$ we have
\begin{align}
\notag \lam_1(z) &= V(z)-\log(z)-\ell_1+\mathcal O \left( z^{-1}\right), \\
\lam_2(z) &= \theta_1(z)+ \tfrac13 \log(z)+Cz^{-1/3}+Dz^{-2/3}+\mathcal O
\left( z^{-1} \right),
\label{eq:asy of lam2} \\
\notag \lam_3(z) &= \theta_2(z)+\begin{cases} \frac13 \log(z)+C\omega z^{-1/3}+D
\omega^2
z^{-2/3}+\mathcal O \left( z^{-1} \right) & \Im z >0, \\
\frac13 \log(z)+C\omega^2 z^{-1/3}+D \omega  z^{-2/3}
+\mathcal O \left( z^{-1} \right) & \Im z <0, \end{cases} \\
\lam_4(z) &= \theta_3(z)+\begin{cases} \frac13 \log(z)+C\omega^2 z^{-1/3}+D \omega  z^{-2/3}+\mathcal O \left( z^{-1} \right) & \Im z >0, \\
\frac13 \log(z)+C\omega z^{-1/3}+D \omega^2  z^{-2/3}+\mathcal O
\left( z^{-1} \right) & \Im z <0, \end{cases}\label{eq:asy of lam4}
\end{align}
where $C,$ $D$ and $\ell_1$ are real constants.
\item[\rm (b)] There exists a positive integer $N$, two sets of ordered
numbers
\[
0=b_0 \leq a_1 < b_1 < a_2 < \cdots < a_N < b_N<a_{N+1}=\infty,
\] \[
1=\alpha_0>\alpha_1>\cdots > \alpha_N=0,
\]
and constants $c_2 \geq 0$ and $c_3 \geq 0$ such that the
$\lam$-functions satisfy the following jump conditions:
\begin{itemize}
\item[\rm (i)] On $\R^+$ we have
\begin{align}
\label{eq:lam1 and lam2}\lam_{1,\pm}-\lam_{2,\mp} &=0 && \text{on $(a_j,b_j)$, $j=1,\ldots,N$,} \\
\notag \lam_{1,+}-\lam_{1,-} &=-2\pi i \alpha_j && \text{on $(b_j,a_{j+1})$, $j=0,\ldots,N$,} \\
\notag \lam_{2,+}-\lam_{2,-} &=2\pi i \alpha_j && \text{on $(b_j,a_{j+1})$, $j=0,\ldots,N$,} \\
\label{eq:lam 3 and lam 4}
\lam_{3,\pm}-\lam_{4,\mp} &=0 && \text{on $(c_3,\infty)$,} \\
\notag \lam_{3,+}-\lam_{3,-} &=-2\pi i/3 && \text{on $(0,c_3)$,} \\
\notag \lam_{4,+}-\lam_{4,-} &= 2\pi i/3 && \text{on $(0,c_3)$.}
\end{align}
\item[\rm (ii)] On $\R^-$ we have
\begin{align}
\lam_{1,+}-\lam_{1,-} &=-2\pi i && \text{on $\R^-$}, \notag \\
\lam_{2,\pm}-\lam_{3,\mp} &=\pm 2\pi i/3 && \text{on $(-\infty,-c_2)$}, \label{eq:relations lam2 and lam3}\\
\lam_{2,+}-\lam_{2,-} &= 2\pi i/3 && \text{on $(-c_2,0)$} \notag\\
\lam_{3,+}-\lam_{3,-} &= 2\pi i/3 && \text{on $(-c_2,0)$} \notag\\
\lam_{4,+}-\lam_{4,-} &= 2\pi i/3 && \text{on $\R^-$}.    \notag
\end{align}
\end{itemize}
\end{itemize}
\end{lemma}
\begin{proof}
Define
\begin{equation}\label{eq:def lambdafunction}
\lam_j(z)=2\lam_j^{\DKM}(\sqrt z), \qquad j=1,2,3,4, \qquad z \in \C
\setminus \R,
\end{equation}
where we take the potentials as in \eqref{eq: relation
chiral/hermitian potentials}. Then (a) follows from \cite[Lemma
4.14]{DKM} and \eqref{eq:theta function}. Note however that error
terms stated there are not optimal. (b)(i) is direct from
\cite[Lemma 4.12]{DKM}. To prove (ii) we need some preparations. By
\cite[Definition 4.1]{DKM}, it follows that for $x \geq 0$
\begin{align*}
g_1^{\DKM}(ix)-g_1^{\DKM}(-ix)&=\pi i, \\
g_{2,\pm}^{\DKM}(ix)-g_{2,\mp}^{\DKM}(-ix)&= 2\pi i/3, \\
g_3^{\DKM}(ix)-g_3^{\DKM}(-ix)&=\pi i/3,
\end{align*}
and
\begin{align*}
g_2^{\DKM}(-z)&=g_2^{\DKM}(z)\mp 2\pi i/3, && \pm\Im z>0,\\
g_3^{\DKM}(-z)&=g_3^{\DKM}(z)\mp \pi i/3, && \pm\Im z>0.
\end{align*}

Then \cite[Definition 4.11]{DKM} and \eqref{eq:even of theta DKM}
yield for $x \geq 0$
\begin{align*}
\lam_1^{\DKM}(ix)-\lam_1^{\DKM}(-ix)&= -\pi i,\\
\lam_{2,-}^{\DKM}(ix)-\lam_{2,-}^{\DKM}(-ix)&= \pm
\pi i/3,&  c_2^{\DKM}<&x<+\infty, \\
\lam_{2,-}^{\DKM}(ix)-\lam_{2,-}^{\DKM}(-ix)&= \pi i/3,& 0< &x < c_2^{\DKM},\\
\lam_{3,-}^{\DKM}(ix)-\lam_{3,-}^{\DKM}(-ix)&= \pi i/3,& 0< &x < c_2^{\DKM},\\
\lam_4^{\DKM}(ix)-\lam_4^{\DKM}(-ix)&= \pi i/3.
\end{align*}
This, together with \eqref{eq:def lambdafunction}, implies (ii) in
(b), where we take $c_2=(c_2^{\DKM})^2$.
\end{proof}
The $\lam$-functions also satisfy a number of inequalities stated in the following lemma.
\begin{lemma}\label{lem:inequality of lamda functions}
\begin{align}
\Re (\lam_{2,+}-\lam_{1,-})<0, \qquad & \text{on $(b_j,a_{j+1})$,
$j=0,\ldots,N$, } \\
\Re (\lam_{2,+}-\lam_{3,-})<0, \qquad & \text{on $(-c_2, 0]$,} \\
\Re (\lam_{4,+}-\lam_{3,-})<0, \qquad & \text{on $[0,c_3)$.}
\end{align}
\end{lemma}
\begin{proof}
This follows from \eqref{eq:def lambdafunction} and \cite[Lemma 4.13]{DKM}.
\end{proof}

\begin{remark}
For the constants $c_2$ and $c_3$ in Lemma \ref{lemma: lambda
functions}, we have $c_2>-y^*(\alpha)$ and $c_3=0$ if $\alpha\geq
0$, and $c_3<x^*(\alpha)$ if $\alpha<0$, where $x^*(\alpha)$ and
$y^*(\alpha)$ are given in Lemma \ref{lem:theta function}. The
functions $\lam_1$ and $\lam_2$ are defined and analytic on $\C
\setminus (-\infty, b_N]$, whereas $\lam_3$ and $\lam_4$ are defined
and analytic in  $\C \setminus \R$.
\end{remark}

For future reference we specify the behavior of the $\lambda$-functions near the origin.

\begin{lemma} \label{lemma: lambda near 0}
In a neighborhood of the origin the $\lambda$-functions, defined in
Lemma \ref{lemma: lambda functions}, have the following behavior for
$z\to 0$.
\begin{itemize}
\item[(a)] In Cases I and IV there exists a constant $c_1>0$ such
that
\[
(\lambda_1-\lambda_2)(z)=\mp 2\pi i \pm ic_1 z^{1/2}+\mathcal O(z),\qquad \pm
\Im z >0.
\]
\item[(b)] In Case III there exists a constant $c_2>0$ such that
\[
(\lambda_2-\lambda_3)(z)=\pm \frac{4\pi i}{3} + c_2 z^{1/2}+\mathcal
O(z),\qquad \pm \Im z >0.
\]
\item[(c)] In Cases I and II there exists a constant $c_3>0$ such that
\[
(\lambda_3-\lambda_4)(z)=\mp \frac{2\pi i}{3} \pm ic_3 z^{1/2}+\mathcal
O(z),\qquad \pm \Im z >0.
\]
\end{itemize}
\end{lemma}

\begin{proof}
(a), (b), and (c) follow from \eqref{eq:def lambdafunction} and
the proofs of \cite[Lemma 7.5, Lemma 7.1, Lemma 7.2]{DKM} respectively.
\end{proof}

We can now define the transformation $X \mapsto U$ in a similar way as in \cite{DKM}.
\begin{definition} \label{def: U}
We define the $4\times 4$ matrix-valued function $U$ by
\begin{equation}\label{eq:def of U}
U(z)=U_0 e^{nL}X(z)e^{nG(z)},
\end{equation}
where $U_0$ is a constant invertible matrix to be determined later,
\[
G=\diag \left( \lam_1-V,\lam_2-\theta_1,\lam_3-\theta_2,\lam_4-\theta_3
\right),
\]
and $L= \diag (\ell_1,0,0,0)$ with $\ell_1$ the constant from
Lemma \ref{lemma: lambda functions}.
\end{definition}

\begin{rhp}\label{rhp:U}
The matrix-valued function $U$ defined in \eqref{eq:def of U} is the
unique solution of the following RH problem.
\begin{enumerate}
\item[\rm (1)] $U(z)$ is analytic for $z\in\mathbb C\setminus\mathbb R$.
\item[\rm (2)] We have
\[
U_+(x)=U_-(x)\left\{\begin{array}{ll}\diag\left( \left(J_U
\right)_1(x),\left( J_U
\right)_3 (x)\right),& \qquad x \in \er_+,\\
\diag\left(e^{\nu\pi i},\left( J_U \right)_2(x),e^{-\nu\pi
i}\right),& \qquad x \in \er_-,
\end{array}\right.
\]
where
\begin{align*}
\left( J_U \right)_1 &= \begin{pmatrix}
e^{n(\lam_{1,+}-\lam_{1,-})}&e^{n(\lam_{2,+}-\lam_{1,-})}
\\ 0&e^{n(\lam_{2,+}-\lam_{2,-})}
\end{pmatrix},\\
\left( J_U \right)_2 &= \begin{pmatrix} e^{-\nu \pi
i}e^{n(\lam_{2,+}-\lam_{2,-})} &0
\\ -e^{n(\lam_{2,+}-\lam_{3,-})}&e^{\nu \pi i}e^{n(\lam_{3,+}-\lam_{3,-})}
\end{pmatrix},\\
\left( J_U \right)_3 &= \begin{pmatrix} e^{n(\lam_{3,+}-\lam_{3,-})}
& e^{n(\lam_{4,+}-\lam_{3,-})}
\\ 0&e^{n(\lam_{4,+}-\lam_{4\,-})}
\end{pmatrix}.
\end{align*}

\item[\rm (3)] We can choose the constant matrix $U_0$ such that
\begin{multline}\label{eq: asy of U}
U(z)= \left[ I+ \mathcal O \left(z^{-1}\right)\right]\diag \left(1,
z^{1/3},z^{-1/3},1\right)\\
\times\diag\left(z^{\nu/2},z^{-\nu/6} A_{\pm}^{-T}\right)\diag
\left( 1, 1,\sigma^{\pm},\sigma^{\mp 1} \right),
\end{multline}
as $z\to \infty$ with $\pm\Im z >0$, where $\sigma$ and $A_{\pm}$
are defined in \eqref{sigma:constant} and \eqref{def:A+}
respectively.
\item[\rm (4)] $U(z)$ has the same behavior as $X(z)$ near the origin; see
\eqref{eq:zero behavior of X}.
\end{enumerate}
\end{rhp}

\begin{proof}
Noting that $e^{n(\lam_{1,+}-\lam_{1,-})}=e^{n(\lam_{4,+}-\lam_{4,-})}=1$ on
$\er^-$, due to Lemma \ref{lemma: lambda functions}(b) and our assumption that
$n \equiv 0 \mod 3$, the jump condition for $U$ in item (2) follows from a
straightforward calculation.

To show the asymptotic behavior of $U$ for large $z$ in item (3), we
see from Lemma \ref{lemma: lambda functions}(a) that
\begin{multline*}
e^{nG(z)}=\diag \left( z^{-n}e^{-n \ell_1},z^{n/3},z^{n/3},z^{n/3}\right) \left[ I+nC \begin{pmatrix} 0&0 \\ 0&
\Omega_{\pm}\end{pmatrix}z^{-1/3}\right. \\
\left. +nE\begin{pmatrix} 0&0 \\ 0&
\Omega_{\mp}\end{pmatrix}z^{-2/3} +\diag \left(\mathcal O(z^{-1}),\mathcal O(z^{-1}),\mathcal O(z^{-1}),\mathcal O(z^{-1})\right)\right],
\end{multline*}
as $z \to \infty$ with $\pm \Im z >0$. Here $E=D+C^2/2$ and
\begin{equation}\label{Omega:plusmin}
\Omega_+= \diag(1,\omega,\omega^2), \qquad \Omega_-=\diag (1, \omega^2,\omega).
\end{equation} Then, by the asymptotic behavior of $X$ in \eqref{asy:X} and the definition of $U$ in
\eqref{eq:def of U}, we get
\begin{multline*}
U(z)=U_0\left( I+ \mathcal O \left(z^{-1}\right)\right)  \diag \left(1, z^{1/3},z^{-1/3},1\right)\\ \times \diag\left( z^{\nu/2},z^{-\nu/6}
A_{\pm}^{-T}\right)
\left[ I+nC \begin{pmatrix} 0&0 \\ 0& \Omega_\pm\end{pmatrix}z^{-1/3}
+nE\begin{pmatrix} 0&0 \\ 0& \Omega_\mp\end{pmatrix}z^{-2/3} \right.\\ \left.
+\diag\left(\mathcal O(z^{-1}), \mathcal O(z^{-1}),\mathcal
O(z^{-1}),\mathcal O(z^{-1})\right)\right]  \times \diag( 1,
1,\sigma^{\pm 1},\sigma^{\mp 1}),
\end{multline*}
as $z \to \infty$ with $\pm \Im z>0$. We can move the terms within square
brackets to the front at the expense of an extra constant contribution. This
follows from the observation that
\begin{equation}\label{Aplusmin:commute}
A_\pm^{-T}\Omega_\pm A_\pm^T=\begin{pmatrix} 0&0&1 \\ 1&0&0 \\ 0&1&0
\end{pmatrix} \quad \text{and} \quad A_\pm^{-T}\Omega_\mp
A_\pm^T=\begin{pmatrix} 0&1&0 \\ 0&0&1 \\ 1&0&0 \end{pmatrix}.
\end{equation}
This gives
\begin{multline}
U(z)= U_0\left[ \begin{pmatrix} 1&0&0&0 \\ 0&1&*&* \\ 0&0&1&0 \\ 0&0&*&1 \end{pmatrix}
+ \mathcal O \left(z^{-1}\right)\right]\diag \left(1, z^{1/3},z^{-1/3},1\right)\diag\left(z^{\nu/2},z^{-\nu/6} A_{\pm}^{-T} \right)  \\
\times  \diag( 1, 1,\sigma^{\pm 1},\sigma^{\mp 1}),
\end{multline}
as $z\to \infty$ with $\pm\Im z >0$. By setting $U_0$ to be the
inverse of the constant matrix between the above square brackets, we
arrive at the claim in item (3).
\end{proof}


Recalling our assumption that $n \equiv 0 \mod 3$, we find from Lemma
\ref{lemma: lambda functions}(b) that
\begin{align}\label{JU1}
\left( J_U \right)_1 &= \begin{cases}
\begin{pmatrix} e^{n(\lam_{1,+}-\lam_{1,-})}&1 \\ 0&e^{n(\lam_{2,+}-\lam_{2,-})}
\end{pmatrix},& \text{on $(a_j,b_j)$, $j=1,\ldots,N$,}\\
\begin{pmatrix} e^{-2\pi in \alpha_j}&e^{n(\lam_{2,+}-\lam_{1,-})}
\\ 0&e^{2\pi in \alpha_j} \end{pmatrix},& \text{on $(b_j,a_{j+1})$, $j=0,\ldots,N$.}
\end{cases}\\
\label{JU2} \left( J_U \right)_2&= \begin{cases}
\begin{pmatrix} e^{-\nu \pi i}e^{n(\lam_{2,+}-\lam_{2,-})} &0
\\ -1&e^{\nu \pi i}e^{n(\lam_{3,+}-\lam_{3,-})} \end{pmatrix},
& \text{on $(-\infty,-c_2)$,}\\
\begin{pmatrix} e^{-\nu \pi i} & 0
\\ -e^{n(\lam_{2,+}-\lam_{3,-})} & e^{\nu \pi i}
\end{pmatrix},
& \text{on $(-c_2,0)$,}
\end{cases}\\
\label{JU3} \left( J_U \right)_3&= \begin{cases}
\begin{pmatrix} 1 & e^{n(\lam_{4,+}-\lam_{3,-})} \\ 0&1
\end{pmatrix},&
\text{on $(0,c_3)$,}\\
\begin{pmatrix} e^{n(\lam_{3,+}-\lam_{3,-})} & 1
\\ 0& e^{n(\lam_{4,+}-\lam_{4,-})} \end{pmatrix},
& \text{on $(c_3,\infty)$.}
\end{cases}
\end{align}
Note that the diagonal terms of $\left( J_U \right)_1$ are 1 on $(0,a_1)$ and
$(b_N,\infty)$.


\subsection{Third transformation $U\mapsto S$}

In the third transformation we open lenses around
$\bigcup_{j=1}^{N}(a_j, b_j)$, $(-\infty,-c_2)$ and $(c_3,
+\infty)$, which are denoted by $L_1$, $L_2$, and $L_3$,
respectively. See Figure \ref{fig: lenses} for a plot of the lenses
in Case IV if $N=3$. These lenses are chosen such that the following
estimates hold.

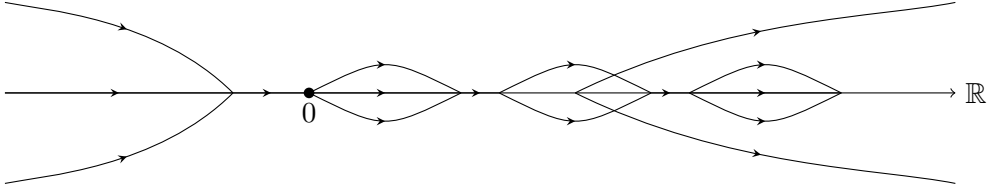
\begin{figure}[t]
\centering
\begin{tikzpicture}[scale=1]
\draw[->] (-4,0)--(8.5,0) node[right]{$\R$};
\begin{scope}[decoration={markings,mark= at position 0.5 with {\arrow{stealth}}}]
\draw[postaction={decorate}]   (-4,0)--(-1,0);
\draw[postaction={decorate}]   (-1,0)--(0,0);
\draw[postaction={decorate}]   (0,0)--(2,0);
\draw[postaction={decorate}]   (2,0)--(2.5,0);
\draw[postaction={decorate}]   (4.5,0)--(5,0);
\draw[postaction={decorate}]   (5,0)--(7,0);
\draw[postaction={decorate}]   (3.5,0) ..  controls (5.5,1) and (7.5,1) .. (8.5,1.2);
\draw[postaction={decorate}]   (3.5,0) .. controls (5.5,-1) and (7.5,-1) .. (8.5,-1.2);
\draw[postaction={decorate}]   (-4,1.2) .. controls (-3.5,1.1) and (-2,1) .. (-1,0);
\draw[postaction={decorate}]   (-4,-1.2) .. controls (-3.5,-1.1) and (-2,-1) .. (-1,0);
\draw[postaction={decorate}] (0,0) .. controls (1,0.5) and (1,0.5) .. (2,0);
\draw[postaction={decorate}] (0,0) .. controls (1,-0.5) and (1,-0.5) .. (2,0);
\draw[postaction={decorate}] (2.5,0) .. controls (3.5,0.5) and (3.5,0.5) .. (4.5,0);
\draw[postaction={decorate}] (2.5,0) .. controls (3.5,-0.5) and (3.5,-0.5) .. (4.5,0);
\draw[postaction={decorate}] (5,0) .. controls (6,0.5) and (6,0.5) .. (7,0);
\draw[postaction={decorate}] (5,0) .. controls (6,-0.5) and (6,-0.5) .. (7,0);
\end{scope}
\fill (0,0) circle (2pt) node[below]{$0$};
\end{tikzpicture}
\caption{Jump contour $\Sigma_S$ for $S$ consisting of the lips of the lenses $L_1$, $L_2$, $L_3$, and the axes. This is a picture of a Case IV situation with $N=3$ since $a_1=0$ and $c_2,c_3>0$. The lips of the lens around $\supp(\mu_3)$ intersect the lips of the lens around $(a_k,b_k)$ if and only if $c_3 \in (a_k,b_k)$.}
\label{fig: lenses}
\end{figure}

\begin{lemma}\label{lem:choice of L2 and L3}
We can find a neighborhood $L_1$ of $\bigcup_{j=1}^{N}(a_j, b_j)$ such that
\begin{equation*}
\Re(\lam_1-\lam_2)(z)<0\qquad \text{for $z\in L_1\setminus \R^+$}.
\end{equation*}

We can find a neighborhood $L_2$ of $(-\infty,-c_2)$ such that
\begin{equation*}
\Re(\lam_3-\lam_2)(z)<0\qquad \text{for $z\in L_2\setminus \R^-$},
\end{equation*}
and such that
\begin{equation*}
\{ z\in\C \mid \Re z<-R, |\Im (z)|<-\varepsilon \Re (z) \} \subset L_2
\end{equation*}
for some $\varepsilon>0$ and $R>0$.

Finally, we can find a neighborhood $L_3$ of $(c_3, +\infty)$ such that
\begin{equation*}
\Re(\lam_3-\lam_4)(z)<0 \qquad \text{for $z\in L_3\setminus \R^+$},
\end{equation*}
and such that
\begin{equation*}
\{ z\in\C \mid  \Re z>R, |\Im (z)|<\varepsilon \Re (z) \} \subset L_3
\end{equation*}
for some $\varepsilon>0$ and $R>0$.

We can assume that $L_1,L_2,L_3$ are symmetric with respect to
complex conjugation. Moreover, we can assume that the intersection
of the boundary $\partial L_j$ with the upper (or lower) half plane
is a collection of smooth curves, which can be oriented so that the
real part strictly increases. We will often refer to $\partial L_j$
as the lips of the lens $L_j$.
\end{lemma}
\begin{proof}
In view of the definition of $\lam$-functions in \eqref{eq:def lambdafunction},
the statements about $L_1$, $L_2$ and $L_3$ are immediate from Lemmas 7.1, 7.2,
and 7.5 in \cite{DKM}.
\end{proof}

To define the third transformation, we observe that in \eqref{JU1}--\eqref{JU3}
we can factorize
\begin{equation} \label{eq:factor of JU3}
\begin{aligned}
(J_U)_1&=\begin{pmatrix} 1 &  0 \\e^{n(\lam_1-\lam_2)_-} & 1\end{pmatrix}
         \begin{pmatrix} 0 & 1 \\ -1 & 0 \end{pmatrix}
         \begin{pmatrix} 1 & 0 \\ e^{n(\lam_1-\lam_2)_+} & 1 \end{pmatrix},
         &&\hbox{on $(a_j, b_j)$},\\
(J_U)_2&=\begin{pmatrix} 1 & -e^{-\nu\pi i}e^{n(\lam_3-\lam_2)_-} \\ 0 & 1 \end{pmatrix}
         \begin{pmatrix} 0 & 1 \\ -1 & 0 \end{pmatrix}
         \begin{pmatrix} 1 & -e^{\nu\pi i}e^{n(\lam_3-\lam_2)_+} \\ 0 & 1 \end{pmatrix},
         &&\hbox{on $(-\infty,-c_2)$},\\
(J_U)_3&=\begin{pmatrix} 1 &  0 \\ e^{n(\lam_3-\lam_4)_-} & 1 \end{pmatrix}
         \begin{pmatrix} 0 & 1 \\ -1 & 0 \end{pmatrix}
         \begin{pmatrix} 1 & 0 \\ e^{n(\lam_3-\lam_4)_+} & 1 \end{pmatrix},
         &&\hbox{on $(c_3,\infty)$},
\end{aligned}
\end{equation}
with the aid of \eqref{eq:lam1 and lam2}, \eqref{eq:lam 3 and lam 4},
\eqref{eq:relations lam2 and lam3} and the fact that $n$ is a multiple of
three.

We then define the third transformation $U\mapsto S$ by successively
setting
\begin{equation}\label{def:UtoT 1}
T(z)=U(z)\times\left\{
\begin{array}{ll}
I \pm e^{\pm \nu\pi i}e^{n(\lam_3-\lam_2)(z)}E_{2,3},
& \hbox{for $z\in L_2\cap\C^{\pm}$,} \\
I \mp e^{n(\lam_3-\lam_4)(z)}E_{4,3}, & \hbox{for $z\in L_3
\cap\C^{\pm}$,}\\
I, & \hbox{elsewhere},
                 \end{array}
               \right.
\end{equation}
and
\begin{equation}\label{def:TtoS 1}
S(z)=\left\{\begin{array}{ll} T(z)\times \left(I \mp
e^{n(\lam_1-\lam_2)(z)}E_{2,1}\right),& \qquad \text{for $z\in
L_1\cap\C^{\pm}$,}\\
T(z),& \qquad \text{elsewhere},\end{array}\right.
\end{equation}
where $\cee^\pm$ denotes the upper or lower half plane respectively,
and $E_{i,j}$ denotes the $4\times 4$ elementary matrix of which all
entries are $0$, expect for the $(i,j)$-th entry, which is $1$. Then
$S$ is the unique solution of the following RH problem.

\begin{rhp}\label{RHP:S}
The matrix-valued function $S$ satisfies the following conditions.
\begin{enumerate}
\item[\rm (1)] $S(z)$ is analytic for $z\in\cee\setminus\Sigma_S$,
where $\Sigma_S$ is the contour consisting of the real axis and the lips of the
lenses $L_i$, $i=1,2,3$, see Figure \ref{fig: lenses}.
\item[\rm (2)] For $z\in \Sigma_S$, $S$ has a jump
\begin{equation*}
S_{+}(z) = S_{-}(z) \left\{
      \begin{array}{ll}
        \diag\left(\left( J_{S}\right)_1(z),\left( J_{S}\right)_3(z)\right), &
\textrm{\parbox[t]{0.4\textwidth}{for $z$ in $\R^+$ and the lips of $L_1,L_3$,}}\\
        \diag\left(e^{\nu\pi i},\left( J_{S} \right)_2(z),e^{-\nu\pi i}\right), &
        \textrm{\parbox[t]{0.4\textwidth}{for $z$ in $\R^-$ and the lips of $L_2$,}}
      \end{array}
    \right.
\end{equation*}
where
\begin{equation*} (J_S)_1=\begin{cases}
            \begin{pmatrix}
            0 & 1 \\ -1 & 0
            \end{pmatrix}, & \hbox{on $(a_j,b_j)$, $j=1,\ldots,N$,}\\
             \begin{pmatrix} e^{-2\pi in \alpha_j}&e^{n(\lam_{2,+}-\lam_{1,-})}
            \\ 0&e^{2\pi in \alpha_j}
            \end{pmatrix},& \hbox{on $(b_j,a_{j+1})$, $j=0,\ldots,N$,}\\
            \begin{pmatrix}
            1 & 0 \\
            e^{n(\lam_1-\lam_2)} & 1
            \end{pmatrix}, & \hbox{on the lips of $L_1$,}\\
            I_2, & \hbox{on the lips of $L_3$,}
          \end{cases}
\end{equation*}
\begin{equation*}
(J_S)_2=\left\{
          \begin{array}{ll}
            \begin{pmatrix}
e^{-\nu \pi i} &0 \\ -e^{n(\lam_{2,+}-\lam_{3,-})}&e^{\nu\pi i}
\end{pmatrix}, & \hbox{on $(-c_2,0)$,} \\
            \begin{pmatrix} 0 & 1 \\
-1 & 0
\end{pmatrix}, & \hbox{on $(-\infty,-c_2)$,}\\
\begin{pmatrix}
1 & -e^{\pm \nu\pi i}e^{n(\lam_3-\lam_2)} \\
0 & 1
\end{pmatrix}, & \hbox{on the upper/lower lip of $L_2$,}
          \end{array}
        \right.
\end{equation*}
and
\begin{equation*}
(J_S)_3=\left\{
\begin{array}{ll}
          \begin{pmatrix} 1 & e^{n(\lam_{4,+}-\lam_{3,-})} \\ 0&1
\end{pmatrix}, & \hbox{on $(0,c_3)$,} \\
            \begin{pmatrix} 0 & 1 \\
-1 & 0
\end{pmatrix}, & \hbox{on $(c_3, +\infty)$,}\\
I_2,& \hbox{on the lips of $L_1$,}\\
\begin{pmatrix}
1 & 0 \\
e^{n(\lam_3-\lam_4)} & 1
\end{pmatrix}, & \hbox{on the upper/lower lip of $L_3$.}
          \end{array}
        \right.
\end{equation*}
\item[\rm (3)] $S(z)$ has the same behavior for $z\to\infty$ as $U(z)$.
\item[\rm (4)]  $S(z)$ has the same behavior near the origin as $U(z)$ (and $X(z)$), see
\eqref{eq:zero behavior of X}, provided that $z\to 0$ outside the
lenses that end in $0$.
\end{enumerate}
\end{rhp}

\begin{proof}
The jump condition is straightforward by \eqref{eq:factor of JU3}
and item (2) in RH problem \ref{rhp:U}. The asymptotic behavior of
$S$ follows from the definition in \eqref{def:UtoT
1}--\eqref{def:TtoS 1} and the large $z$ behavior of the $\lam$-functions in \eqref{eq:asy of lam2}--\eqref{eq:asy of lam4}.
\end{proof}

In view of Lemmas \ref{lem:inequality of lamda functions} and \ref{lem:choice
of L2 and L3} it is easily seen that each entry of $J_S$ is either constant or
exponentially small as $n \to +\infty$.

\subsection{Global parametrix}\label{section:global}

In this section we look for a global parametrix $S^{(\infty)}$. This
will be a good global approximation of the matrix-valued function
$S$ when $n$ is large. Ignoring all exponentially small entries for
$n\to\infty$ in $J_S$, we are led to the following model RH problem
for $S^{(\infty)}$.

\begin{rhp} \label{rhp: global parametrix}
We look for a $4\times 4$ matrix-valued function $S^{(\infty)}$ that satisfies the following conditions.
\begin{enumerate}
\item[\rm (1)] $S^{(\infty)} (z)$ is analytic for $z\in\cee\setminus \R $.
\item[\rm (2)] For $x\in \R $, $S^{(\infty)}$ has a jump
\begin{equation*}
S^{(\infty)}_{+}(x) = S^{(\infty)}_{-}(x) \left\{
      \begin{array}{ll}
        \diag\left(
\left( J_{S^{(\infty)}} \right)_1(x),\left( J_{S^{(\infty)}} \right)_3(x)\right), & \hbox{for $x$ in $\R^+$,} \\
        \diag\left(e^{\nu\pi i},\left( J_{S^{(\infty)}} \right)_2(x),e^{-\nu\pi i}\right), & \hbox{for $x$ in $\R^-$,}
      \end{array}
    \right.
\end{equation*}
where
\begin{equation*}
\left(J_{S^{(\infty)}}\right)_1=\left\{
          \begin{array}{ll}
            \begin{pmatrix}
            0 & 1 \\ -1 & 0
            \end{pmatrix}, & \hbox{on $(a_j,b_j)$, $j=1,\ldots,N$,}\\
            \diag\left( e^{-2\pi in \alpha_j},e^{2\pi in \alpha_j}\right),&
            \hbox{on $(b_j,a_{j+1})$, $j=1,\ldots,N-1$,}\\
            I_2, & \hbox{on $(0,a_1)\cup(b_N,+\infty)$,}
          \end{array}
        \right.
\end{equation*}
\begin{equation*}
\left(J_{S^{(\infty)}}\right)_2=\left\{
          \begin{array}{ll}
            \diag\left(e^{-\nu \pi i},e^{\nu \pi i}\right), & \hbox{on $(-c_2,0)$,} \\
            \begin{pmatrix} 0 & 1 \\
            -1 & 0
            \end{pmatrix}, & \hbox{on $(-\infty,-c_2)$,}\\
            \end{array}
        \right.
\end{equation*}
and
\begin{equation*}
\left(J_{S^{(\infty)}}\right)_3=\left\{
\begin{array}{ll}
          I_2, & \hbox{on $(0,c_3)$,} \\
           \begin{pmatrix} 0 & 1 \\
           -1 & 0
           \end{pmatrix}, & \hbox{on $(c_3, +\infty)$.}
          \end{array}
        \right.
\end{equation*}
\item[\rm (3)] As $z\to \infty$ with $\pm\Im z>0$, we have
\begin{multline*}
S^{(\infty)} (z)= \left[ I+ \mathcal O \left(z^{-1}\right)\right]\diag \left(1, z^{1/3},z^{-1/3},1\right)\diag\left(z^{\nu/2},z^{-\nu/6} A_{\pm}^{-T} \right)  \\
\times \diag \left( 1,1,\sigma^{\pm 1},\sigma^{\mp 1}\right).
\end{multline*}
\end{enumerate}
\end{rhp}

In this section we will construct a solution to the above RH problem.


\subsubsection*{Transforming the global parametrix: $S^{(\infty)} \mapsto N_\nu$}

We will look for a solution $S^{(\infty)}$ to RH problem \ref{rhp: global parametrix} in the form
\begin{equation}\label{transform:S to Nnu}
S^{(\infty)} (z)=
N_\nu(z)\diag\left(z^{\frac{\nu-1}{2}},z^{\frac{1-\nu}{6}}I_3\right)\diag
\left(1,1,\sigma^{\pm 1},\sigma^{\mp 1}\right), \qquad \pm \ \Im z>0.
\end{equation}
Then $N_\nu$ must satisfy the following RH problem.

\begin{rhp} \label{rhp: Nnu}
We look for a $4 \times 4$ matrix-valued function $N_\nu$ satisfying the following conditions.
\begin{enumerate}
\item[\rm (1)] $N_\nu(z)$ is analytic for $z\in\cee\setminus \R $.
\item[\rm (2)] For $x\in \R $, $N_\nu$ has a jump
\begin{equation*}
N_{\nu,+}(x) = N_{\nu,-}(x) \left\{
      \begin{array}{ll}
        \diag\left(\left( J_{N_\nu} \right)_1(x),\left( J_{N_\nu} \right)_3(x)\right), & \hbox{for $x$ in $\R^+$,} \\
        \diag\left(-1,\left( J_{N_\nu} \right)_2(x),-1\right), & \hbox{for $x$ in $\R^-$,}
      \end{array}
    \right.
\end{equation*}
where
\begin{equation*}
\left(J_{N_\nu}\right)_1=\left\{
          \begin{array}{ll}
            \begin{pmatrix}
            0 & z^{\frac{2\nu-2}{3}} \\ -z^{\frac{2-2\nu}{3}} & 0
            \end{pmatrix}, & \hbox{on $(a_j,b_j)$, $j=1,\ldots,N$,}\\
            \diag\left(e^{-2\pi in \alpha_j},
            e^{2\pi in \alpha_j}\right),& \hbox{on $(b_j,a_{j+1})$, $j=1,\ldots,N-1$,}\\
            I_2, & \hbox{on $(0,a_1)\cup(b_N,+\infty)$,}
          \end{array}
        \right.
\end{equation*}
\begin{equation*}
\left(J_{N_\nu}\right)_2=\left\{
          \begin{array}{ll}
            \diag\left(e^{-\frac{\pi}{3}(1+2\nu)i},e^{\frac{\pi}{3}(1+2\nu)i}\right), & \hbox{on $(-c_2,0)$,} \\
            \begin{pmatrix} 0 & 1 \\
            -1 & 0
            \end{pmatrix}, & \hbox{on $(-\infty,-c_2)$,}\\
            \end{array}
        \right.
\end{equation*}
and
\begin{equation*}
\left(J_{N_\nu}\right)_3=\left\{
\begin{array}{ll}
          \diag\left(-e^{-\frac{\pi}{3}(1+2\nu)i},-e^{\frac{\pi}{3}(1+2\nu)i}\right), & \hbox{on $(0,c_3)$,} \\
           \begin{pmatrix} 0 & 1 \\
           -1 & 0
           \end{pmatrix}, & \hbox{on $(c_3, +\infty)$.}
          \end{array}
        \right.
\end{equation*}
\item[\rm (3)] As $z\to \infty$ with $\pm \ \Im z>0$, we have
\begin{equation}
N_\nu(z)= \left[ I+ \mathcal O \left(z^{-1}\right)\right]\diag
\left(z^{1/2}, z^{1/6},z^{-1/2},z^{-1/6}\right)
\diag\left(1,A_{\pm}^{-T}\right).
\end{equation}
\end{enumerate}
\end{rhp}

\subsubsection*{Boundary value problem on a Riemann surface} \label{subsection:Szego:def}

The next step in the construction of the global parametrix is to find a \lq
Szeg\H o function\rq\ on a certain Riemann surface. We define a four-sheeted Riemann surface $\mathcal R$ as follows. We let $\mathcal R_j$, $j=1,2,3,4$, denote
\begin{align*}
\mathcal R_1 & = \C \setminus \bigcup_{k=1}^N [a_k,b_k],  &
\mathcal R_2 & = \C \setminus \left(\bigcup_{k=1}^N [a_k,b_k] \cup (-\infty,-c_2] \right), \\
\mathcal R_3 & = \C \setminus \left((-\infty,-c_2]\cup [c_3,\infty) \right), &
\mathcal R_4 & = \C \setminus [c_3,\infty).
\end{align*}
We connect the sheets $\mathcal R_j$, $j=1,2,3,4$, to each other in
the usual crosswise manner, e.g. $\mathcal R_1$ is connected to
$\mathcal R_2$ along the cuts $[a_k,b_k]$, $k=1,\ldots,N$. The
Riemann surface is compactified by adding two points at infinity:
$\infty_1$ is added to the first sheet while $\infty_2$ is common to
the other sheets. We define $\mathcal B$ as the union of four small
disks, one around the origin of each sheet. We denote by $\mathcal
C\subset\mathcal R$ the contour consisting of the intervals
$[a_1,b_N]$ on the first two sheets, $[-c_2,0]$ on sheets $2$ and
$3$, and $[0,c_3]$ on sheets $3$ and $4$. See Figure \ref{fig:
Riemann surface case V} for an illustration in Case V.

\begin{figure}
\centering
\begin{tikzpicture}[scale=1]
\draw (-4,-0.5)--(6,-0.5)--(8,0.5)--(-2,0.5)--cycle
      (-4,-2)--(6,-2)--(8,-1)--(-2,-1)--cycle
      (-4,-3.5)--(6,-3.5)--(8,-2.5)--(-2,-2.5)--cycle
      (-4,-5)--(6,-5)--(8,-4)--(-2,-4)--cycle;
\filldraw (0,0) circle (1pt) node[below]{$0$}
          (0,-1.5) circle (1pt) node[below]{$0$}
          (0,-3) circle (1pt) node[below]{$0$}
          (0,-4.5) circle (1pt) node[below]{$0$}
          (0.5,0)circle (1pt)(1.5,0)circle (1pt)(2.5,0)circle (1pt)(3.5,0)circle (1pt)(4.5,0)circle (1pt)(5.3,0)circle (1pt)(0.5,-1.5)circle (1pt)(1.5,-1.5)circle (1pt)(2.5,-1.5)circle (1pt)(3.5,-1.5)circle (1pt)(4.5,-1.5)circle (1pt)(5.3,-1.5)circle (1pt)
          (-1.5,-1.5) circle (1pt)(-1.5,-3)circle (1pt) (2,-3)circle (1pt)  (2,-4.5)circle (1pt) ;
\draw[thick] (0.5,0)node[above]{$a_1$}--(1.5,0)node[above]{$b_1$} (2.5,0)node[above]{$a_2$}--(3.5,0)node[above]{$b_2$} (4.5,0)node[above]{$a_3$}--(5.3,0)node[above]{$b_3$}
             (0.5,-1.5)node[below]{$a_1$}--(1.5,-1.5)node[below]{$b_1$} (2.5,-1.5)node[below]{$a_2$}--(3.5,-1.5)node[below]{$b_2$} (4.5,-1.5)node[below]{$a_3$}--(5.3,-1.5)node[below]{$b_3$};
\draw[thick] (-3,-1.5)--(-1.5,-1.5) node[above]{$-c_2$}
             (-3,-3)--(-1.5,-3) node[below]{$-c_2$};
\draw[thick] (2,-3) node[above]{$c_3$}--(7,-3)
             (2,-4.5) node[below]{$c_3$}--(7,-4.5);
\draw[dashed] (0.5,0)--(0.5,-1.5) (1.5,0)--(1.5,-1.5) (2.5,0)--(2.5,-1.5) (3.5,0)--(3.5,-1.5) (4.5,0)--(4.5,-1.5) (5.3,0)--(5.3,-1.5);
\draw[dashed] (-3,-1.5)--(-3,-3) (-1.5,-1.5)--(-1.5,-3);
\draw[dashed] (2,-3)--(2,-4.5)  (7,-3)--(7,-4.5);
\end{tikzpicture}
\caption{Plot of the Riemann surface $\mathcal R$ for Case V (i.e. $a_1,c_2,c_3>0$) and genus $g=2$. }
\label{fig: Riemann surface case V}
\end{figure}
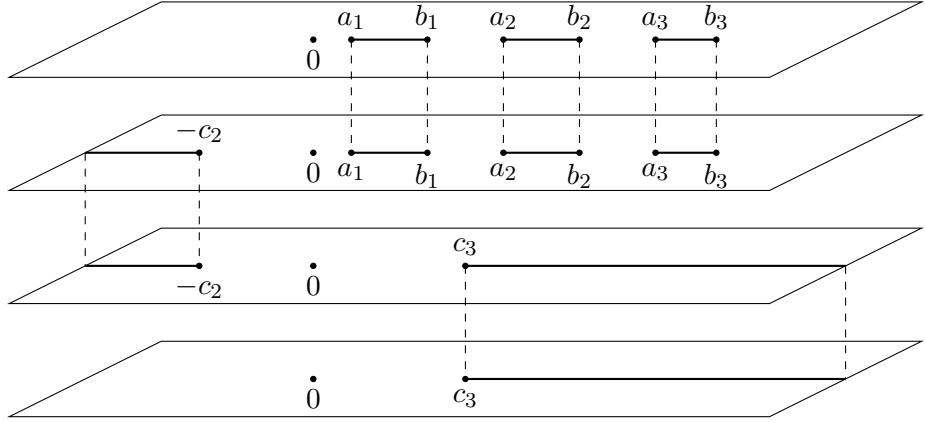

We want to construct a Szeg\H o function on this Riemann surface, i.e. we look for a scalar-valued function $f$ on $\mathcal R$ satisfying a boundary value problem. We denote the restriction of $f$ to the $j$th sheet with $f_j(z)$, $j=1,2,3,4$.

\begin{bvp} \label{bvp: f}
We look for a scalar-valued function $f$ satisfying the following conditions.
\begin{enumerate}
\item[(a)] $f$ is analytic on $\mathcal R\setminus \mathcal C$, and there exist
constants $C_1,C_2$ such that $0<C_1<|f(z)|<C_2<\infty$ on $\mathcal
R\setminus(\mathcal C\cup\mathcal B)$.
\item[(b)] $f$ has the following jumps on $\mathcal C$:
\begin{align}
\label{eq: jump f12 szego} f_{1,\pm}(x) &= f_{2,\mp}(x) x^{\frac{2\nu-2}{3}}, && x\in (a_j,b_j),\\
\label{constant:beta1}f_{1,+}(x) &= f_{1,-}(x) e^{-2\pi i\beta_j}, && x\in (b_j,a_{j+1}),\\
\label{constant:beta2}f_{2,+}(x) &= f_{2,-}(x) e^{2\pi i\beta_j}, && x\in (b_j,a_{j+1}),\\
\nonumber f_{2,+}(x) &= f_{2,-}(x) e^{\frac{\pi i}{3}(1+2\nu)}, && x\in (-c_2,0),\\
\nonumber f_{3,+}(x) &= f_{3,-}(x) e^{-\frac{\pi i}{3}(1+2\nu)}, && x\in (-c_2,0),\\
\nonumber f_{2,\pm}(x) &= f_{3,\mp}(x), && x\in (-\infty,-c_2),\\
\nonumber f_{3,+}(x) &= f_{3,-}(x) e^{\frac{\pi i}{3}(1+2\nu)}, && x\in (0,c_3),\\
\nonumber f_{4,+}(x) &= f_{4,-}(x) e^{-\frac{\pi i}{3}(1+2\nu)}, && x\in (0,c_3),\\
\nonumber f_{3,\pm}(x) &= f_{4,\mp}(x), && x\in (c_3,\infty),
\end{align}
where in \eqref{constant:beta1}--\eqref{constant:beta2} we have certain real
numbers $\beta_j$, $j=1,\ldots,N-1$, to be specified in Lemma \ref{lemma: solution bvp}.
\item[(c)] $f$ is regular at $z=\infty$ in the sense that
$$ f_1(z) = c+\mathcal O(z^{-1}),\qquad c\neq 0,
$$
and
$$ \begin{pmatrix}
f_2(z) \\ f_3(z) \\ f_4(z)
\end{pmatrix} = c_1\begin{pmatrix}1\\ 1\\ 1\end{pmatrix}+c_2 z^{-1/3}\begin{pmatrix}1\\ \om \\ \om^2\end{pmatrix}
+c_3z^{-2/3}\begin{pmatrix}1\\ \om^2 \\
\om\end{pmatrix}+\mathcal O(z^{-1}),\qquad c_1\neq 0,
$$
as $z\to\infty$ in the upper half plane.
\item[(d)] $f$ has the following behavior around the origin of each sheet for $\Im z >0$. The behavior depends on the particular case we deal with, see Section \ref{sec: classification}.
\begin{align}
f_1(z) &= \begin{cases}
\kappa_0 e^{\frac{i \pi}{4}(1-2\nu)}z^{\frac{1}{4}(2\nu-1)}
\left(1+\mathcal O\left(z^{1/2}\right)\right), & \text{in Cases I and IV,} \\
\mathcal O(1), & \text{in Cases II, III, and V,}
\end{cases} \label{eq: asymptotics f1}\\
f_2(z) &= \begin{cases}
\kappa_0 e^{-\frac{i \pi}{4}(1-2\nu)}z^{\frac{1}{12}(5-2\nu)}
\left(1+\mathcal O\left(z^{1/2}\right)\right), & \text{in Cases I and IV,} \\
\kappa_1 z^{\frac{1}{6}(1+2\nu)}(1+\mathcal O(z)), & \text{in Cases II and V,} \\
\kappa_2 z^{-\frac{1}{12}(1+2\nu)}\left(1+\mathcal O\left(z^{1/2}\right)\right), & \text{in Case III;}
\end{cases} \label{eq: asymptotics f2} \\
f_3(z) &= \begin{cases}
\kappa_3 e^{-\frac{i \pi}{12}(1+2\nu)}z^{-\frac{1}{12}(1+2\nu)}
\left(1+\mathcal O\left(z^{1/2}\right)\right), & \text{in Cases I and II,} \\
\kappa_2 e^{\frac{\pi i}{6}(1+2\nu)}z^{-\frac{1}{12}(1+2\nu)}
\left(1+\mathcal O\left(z^{1/2}\right)\right), & \text{in Case III,} \\
\kappa_4 e^{\frac{\pi
i}{6}(1+2\nu)}z^{-\frac{1}{3}(1+2\nu)}(1+\mathcal O(z)), & \text{in
Cases IV and V,}
\end{cases} \label{eq: asymptotics f3} \\
f_4(z) &= \begin{cases}
\kappa_3 e^{\frac{i \pi}{12}(1+2\nu)}z^{-\frac{1}{12}(1+2\nu)}
\left(1+\mathcal O\left(z^{1/2}\right)\right), & \text{in Cases I and II,} \\
\kappa_5 e^{-\frac{\pi
i}{6}(1+2\nu)}z^{\frac{1}{6}(1+2\nu)}(1+\mathcal O(z)), & \text{in
Cases III, IV, and V,}
\end{cases}\label{eq: asymptotics f4}
\end{align}
as $z\to 0$ with $\Im z>0$. The behavior in the lower half plane can be
obtained using the symmetry condition $f(\overline z)=\overline{f(z)}$. The
$\kappa_j$, $j=0,\ldots,5$ are real constants.
\end{enumerate}
\end{bvp}
Note that (c) is not an extra restriction as it is implied by (a).

We will settle the solvability of this boundary value problem at the end of this section.

\subsubsection*{Transforming the global parametrix: $N_\nu \mapsto M_\nu$}
Assuming the solvability of the boundary value problem for the Szeg\H o function $f$ we can further reduce the RH problem for $N_\nu$, i.e. we look for a solution to RH problem \ref{rhp: Nnu} in the form
\begin{equation}\label{globalpar:afterSzego}
N_\nu(z)= C M_\nu(z) \diag\left(\frac{1}{f_1(z)},\frac{1}{f_2(z)},\frac{1}{f_3(z)},\frac{1}{f_4(z)}\right),
\end{equation}
where $C$ denotes the explicit constant matrix
\begin{equation}\label{globalpar:afterSzego:C}
C = \begin{pmatrix}c & 0 & 0 & 0 \\ 0 & c_1 &c_3&c_2 \\
                   0& 0 & c_1 & 0 \\ 0& 0 & c_2 & c_1\end{pmatrix},
\end{equation}
with $c,c_1,c_2,c_3$ the constants from condition (c) in Boundary Value Problem \ref{bvp: f}. The matrix $C$ does not influence the jumps but will serve to get the appropriate asymptotics of $M_\nu(z)$ for large $z$. Putting $\wtil\alpha_j:=\alpha_j+\beta_j/n$ for $j=1,\ldots,N-1$, the matrix-valued
function $M_\nu$ must satisfy the following RH problem.

\begin{rhp} \label{rhp: M nu}
We look for a $4\times 4$ matrix-valued function $M_\nu$ satisfying the following conditions.
\begin{enumerate}
\item[\rm (1)] $M_\nu(z)$ is analytic for $z\in\cee\setminus \R $.
\item[\rm (2)] For $x\in \R $, $M_\nu$ has a jump
\begin{equation}
\left(M_\nu\right)_{+}(x) = \left(M_\nu\right)_{-}(x) \left\{
      \begin{array}{ll}
\diag\left(\left(J_{M_\nu} \right)_1(x),\left(J_{M_\nu} \right)_3(x)\right), & \hbox{for $x$ in $\R^+$,} \\
        \diag\left(-1,\left( J_{M_\nu}\right)_2(x),-1\right), & \hbox{for $x$ in $\R^-$,}
      \end{array}
    \right.
\end{equation}
where
\begin{equation*}
\left(J_{M_\nu}\right)_1=\left\{
          \begin{array}{ll}
            \begin{pmatrix}
            0 & 1 \\ -1 & 0
            \end{pmatrix}, & \hbox{on $(a_j,b_j)$, $j=1,\ldots,N$,}\\
            \diag(e^{-2\pi in\wtil\alpha_j},e^{2\pi in\wtil\alpha_j}),
            & \hbox{on $(b_j,a_{j+1})$, $j=1,\ldots,N-1$,}\\
            I_2, & \hbox{on $(0,a_1)\cup(b_N,+\infty)$,}
          \end{array}
        \right.
\end{equation*}
\begin{equation*}
\left(J_{M_\nu}\right)_2=\left\{
          \begin{array}{ll}
            I_2, & \hbox{on $(-c_2,0)$,} \\
            \begin{pmatrix} 0 & 1 \\
            -1 & 0
            \end{pmatrix}, & \hbox{on $(-\infty,-c_2)$,}\\
            \end{array}
        \right.
\end{equation*}
and
\begin{equation*}
\left(J_{M_\nu}\right)_3=\left\{
\begin{array}{ll}
          -I_2, & \hbox{on $(0,c_3)$,} \\
           \begin{pmatrix} 0 & 1 \\
           -1 & 0
           \end{pmatrix}, & \hbox{on $(c_3, +\infty)$.}
          \end{array}
        \right.
\end{equation*}
\item[\rm (3)] As $z\to \infty$ and $\pm \Im z>0$, we have
\begin{equation}
M_\nu(z)= \left[ I+ \mathcal O \left(z^{-1}\right)\right]\diag
\left(z^{1/2}, z^{1/6},z^{-1/2},z^{-1/6}\right)
\diag\left(1,A_{\pm}^{-T}\right).
\end{equation}
\end{enumerate}
\end{rhp}

Indeed, the conditions (1) and (2) in the RH problem for $M_\nu$ are immediate
from \eqref{globalpar:afterSzego}, RH problem \ref{rhp: Nnu}, and Boundary
Value Problem \ref{bvp: f}. For condition (3) one also uses
the identities \eqref{Omega:plusmin}--\eqref{Aplusmin:commute}
and \eqref{globalpar:afterSzego:C}.

\subsubsection*{Constructing the global parametrix}

We will immediately prove the solvability of RH problem \ref{rhp: M nu} which
then in combination with \eqref{transform:S to Nnu} and
\eqref{globalpar:afterSzego} finishes the construction of the global
parametrix.

We solve RH problem \ref{rhp: M nu} by reducing it to the RH problem for the
global parametrix in the non-chiral two-matrix model described in
\cite[(8.1)--(8.6)]{DKM}. The latter RH problem was solved in the same paper.
We denote that solution here as $M^{\DKM}(z)$. Note that $M^{\DKM}(z)$ depends
on certain parameters $n \in \Z$ and $0< \alpha_1^{\DKM} < \cdots <
\alpha_{N-1}^{\DKM}$ that will be specified later. We will also need the
symmetry relation \begin{equation}\label{DKM:global:symmetry} M^{\DKM}(-z) =
\diag(1,-1,1,-1)M^{\DKM}(z)\diag(1,-1,1,-1),
\end{equation}
which is not hard to verify. Now we claim that,
\begin{equation}\label{eq:solution of Mnu}
M_\nu(z):= KL\diag(-z^{1/2},1,z^{-1/2},1)M^{\DKM}(z^{1/2})\diag(-1,-1,\mp
1,\pm 1), \qquad \pm\Im z>0,
\end{equation}
solves RH problem \ref{rhp: M nu}. Here $K$ is a constant matrix of
the form
\begin{equation}\label{Kmatrix}
K = \begin{pmatrix} 1 & * & * & * \\
0 & 1 & * & 0 \\
0 & 0 & 1 & 0 \\
0 & 0 & * & 1
\end{pmatrix},
\end{equation}
for certain constants $*$, serving to get the correct asymptotics for
$z\to\infty$. Note that an apparent $z^{1/2}$ contribution for $z\to\infty$
vanishes due to \eqref{DKM:global:symmetry}. $L$ is a matrix of the form
\begin{equation}\label{Lmatrix}
L = I+ \kappa z^{-1}E_{3,1},
\end{equation}
for a suitable constant $\kappa$. The matrix $L$ serves to get the
correct behavior as $z\to 0$ and will be constructed in the proof of
Lemma~\ref{lemma:Sinfty:nearzero}. A straightforward verification
shows that $M_\nu$ defined in \eqref{eq:solution of Mnu} indeed
solves RH problem \ref{rhp: M nu} for the right choice of parameters
$n$ and $0<\alpha_1^{\DKM}< \cdots < \alpha_{N-1}^{\DKM}$. Since the
RH problem for $M^{\DKM}$ is solvable for any choice of these
parameters, see \cite[Section 8]{DKM}, we have proved the
solvability of the RH problem for $M_\nu$.

Apart from the solvability of Boundary Value Problem \ref{bvp: f} we have now finished the construction of the global parametrix by \eqref{transform:S to Nnu}, \eqref{globalpar:afterSzego}, and \eqref{eq:solution of Mnu}. We will settle the solvability of Boundary Value Problem \ref{bvp: f} at the very end of this section but first we will discuss the behavior of the global parametrix around the origin.

\subsubsection*{Behavior of $S^{(\infty)}$ near the origin}

In the next lemma we discuss the behavior of $S^{(\infty)}$ near the origin.

\begin{lemma}\label{lemma:Sinfty:nearzero}
The constant matrix $L$ in \eqref{Lmatrix} can be chosen such that
$S^{(\infty)}$ has the following behavior near the origin $z=0$:
\begin{equation}\label{Sinfty:nearzero}
\begin{array}{ll}
S^{(\infty)}(z)\diag(z^{1/4},z^{1/4},z^{1/4},z^{1/4})=\mathcal O(1),&
\textrm{in Case~I,}\\
S^{(\infty)}(z)\diag(z^{-\nu/2},z^{(1+\nu)/2},z^{1/4},z^{1/4})=\mathcal O(1),&
\textrm{in Case~II,}\\
S^{(\infty)}(z)\diag(z^{-\nu/2},z^{1/4},z^{1/4},z^{(1+\nu)/2})=\mathcal O(1),&
\textrm{in Case~III,}\\
S^{(\infty)}(z)\diag(z^{1/4},z^{1/4},z^{-\nu/2},z^{(1+\nu)/2})=\mathcal O(1),&
\textrm{in Case~IV,}\\
S^{(\infty)}(z)\diag(z^{-\nu/2},z^{(1+\nu)/2},z^{-\nu/2},z^{(1+\nu)/2})=\mathcal
O(1),& \textrm{in Case~V.}
\end{array}
\end{equation}
\end{lemma}

\begin{proof}
By combining the above transformations \eqref{transform:S to Nnu},
\eqref{globalpar:afterSzego} and \eqref{eq:solution of Mnu} we get
\begin{multline}\label{Sinfty:combined}
S^{(\infty)}(z) =
CKL\diag(-z^{1/2},1,z^{-1/2},1)M^{\DKM}(z^{1/2})\diag(-1,-1,\mp 1,\pm 1)\\
\times\diag\left(f_1^{-1}(z),f_2^{-1}(z),\sigma^{\pm 1}
f_3^{-1}(z),\sigma^{\mp 1}f_4^{-1}(z)\right)
\diag\left(z^{\frac{\nu-1}{2}},z^{\frac{1-\nu}{6}}I_3\right),
\end{multline}
for $\pm ~ \Im z>0$.

From the construction in \cite[Sec.~8]{DKM} and the symmetry
\eqref{DKM:global:symmetry} it follows that
\begin{equation}\label{DKM:global:zero:1}
M^{\DKM}(z) = A+\mathcal O(z),
\end{equation}
as $z\to 0$ in the first quadrant of $\cee$, where
\begin{equation}\label{DKM:global:zero:2}
A = \left\{
\begin{array}{ll}
\begin{pmatrix}a&a&*&* \\
b&-b&*&* \\
c&c&*&* \\
d&-d&*&*
\end{pmatrix},& \qquad \textrm{Cases~I, IV},\\
\begin{pmatrix}a&*&*&* \\
0&*&*&* \\
c&*&*&* \\
0&*&*&*
\end{pmatrix},& \qquad \textrm{Cases~II, III, V},\\
\end{array}
\right.
\end{equation}
where $a,b,c,d$ are certain constants with $abcd\neq 0$. We then
define the matrix $L$ as in \eqref{Lmatrix} with the constant
$\kappa$ given by $\kappa=c/a$. The lemma then follows from a
straightforward calculation using
\eqref{Sinfty:combined}--\eqref{DKM:global:zero:2} and the behavior
of the Szeg\H{o} functions $f_i$ near the origin in \eqref{eq:
asymptotics f1}--\eqref{eq: asymptotics f4}.
\end{proof}

\begin{remark} By using a finer analysis of the
structure in \eqref{DKM:global:zero:2}, one can show that each of the terms
$z^{(1+\nu)/2}$ in the Cases~II--V in \eqref{Sinfty:nearzero} can be replaced
by $z^{\nu/2}$.
\end{remark}

\subsubsection*{Solvability of Boundary Value Problem \ref{bvp: f}}
\label{subsection:Szego:exists}

It remains to prove the existence of the Szeg\H o function $f$ as a solution of Boundary Value Problem \ref{bvp: f}. On the Riemann surface
$\mathcal R$ we construct a canonical homology basis consisting of closed
curves $A_1,\ldots,A_g$ and $B_1,\ldots,B_g$ with $g$ the genus of $\mathcal
R$. Here $B_j$ is a closed curve on the first sheet going counterclockwise
around the union of cuts $\bigcup_{k=1}^j [a_k,b_k]$, $j=1,\ldots,g$. On the
other hand, $A_j$ is a closed curve on the first two sheets whose intersection
with the first sheet is a path connecting a point of the cut $(a_j,b_j)$ to a
point of the cut $(a_{j+1},b_{j+1})$, and whose intersection with the second
sheet is the complex conjugate of this path, with the reverse orientation. See Figure \ref{fig: homology basis} for an illustration in Case V.

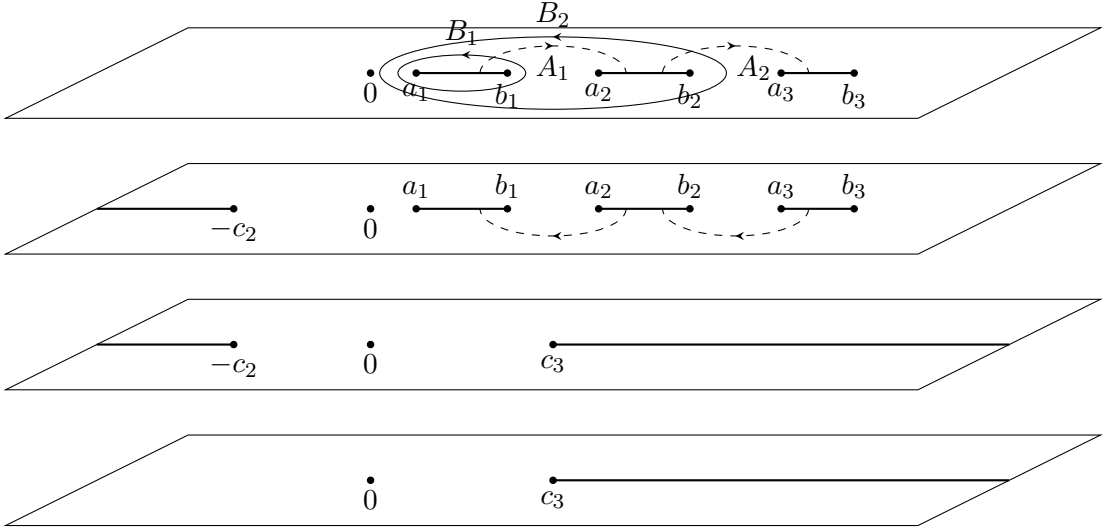
\begin{figure}
\begin{tikzpicture}[scale=1.2]
\draw (-4,-0.5)--(6,-0.5)--(8,0.5)--(-2,0.5)--cycle
      (-4,-2)--(6,-2)--(8,-1)--(-2,-1)--cycle
      (-4,-3.5)--(6,-3.5)--(8,-2.5)--(-2,-2.5)--cycle
      (-4,-5)--(6,-5)--(8,-4)--(-2,-4)--cycle;
\filldraw (0,0) circle (1pt) node[below]{$0$}
          (0,-1.5) circle (1pt) node[below]{$0$}
          (0,-3) circle (1pt) node[below]{$0$}
          (0,-4.5) circle (1pt) node[below]{$0$}
          (0.5,0)circle (1pt)(1.5,0)circle (1pt)(2.5,0)circle (1pt)(3.5,0)circle (1pt)(4.5,0)circle (1pt)(5.3,0)circle (1pt)(0.5,-1.5)circle (1pt)(1.5,-1.5)circle (1pt)(2.5,-1.5)circle (1pt)(3.5,-1.5)circle (1pt)(4.5,-1.5)circle (1pt)(5.3,-1.5)circle (1pt)
          (-1.5,-1.5) circle (1pt)(-1.5,-3)circle (1pt) (2,-3)circle (1pt)  (2,-4.5)circle (1pt) ;
\draw[thick] (0.5,0)node[below]{$a_1$}--(1.5,0)node[below]{$b_1$} (2.5,0)node[below]{$a_2$}--(3.5,0)node[below]{$b_2$} (4.5,0)node[below]{$a_3$}--(5.3,0)node[below]{$b_3$}
             (0.5,-1.5)node[above]{$a_1$}--(1.5,-1.5)node[above]{$b_1$} (2.5,-1.5)node[above]{$a_2$}--(3.5,-1.5)node[above]{$b_2$} (4.5,-1.5)node[above]{$a_3$}--(5.3,-1.5)node[above]{$b_3$};
\draw[thick] (-3,-1.5)--(-1.5,-1.5) node[below]{$-c_2$}
             (-3,-3)--(-1.5,-3) node[below]{$-c_2$};
\draw[thick] (2,-3) node[below]{$c_3$}--(7,-3)
             (2,-4.5) node[below]{$c_3$}--(7,-4.5);
\begin{scope}[decoration={markings,mark= at position 0.25 with {\arrow{stealth}}}]
\draw[postaction={decorate}] (1,0) ellipse(0.7 and 0.2);
\draw[postaction={decorate}] (2,0) ellipse(1.9 and 0.4);
\draw (1,0.2) node[above]{$B_1$} (2,0.4) node [above]{$B_2$};
\end{scope}
\begin{scope}[decoration={markings,mark= at position 0.5 with {\arrow{stealth}}}]
\draw[dashed,postaction={decorate}] (1.2,0) arc(180:0:0.8 and 0.3);
\draw[dashed,postaction={decorate}] (3.2,0) arc(180:0:0.8 and 0.3);
\draw[dashed,postaction={decorate}] (2.8,-1.5) arc(0:-180:0.8 and 0.3);
\draw[dashed,postaction={decorate}] (4.8,-1.5) arc(0:-180:0.8 and 0.3);
\draw (2,0.3) node[below]{$A_1$}  (4.2,0.3) node[below]{$A_2$};
\end{scope}
\end{tikzpicture}
\caption{Plot of the Riemann surface $\mathcal R$ and the canonical homology basis $(A_1,\ldots,A_g,B_1,\ldots,B_g)$ for Case V (i.e. $a_1,c_2,c_3>0$) and genus $g=2$.}
\label{fig: homology basis}
\end{figure}

Let $\ud w_1(z),\ldots,\ud w_g(z)$ be a basis for the space of holomorphic
differentials on $\mathcal R$ normalized with respect to the cycles
$A_1,\ldots,A_g$, i.e., such that
$$ \int_{A_j} \ud w_k(t) = \delta_{jk}, \qquad j,k=1,\ldots,g,
$$
where $\delta_{jk}$ denotes the Kronecker delta.

Let $P$ be a point on the Riemann surface $\err$ which lies in the upper half
plane of one of the sheets $\err_j$, $j=1,\ldots,4$. There exists a
sufficiently small neighborhood $U$ of $P$, lying entirely in the upper half
plane of $\err_j$, such that the holomorphic differential $\ud w_k$ allows the
representation
\begin{equation} \label{eq: merodiff local}
\ud w_k(z) = \rho_k(z)\ \ud z,\qquad z\in U,
\end{equation}
for a certain analytic function $\rho_k$, where we use $z$ as the complex
coordinate on the sheet $\err_j$. Letting $\bar U$ be the set of complex
conjugate points of $U$ lying on the same sheet $\err_j$, we then have from the
symmetry under complex conjugation that
\begin{equation} \label{eq: merodiff local bar}
 \ud w_k(z) = \overline{\rho_k(\bar z)}\ \ud z,\qquad z\in \bar U,
\end{equation}
where we again use $z\in \bar U$ as the complex coordinate on the sheet
$\err_j$.

\smallskip
Now we construct the Szeg\H o function $f$ satisfying (a)--(d) above
following Zverovich \cite{Zverovich}. In the language of
\cite[Page~135]{Zverovich} we must find a solution to
\emph{Riemann's homogeneous problem} where $\Phi(p)$ is our function
$f(z)$, $L$ is our contour $\mathcal C$, and with the divisors
$\mathcal D$ and $\mathcal J$ prescribing the singularities given by
$\mathcal D=1$ and (using additive notation)
\[ \displaystyle
\mathcal J^{-1}= \begin{cases}
\min\{\frac{-1+2\nu}{4},\frac{5-2\nu}{12}\} 0_{1,2}+\frac{-1-2\nu}{12} 0_{3,4},&\quad\textrm{Case~I},\\
\frac{1+2\nu}{6} 0_{2}+\frac{-1-2\nu}{12} 0_{3,4},&\quad\textrm{Case~II},\\
\frac{-1-2\nu}{12} 0_{2,3}
+\frac{1+2\nu}{6} 0_{4},&\quad\textrm{Case~III},\\
\min\{\frac{-1+2\nu}{4},\frac{5-2\nu}{12}\} 0_{1,2}+
\frac{-1-2\nu}{3}
0_{3}+\frac{1+2\nu}{6} 0_{4},&\quad\textrm{Case~IV},\\
\frac{1+2\nu}{6}
0_{2}+\frac{-1-2\nu}{3}
0_{3}+\frac{1+2\nu}{6} 0_{4},&\quad\textrm{Case~V}.
\end{cases}\]
Thus the only singularities or poles of $f$ are allowed at
the points lying over the origin. The precise form of the divisor $\mathcal J$
is due to \eqref{eq: asymptotics f1}--\eqref{eq: asymptotics f4}.

Denote by $G(t)$ the multiplicative factors appearing in the jump
conditions in part (b) of the above boundary value problem for $f$.
Note that $G(t)$ satisfies a H\"older condition on each analytic arc
of the contour $\mathcal C$, except possibly at the origin of the
first two sheets, in the case where $a_1=0$. Thus the behavior at
the origin needs to be analyzed separately.

Following \cite[Page~137]{Zverovich} we introduce the piecewise analytic function $X(q)$, $q \in \mathcal R$, as
\begin{equation} \label{eq: X(q)}
X(q):=e^{\frac{1}{2\pi i}\int_{\mathcal C}\ln G(\tau) \ud \hat \omega_{qq_0}(\tau)}.
\end{equation}
Here $\ud \hat \omega_{qq_0}(t)$ is the discontinuous analogue to the Cauchy kernel, $q_0 \not \in \mathcal C$. We will show that $X(q)$ solves Boundary Value Problem \ref{bvp: f} for well-chosen values of the constants $\beta_j$, $j=1,\ldots,g$.

\begin{lemma} \label{lemma: solution bvp}
There exist $\beta_j \in \R$, $j=1,\ldots,g$, such that the function $X(q)$ defined in \eqref{eq: X(q)} solves Boundary Value Problem \ref{bvp: f}.
\end{lemma}

\begin{proof}
We first prove that $X(q)$ satisfies condition (b) of the boundary value problem. The function $X(q)$ is analytic on $\mathcal R \setminus \left( \mathcal C \cup \bigcup_{k=1}^g A_k \right)$ and has the following jumps, see \cite{Zverovich}
\begin{align}
X_+(t) &= X_-(t)G(t), && \text{for }t \in \mathcal C, \\
X_+(t) &= X_-(t)e^{-\int_{\mathcal C}\ln G(\tau) \ud w_k(\tau)}, && \text{for }t \in A_k, \quad k=1,\ldots,g.
\end{align}
Note that the jump on $A_k$, $k=1,\ldots,g$, is constant. We claim that we can define the real constants $\beta_j$ such that these jumps are actually trivial. To see this we first show that
\begin{equation} \label{Jacobi:inversion}
\int_{\mathcal C}\ln G(\tau) \ud w_k(\tau), \qquad k=1,\ldots,g,
\end{equation}
is real. This follows from the symmetry under complex conjugation and the particular
form of the jump factors $G(z)$. Indeed, using \eqref{eq: jump f12 szego} and \eqref{eq: merodiff local}--\eqref{eq: merodiff local bar} the contribution of the cut $[a_j,b_j]$ on \eqref{Jacobi:inversion} can be written as
\begin{multline*}
\frac{1}{2\pi i}\int_{[a_j,b_j]} \frac{2\nu-2}{3}\log (x+) \ud w_k(x+)-\frac{1}{2\pi i}\int_{[a_j,b_j]} \frac{2\nu-2}{3}\log (x-) \ud w_k(x-) \\ =\frac{1}{2\pi i} \frac{2\nu-2}{3}\int_{a_j}^{b_j}\log (x)(\rho_k(x)-\overline{\rho_k(x)})\ud x \in \R,
\end{multline*}
where the first (second) integral is over the interval $[a_j,b_j]$ in the upper (lower) half plane of the first sheet. Also all contributions
\[
\frac{1}{2\pi i}\int_{K} \log G(x) \ud w_k(x),
\]
are real, where $K$ is any of the intervals $(b_j,a_j)$, $j=1,\ldots, N-1$, on the
first or second sheet, $(-c_2,0)$ on the second or third sheet, or $(0,c_3)$ on
the third or fourth sheet. This follows from the fact that on these intervals
the meromorphic differential $\ud w_k$ is real and the particular form of $\log
G(x)$ on these intervals. Moreover, in \eqref{Jacobi:inversion}, the contribution from the gaps $(b_j,a_{j+1})$ on the
first two sheets is given by (use \eqref{constant:beta1}--\eqref{constant:beta2})
\[
\frac{1}{2\pi i} \int_{A_k} (-2\pi i\beta_j)\ \ud w_k(t) = -\beta_j \delta_{k,j}.
\]
Hence, the constants $\beta_j$, $j=1,\ldots,g$, in
\eqref{constant:beta1}--\eqref{constant:beta2} can indeed be chosen so that \eqref{Jacobi:inversion} is zero. Therefore the function $X(q)$ already satisfies condition (b) of Boundary Value Problem \ref{bvp: f}.

Next we prove condition (d). It follows from the precise definition
of $G(z)$ that the quantities $\mathcal\chi_k$, $\mathcal\chi$ in
\cite[Page~138]{Zverovich} are all zero except at the points lying
over the origin. Hence the function $X(q)$ can only have
singularities at these points. The order of these singularities is
described by the divisor $\mathcal E$ in \cite[Page~138]{Zverovich}.
In our case this divisor is precisely the inverse of $\mathcal J$
above: $\mathcal E = \mathcal J^{-1}$. Indeed, this follows from the
formulas in \cite{Zverovich} except for the point $0_{1,2}$ in Cases
I and IV, since there the jump matrices do not satisfy the
boundedness and/or H\"older conditions. In Cases I and IV we have
that $a_1=0$ and $c_2>0$, so in the union of disks $\mathcal B$ the
Szeg\H o function $f$ has the jumps
\begin{align}
f_{1,\pm}(x) &= f_{2,\mp}(x) x^{\frac{2\nu-2}{3}}, && x>0,  \label{eq: jump szego 1}\\
f_{2,+}(x)   &= f_{2,-}(x) e^{\frac{\pi i}{3}(1+2\nu)}, && x<0. \label{eq: jump szego 2}
\end{align}
We define the conformal map
\[
\zeta= \begin{cases}
z^{1/2} & z \in \mathcal R_1 \cap \mathcal B, \\
-z^{1/2} & z \in \mathcal R_2 \cap \mathcal B,
\end{cases}
\]
where the cut of the fractional power is chosen along the positive real line,
i.e. $0<\arg z < 2 \pi$. This map sends $\overline{(\mathcal R_1 \cup \mathcal
R_2)} \cap \mathcal B$ to a disk $\mathcal B^*$ centered at the origin of the complex
$\zeta$-plane. Lifting the Szeg\H o function to the $\zeta$-plane we get
\[
f(\zeta)=\begin{cases} f_1(\zeta^2) & \Im \zeta >0, \\
f_2(\zeta^2) & \Im \zeta<0. \end{cases}
\]
Then \eqref{eq: jump szego 1}--\eqref{eq: jump szego 2} translate into
\begin{align*}
f_+(\zeta) &= f_-(\zeta)|\zeta|^\frac{4\nu-4}{3}, && \zeta \in \R \cap \mathcal B^*, \\
f_+(\zeta) &= f_-(\zeta)e^{\frac{\pi i}{3}(1+2\nu)}, && \zeta \in i\R^- \cap
\mathcal B^*,
\end{align*}
where the part of the real line is oriented from left to right and the part of the imaginary axis from bottom to top. Applying $\log$ at both sides of the equation leads to an additive problem that can be solved using Cauchy transforms and the Sokhotski-Plemelj formula
\begin{equation} \label{eq: szego log f}
\log f(\zeta)=\frac{2\nu-2}{3\pi i} \int_{-1}^1 \frac{\log |t|}{t-\zeta} \ud
t+\frac{1+2\nu}{6} \int_{-i}^0\frac{1}{t-\zeta}\ud t+h(\zeta), \qquad \zeta \in
\mathcal B^* \setminus (\R \cup i\R^-),
\end{equation}
where $h$ is meromorphic in $\mathcal B^*$ with no poles except possibly at
zero. The asymptotic behavior of the integrals in \eqref{eq: szego log f} as
$\zeta \to 0$ can be calculated, e.g. using Mathematica, and is given by
\begin{align*}
\int_{-1}^1 \frac{\log |t|}{t-\zeta} \ud t &= \pm \pi i \log \zeta + \frac{\pi^2}{2}+\mathcal O(\zeta), \qquad \text{as }\zeta \to 0, \zeta \in \C^{\pm} \\
\int_{-i}^0\frac{1}{t-\zeta}\ud t          &= \begin{cases} \log \zeta - \frac{\pi i}{2}+\mathcal O(\zeta), & \text{as }\zeta \to 0, \zeta \in I \cup II \cup IV,\\
                                                            \log \zeta + \frac{3\pi i}{2}+\mathcal O(\zeta) & \text{as }\zeta \to 0, \zeta \in III.\end{cases}
\end{align*}
Here $I,II,III,IV$ denote the open quadrants of the complex $\zeta$-plane.
Plugging in these asymptotics in \eqref{eq: szego log f} we get
\[
\log f(\zeta) = \begin{cases}
\frac12(2\nu-1)\log \zeta+\frac{i\pi}{4}(1-2\nu)+h(\zeta)+\mathcal O(\zeta), & \text{as }\zeta \to 0,\zeta \in I \cup II, \\
\frac16(5-2\nu)\log \zeta+\frac{i\pi}{12}(7+2\nu)+h(\zeta)+\mathcal O(\zeta), & \text{as }\zeta \to 0,\zeta \in III, \\
\frac16(5-2\nu)\log \zeta+\frac{i\pi}{4}(1-2\nu)+h(\zeta)+\mathcal O(\zeta), & \text{as }\zeta \to 0,\zeta \in IV.
\end{cases}
\]
In terms of the original functions $f_1,f_2$ (assuming that $h$ is
analytic) this precisely turns into \eqref{eq: asymptotics
f1}--\eqref{eq: asymptotics f2} for Cases I and IV. This concludes
the proof of the asymptotic behavior of the Szeg\H o function around
the origin.

Summarizing, we have now proved that $f(q):=X(q)$ solves Boundary
Value Problem \ref{bvp: f}, where in (a) we only established the
upper bound $|f(q)|<C_2<+\infty$ on $\mathcal R \setminus \mathcal
B$. To obtain the lower bound $0<C_1<|f(q)|$ on $\mathcal R
\setminus \mathcal B$, let us first define the function $\til f(q)$
on $\mathcal R$ which is the solution to the same boundary value
problem as $f(q)$ except that each of the multiplicative jumps
$G(t)$ is replaced by its inverse $1/G(t)$ and the asymptotic
behavior at 0 is inverted as well. Proceeding as above we find such
a solution $\til f(q)$ which is bounded, i.e., $|\til f(q)|<\til
C_2<+\infty$ on $\mathcal R \setminus \mathcal B$. But then the
product function $f(q)\til f(q)$ is analytic and bounded on
$\mathcal R$ so by Liouville's theorem it must be a constant.
Clearly it cannot be identically zero so it is a nonzero constant
$C\neq 0$. This shows that $\til f(q) = C/f(q)$ and from the
boundedness of $\til f$ we then obtain the desired lower bound
$0<C_1<|f(q)|$ with $C_1=C/\til C_2$.
\end{proof}

\subsection{Local parametrices near the nonzero branch points} \label{sec: local par nonzero}

Near each of the branch points in $\{a_j,b_j~|~j=1,\ldots,N\}\cup\{-c_2,c_3\}\setminus\{0\}$, a local parametrix
$S^{\Airy}$ can be built in the standard way with the help of Airy functions, see e.g.~\cite{DKM}. We omit the details here.

\subsection{Local parametrix near the origin} \label{subsection:local:Bessel}

In this section we construct the local parametrix $S^{(0)}$ near the origin. We will show that the local RH problem can be reduced to the RH problem in \cite[Section 5.6]{DKRZ}. As a first step we perform a preliminary transformation on RH problem \ref{RHP:S}.

We define the $4 \times 4$ matrix-valued function $P$ by
\begin{equation}\label{def:P}
P(z)= S(-z)\diag\left(e^{\pm\nu \pi i/2},e^{\mp\nu \pi
i/2},e^{\pm\nu \pi i/2},e^{\mp\nu \pi i/2}\right)J, \qquad\hbox{for
$\pm\Im z>0$,}
\end{equation}
where
\begin{equation}\label{def:J}
J=
\begin{pmatrix}
0 & 0 & 0 & 1  \\
0 & 0 & 1 & 0  \\
0 & 1 & 0 & 0  \\
1 & 0 & 0 & 0
\end{pmatrix}.
\end{equation}

Let us denote by $\wtil L_i=-L_i$, $i=1,2,3$, the lens around
$\bigcup_{j=1}^{N}(-b_j, -a_j)$, $(c_2,+\infty)$ and $( -\infty,-c_3)$,
respectively. We also introduce
\begin{equation} \label{eq: lambda tilde}
\wtil \lam_i(z)=\lam_{5-i}(-z), \qquad i=1,2,3,4.
\end{equation}
Then $P$ solves the following RH problem.

\begin{rhp}The matrix-valued function $P$ satisfies the following conditions.
\begin{enumerate}
\item[\rm (1)] $P(z)$ is analytic for $z\in\cee\setminus\Sigma_P$,
where $\Sigma_P$ is the contour consisting of the real axis and the lips of the
lenses $\wtil L_i$, $i=1,2,3$.
\item[\rm (2)] For $z\in \Sigma_P$, $P$ has a jump
\begin{equation*}
P_{+}(z) = P_{-}(z) \left\{
      \begin{array}{ll}
       \diag\left(1,\left( J_P \right)_2,1\right),
       & \hbox{for $z$ in $\R^+$ and the lips of $\wtil L_2$,} \\
       \diag\left( \left( J_P \right)_1,\left( J_P \right)_3\right),
       & \hbox{for $z$ in $\R^-$ and the lips of $\wtil L_1$, $\wtil L_3$,}
      \end{array}
    \right.
\end{equation*}
where
\begin{equation*}
(J_P)_1=\left\{
\begin{array}{ll}
          \begin{pmatrix} e^{-\nu \pi i} & 0 \\
 -e^{n(\wtil\lam_{1,-}-\wtil\lam_{2,+})} & e^{\nu \pi i}
\end{pmatrix}, & \hbox{on $(-c_3,0)$,} \\
            \begin{pmatrix} 0 & 1 \\
-1 & 0
\end{pmatrix}, & \hbox{on $(-\infty,-c_3)$,}\\
I_2, & \hbox{on the lips of $\wtil L_1$,}\\
\begin{pmatrix}
1 & -e^{\pm \nu\pi i}e^{n(\wtil\lam_2-\wtil\lam_1)} \\
0 & 1
\end{pmatrix}, & \hbox{on the upper/lower lip of $\wtil L_3$,}
          \end{array}
        \right.
\end{equation*}
\begin{equation*}
(J_P)_2=\left\{
          \begin{array}{ll}
            \begin{pmatrix}
1 & e^{n(\wtil\lam_{3,-}-\wtil\lam_{2,+})} \\ 0 &1
\end{pmatrix}, & \hbox{on $(0,c_2)$,} \\
            \begin{pmatrix} 0 & 1 \\
                             -1 & 0
            \end{pmatrix}, & \hbox{on $(c_2,\infty)$,}\\
\begin{pmatrix}
1 & 0 \\
e^{n(\wtil\lam_2-\wtil\lam_3)} & 1
\end{pmatrix}, & \hbox{on the lips of $\wtil L_2$,}
          \end{array}
        \right.
\end{equation*}
and
\begin{equation*}
(J_P)_3=\left\{
          \begin{array}{ll}
            \begin{pmatrix}
            0 & 1 \\ -1 & 0
            \end{pmatrix}, & \hbox{on $(-b_j,-a_j)$, $j=1,\ldots,N$,}\\
             \begin{pmatrix} e^{-\nu \pi i-2\pi in \alpha_j}& 0
            \\ -e^{n(\wtil\lam_{3,-}-\wtil\lam_{4,+})} & e^{\nu\pi i+2\pi in \alpha_j}
            \end{pmatrix},& \hbox{on $(-a_{j+1},-b_j)$, $j=0,\ldots,N$,}\\
            \begin{pmatrix}
            1 & -e^{\pm\nu\pi i}e^{n(\wtil \lam_4-\wtil \lam_3)} \\
            0 & 1
            \end{pmatrix}, & \hbox{on the upper/lower lip of $\wtil L_1$,}\\
            I_2, & \hbox{on the lips of $\wtil L_3$.}
          \end{array}
        \right.
\end{equation*}
\item[\rm (3)]  If $z\to 0$ outside the lenses that end in $0$,
we have
\begin{equation}\label{eq:zero behavior of P}
\left\{
\begin{array}{ll}
P(z)\diag(|z|^{\nu/2},|z|^{-\nu/2},|z|^{\nu/2},|z|^{-\nu/2})=\mathcal O(1),& \hbox{if $\nu>0$,} \\
P(z)\diag((\log|z|)^{-1},1,(\log|z|)^{-1},1)=\mathcal O(1),& \hbox{if $\nu=0$,}\\
         P(z)=\mathcal O(|z|^{\nu/2}),\qquad P^{-1}(z)=\mathcal O(|z|^{\nu/2}),& \hbox{if $-1<\nu<0$.}
       \end{array}
     \right.
\end{equation}
\end{enumerate}
\end{rhp}

We can apply the transformation \eqref{def:P} to the global parametrix
$S^{(\infty)}$ as well. That is, we define $S^{(\infty)}\mapsto P^{(\infty)}$ by
\begin{equation}\label{Pinfty}
P^{(\infty)}(z)= S^{(\infty)}(-z)\diag\left(e^{\pm\nu \pi
i/2},e^{\mp\nu \pi i/2},e^{\pm\nu \pi i/2},e^{\mp\nu \pi
i/2}\right)J, \qquad\hbox{for $\pm\Im z>0$,}
\end{equation}
with $J$ as in \eqref{def:P}. Then the jumps for $P^{(\infty)}$
equal the jumps for $P$ above, but with all the exponentially
decaying entries of the form $e^{n(\widetilde \lam_j-\widetilde
\lam_k)}$ removed. The behavior of $P^{(\infty)}$ near the origin
follows trivially from \eqref{Sinfty:nearzero}.


Now we observe that the jumps in the RH problems for $P$ and $P^{(\infty)}$ are
reminiscent of those in \cite[Section 5.6]{DKRZ}, with the variable
$\alpha$ in the latter paper playing the role of our $\nu$. Our construction of the local parametrix $P^{(0)}$ will be inspired by \cite{DKRZ}.

In the construction we have to make a case distinction between
$\nu<0$ and $\nu\geq0$. To that end we define $\mathbf 1_{\nu<0}=1$
if $\nu<0$ and $0$ if $\nu \geq 0$. The local parametrix $P^{(0)}$
is defined in a fixed but sufficiently small disk $D(0,\delta)$
around the origin, with radius $\delta>0$. It satisfies the
following RH problem.

\begin{rhp}We look for a $4 \times 4$ matrix-valued function $P^{(0)}(z):D(0,\delta)\setminus\Sigma_P \to \C^{4 \times 4}$ satisfying the following conditions.
\begin{enumerate}\label{rhp:p0}
\item[\rm (1)] $P^{(0)}(z)$ is analytic for $z\in D(0,\delta)\setminus\Sigma_P$.
\item[\rm (2)] For $z\in D(0,\delta)\cap\Sigma_P$, $P^{(0)}$ has a jump
\begin{equation*}
P^{(0)}_{+}(z) = P^{(0)}_{-}(z) \begin{cases}
  \diag\left( \left(J_{P^{(0)}} \right)_1,\left( J_{P^{(0)}} \right)_3\right),
    & \textrm{\parbox[t]{0.3\textwidth}{for $z$ in $(-\delta,0)$ and the lips of $\wtil L_1$, $\wtil L_3$,}} \\
  \diag\left(1,\left( J_{P^{(0)}} \right)_2,1\right),
    & \textrm{\parbox[t]{0.3\textwidth}{for $z$ in $(0,\delta)$ and the lips of $\wtil L_2$,}}
\end{cases}
\end{equation*}
where
\begin{equation*}
(J_{P^{(0)}})_1=\left\{
\begin{array}{ll}
          \begin{pmatrix} e^{-\nu \pi i} & 0 \\
 -\mathbf 1_{\nu<0}e^{n(\wtil\lam_{1,-}-\wtil\lam_{2,+})} & e^{\nu \pi i}
\end{pmatrix}, & \hbox{on $(-\delta,0)$ in Cases~III,IV,V,} \\
            \begin{pmatrix} 0 & 1 \\
-1 & 0
\end{pmatrix}, & \hbox{on $(-\delta,0)$ in Cases~I,II,}\\
I_2, & \hbox{on $D(0,\delta)\cap$ the lips of $\wtil L_1$,}\\
\begin{pmatrix}
1 & -e^{\pm \nu\pi i}e^{n(\wtil\lam_2-\wtil\lam_1)} \\
0 & 1
\end{pmatrix}, & \hbox{on $D(0,\delta)\cap$ the upper/lower lip of $\wtil L_3$,}
          \end{array}
        \right.
\end{equation*}
\begin{equation*}
(J_{P^{(0)}})_2=\left\{
          \begin{array}{ll}
            \begin{pmatrix}
1 & \mathbf 1_{\nu<0}e^{n(\wtil\lam_{3,-}-\wtil\lam_{2,+})} \\
0 &1
\end{pmatrix}, & \hbox{on $(0,\delta)$ in Cases~I,II,IV,V,} \\
            \begin{pmatrix} 0 & 1 \\
                             -1 & 0
            \end{pmatrix}, & \hbox{on $(0,\delta)$ in Case~III,}\\
\begin{pmatrix}
1 & 0 \\
e^{n(\wtil\lam_2-\wtil\lam_3)} & 1
\end{pmatrix}, & \hbox{on $D(0,\delta)\cap$ the lips of $\wtil L_2$,}
          \end{array}
        \right.
\end{equation*}
and
\begin{equation*}
(J_{P^{(0)}})_3=\left\{
          \begin{array}{ll}
            \begin{pmatrix}
            0 & 1 \\ -1 & 0
            \end{pmatrix}, & \hbox{on $(-\delta,0)$ in Cases~I,IV,}\\
             \begin{pmatrix} e^{-\nu \pi i}& 0
            \\ -\mathbf 1_{\nu<0}e^{n(\wtil\lam_{3,-}-\wtil\lam_{4,+})} & e^{\nu\pi i}
            \end{pmatrix},& \hbox{on $(-\delta,0)$ in Cases~II,III,V,}\\
            \begin{pmatrix}
            1 & -e^{\pm\nu\pi i}e^{n(\wtil \lam_4-\wtil \lam_3)} \\
            0 & 1
            \end{pmatrix}, & \hbox{on $D(0,\delta)\cap$ the upper/lower lip of $\wtil L_1$,}\\
            I_2, & \hbox{on $D(0,\delta)\cap$ the lips of $\wtil L_3$.}
          \end{array}
        \right.
\end{equation*}
\item[\rm (3)]  If $z\to 0$ outside the lenses that end in $0$,
we have
\begin{equation}\label{eq:zero behavior of P0}
\left\{
\begin{array}{ll}
P^{(0)}(z)\diag(|z|^{\nu/2},|z|^{-\nu/2},|z|^{\nu/2},|z|^{-\nu/2})=\mathcal O(1),& \hbox{if $\nu>0$,} \\
P^{(0)}(z)\diag((\log|z|)^{-1},1,(\log|z|)^{-1},1)=\mathcal O(1),& \hbox{if $\nu=0$,} \\
         P^{(0)}(z)=\mathcal O(|z|^{\nu/2}), \qquad \left(P^{(0)}\right)^{-1}(z)=\mathcal O(|z|^{\nu/2}),& \hbox{if $-1<\nu<0$.}
       \end{array}
     \right.
\end{equation}
\item[\rm (4)] On the boundary of $D(0,\delta)$ we have the uniform estimate
\begin{equation}\label{eq:matching condition P0}
 P^{(0)}(z) = P^{(\infty)}(z)(I+\mathcal O(1/n)),\qquad n\to\infty.
\end{equation}
\end{enumerate}
\end{rhp}

Note that the exponentially small entries in the jump matrices on
$(-\delta,0)\cup(0,\delta)$ in the above RH problem are only present
if $\nu<0$; we neglect them if $\nu\geq 0$. The reason for this case
distinction between $\nu<0$ and $\nu\geq0$ is explained in
\cite{DKRZ} (with there $\alpha$ playing the role of our $\nu$); see
also the estimates in Section~\ref{subsection:finaltransfo} below.

To construct $P^{(0)}$, we need the model RH problem for the modified Bessel function.
\begin{rhp}\label{rhp:model RHP for Bessel}
Denoting with $\gamma_j$, $j=1,2,3$ the complex rays
$\{\zeta\in\cee\mid\arg\zeta=(j+1)\pi/3\}$, we look for a $2\times 2$
matrix-valued function $\Psi^{\Bessel}$ such that
\begin{enumerate}
\item[\rm (1)] $\Psi^{\Bessel}$ is analytic in $\cee\setminus \bigcup_{j=1}^3
\gamma_j$.
\item[(2)] With the rays $\gamma_j$, $j=1,2,3$ all oriented towards the origin,
$\Psi^{\Bessel}$  has the jumps \begin{equation*} \Psi^{\Bessel}_+ =
\Psi^{\Bessel}_- \times\left\{\begin{array}{ll}\begin{pmatrix} 1 & 0 \\
e^{\nu \pi i} & 1
\end{pmatrix}, &\text{ on } \gamma_1, \\
\begin{pmatrix} 0 & 1 \\ -1 & 0
\end{pmatrix}, &\text{ on } \gamma_2, \\
\begin{pmatrix} 1 & 0 \\ e^{-\nu \pi i} & 1
\end{pmatrix}, &\text{ on } \gamma_3.
\end{array}\right.
\end{equation*}
\item[\rm (3)] Uniformly for $\zeta\to\infty$ we have
\begin{equation}\label{Psi Bessel:asy} \Psi^{\Bessel}(\zeta)
= (2\pi\zeta^{1/2})^{-\sigma_3/2}
\left(\frac{1}{\sqrt{2}}\begin{pmatrix}1 & i \\ i &
1\end{pmatrix}+\mathcal
O(\zeta^{-1/2})\right)e^{2\zeta^{1/2}\sigma_3},
\end{equation}
with $\sigma_3=\diag (1,-1)$.
\item[\rm (4)] As $\zeta\to 0$ in $|\arg\zeta|<2\pi/3$ we have
\begin{equation}\label{Hankel:zero2}
\Psi^{\Bessel}(\zeta) = \left\{\begin{array}{ll}
\mathcal O\begin{pmatrix}\zeta^{\nu/2} & \zeta^{-\nu/2}\\
\zeta^{\nu/2} & \zeta^{-\nu/2}\end{pmatrix}, &\text{if } \nu>0, \\
\mathcal O\begin{pmatrix}
1 & \log|\zeta| \\
1 & \log|\zeta|
\end{pmatrix}, &\text{if } \nu=0,\\
\mathcal O(\zeta^{\nu/2}), &\text{if } \nu<0.
\end{array}\right.
\end{equation}
\end{enumerate}
\end{rhp}
This RH problem has u unique solution which is given in terms of modified Bessel and Hankel functions, see \cite{KMVV2004}.

We are now ready to construct $P^{(0)}$ case by case, following the
lines in \cite{DKRZ}. First we give the construction for
\textbf{Case I}. We may assume without loss of generality that the
lips of $\wtil L_1$ and $\wtil L_3$ coincide within $D(0,\delta)$.
If $\nu\geq 0$, we consider the functions
\begin{equation}
\phi_1(z):=\left(\wtil\lam_1(z)-\wtil\lam_2(z)\pm \frac{2\pi i}{3}\right)^2,
\quad \phi_3(z):=\left(\wtil\lam_3(z)-\wtil\lam_4(z)\pm 2\pi i \right)^2, \quad
\pm\Im z>0.
\end{equation}
It follows from \eqref{eq: lambda tilde} and Lemma \ref{lemma: lambda near 0}
that these functions have analytic continuations to $D(0,\delta)$ that give
conformal maps from a neighborhood of the origin onto itself, such that
$\phi_i(x)$, $i=1,3$, is real and positive for $x\in(0,\delta)$. We deform the
lips of $\wtil L_1$ ($\wtil L_3$) near $0$ such that $\phi_i$ maps the upper
and lower lips of $\wtil L_1$ ($\wtil L_3$) to the rays with angles $2\pi/3$
and $-2\pi/3$, respectively.

We now define
\begin{multline}\label{Q first part Case I}
\what P^{(0)}(z) =E(z)\diag\left(
\sigma_1\Psi^{\Bessel}\left(\frac{n^2\phi_1(z)}{16}\right)\sigma_1,\sigma_1\Psi^{\Bessel}\left(\frac{n^2\phi_3(z)}{16}\right)\sigma_1\right)
\\
\times\diag(\sigma_3 e^{\frac{n}{2}(\wtil\lam_1(z)-\wtil\lam_2(z))\sigma_3},
\sigma_3 e^{\frac{n}{2}(\wtil\lam_3(z)-\wtil\lam_4(z))\sigma_3}),
\end{multline}
where $\sigma_1=\begin{pmatrix} 0 & 1 \\
1 & 0
\end{pmatrix}$ and $\sigma_3=\diag(1,-1)$, and where the prefactor $E(z)$ is analytic in $D(0,\delta)$
and is chosen to satisfy the matching condition on $\partial D(0,\delta)$,
see below. We use the hat superscript to emphasize that we are in the
situation
$\nu\geq 0$. Note that $\what P^{(0)}$ essentially decouples into two blocks
containing the model RH problem for the modified Bessel function. With this
definition, and assuming $n \equiv 0 \mod 3$, the items $(1)$, $(2)$, and $(3)$
in the RH problem for $P^{(0)}$ are satisfied, by virtue of items $(1)$, $(2)$,
and $(4)$ in the RH problem for $\Psi^{\Bessel}$.

To achieve the matching condition in item $(4)$ of the RH problem
for $P^{(0)}(z)$, we take $E(z)$ in \eqref{Q first part Case I} as
\begin{multline}\label{E case II}
E(z) =P^{(\infty)}(z)(-1)^n\diag(\sigma_3,\sigma_3)
\diag\left(\frac{1}{\sqrt{2}}\begin{pmatrix}1 & -i \\ -i & 1\end{pmatrix},
\frac{1}{\sqrt{2}}\begin{pmatrix}1 & -i \\ -i & 1\end{pmatrix}\right)  \\
\times \diag\left( \left(\tfrac{\pi n}{2}\phi_1^{1/2}(z)\right)^{-\sigma_3/2},
\left(\tfrac{\pi n}{2}\phi_3^{1/2}(z)\right)^{-\sigma_3/2}\right).
\end{multline}
Obviously $E(z)$ is analytic for $z\in D(0,\delta)\setminus\er^-$. Moreover,
one checks that $E(z)$ is also analytic across $(-\delta,0)$. Finally, since
$\phi_{1,3}(z)^{1/4}=\mathcal O(z^{1/4})$ as $z\to 0$ it then follows from
\eqref{Sinfty:nearzero} and \eqref{Pinfty} that
\begin{equation}
E(z)=\mathcal O(z^{-1/2}), \qquad z\to 0,
\end{equation}
so $E(z)$ cannot have a pole at zero. We conclude that $E(z)$ is analytic in
the disk $D(0,\delta)$. By virtue of \eqref{Psi Bessel:asy}, the matching
condition \eqref{eq:matching condition P0} in the RH problem for $P^{(0)}$ is
satisfied.

If $-1<\nu<0$, we cannot simply ignore the jumps on the real axis.
The local parametrix $P^{(0)}$ is then constructed in the following
form
\begin{multline}\label{def of P case I}
P^{(0)}(z)=\what P^{(0)}(z)
\diag(e^{-n\wtil \lambda_1(z)},e^{-n\wtil\lambda_2(z)},
      e^{-n\wtil \lambda_3(z)},e^{-n\wtil\lambda_4(z)})Q(z)\\
\times \diag(e^{n\wtil \lambda_1(z)},e^{n\wtil \lambda_2(z)},
             e^{n\wtil \lambda_3(z)},e^{n\wtil \lambda_4(z)}),
\end{multline}
where $\what P^{(0)}$ is the parametrix for the case $\nu\geq 0$ given in
\eqref{Q first part Case I}, and $Q(z)$ is a piecewise constant matrix. More
precisely, following the idea in \cite[Section 5.6.3]{DKRZ}, we have
\begin{equation}\label{eq:S1expli}
Q(z) = I-\frac{e^{-i\pi\nu}}{2i\sin(\pi\nu)}E_{2,3},
\end{equation}
for $z$ in the region bounded by $(0,\delta)$ and the upper lip of $\wtil L_1$,
\begin{equation}
Q(z) = I-\frac{e^{i\pi \nu}}{2i\sin(\pi\nu)}E_{2,3},
\end{equation}
for $z$ in the region bounded by $(0,\delta)$ and the lower lip of $\wtil L_1$,
\begin{equation}
Q(z) = \begin{pmatrix}
1&0&\frac{1}{2i\sin(\pi\nu)}&\frac{e^{\nu\pi i}}{2i\sin(\pi\nu)}\\
0&1&-\frac{e^{-\nu \pi i}}{2i\sin(\pi\nu)}&-\frac{1}{2i\sin(\pi\nu)} \\
0&0&1&0\\ 0&0&0&1 \end{pmatrix},
\end{equation}
for $z$ in the region bounded by $(-\delta,0)$ and the upper lip of
$\wtil L_1$,
\begin{equation}\label{eq:S4}
Q(z) = \begin{pmatrix}
1&0&-\frac{1}{2i\sin(\pi\nu)}&\frac{e^{-\nu\pi i}}{2i\sin(\pi\nu)}\\
0&1&-\frac{e^{\nu \pi i}}{2i\sin(\pi\nu)}&\frac{1}{2i\sin(\pi\nu)} \\
0&0&1&0\\ 0&0&0&1 \end{pmatrix},
\end{equation}
for $z$ in the region bounded by $(-\delta,0)$ and the lower lip of
$\wtil L_1$. With $Q$ given in \eqref{eq:S1expli}--\eqref{eq:S4},
one can check that $P^{(0)}$ defined in \eqref{def of P case I}
indeed satisfies the items (1)--(4) in RH problem \ref{rhp:p0} if
$-1<\nu<0$. For the jump condition (2) we use that $n \equiv 0 \mod
3$. For the matching condition (4) we also need the inequalities
$\Re\wtil \lambda_{3,4}(z)<\Re\wtil\lambda_{1,2}(z)$ for $z$ in a
neighborhood of the origin, see Lemma~\ref{lem:inequality of lamda
functions} and \eqref{eq:lam1 and lam2}--\eqref{eq:lam 3 and lam 4}.
Moreover, this construction actually works as long as
$\sin(\pi\nu)\neq 0$, i.e., $\nu \notin \mathbb{N}\cup\{0\}$.

For \textbf{Case II}, the RH problem for $P^{(0)}$ is exactly the
same as the one considered in \cite[Section 5.6]{DKRZ}. Also the
construction of $P^{(0)}$ in \textbf{Case IV} is similar to Case II.
We thus omit the details for these cases.

For \textbf{Case III}, the Bessel parametrix will appear in the middle block.
More precisely, by setting
\begin{equation}
\phi_2(z):=\left(\wtil\lam_3(z)-\wtil\lam_2(z)\pm \frac{4\pi
i}{3}\right)^2=-c_2^2 z + \mathcal O(z^2), \qquad \pm ~\Im z >0,
\end{equation}
we have
\begin{multline}\label{case III nu>0}
\what P^{(0)}(z)=E(z)\diag\left(
z^{-\nu/2},\Psi^{\Bessel}\left(\frac{n^2\phi_2(z)}{16}\right),z^{\nu/2}\right)
\\
\times \diag\left(1,\sigma_3 e^{\pm \frac{\nu \pi
i}{2}\sigma_3}e^{\frac{n}{2}(\wtil\lam_2(z)-\wtil\lam_3(z))\sigma_3},1 \right),
\qquad \pm ~\Im z>0,
\end{multline}
and the analytic prefactor $E$ is given by
\begin{multline}\label{E case III}
E(z)=P^{(\infty)}(z)
\diag \left(1,\sigma_3 e^{\mp \frac{\nu \pi i}{2}\sigma_3},1\right)\\ \times \diag \left(z^{\nu/2}, \frac{1}{\sqrt{2}}\begin{pmatrix} 1 & -i \\
-i & 1
\end{pmatrix}\left(\frac{n\pi}{2}\phi_2^{1/2}\right)^{\sigma_3/2},
z^{-\nu/2} \right),  \qquad \pm ~\Im z>0.
\end{multline}
This describes the parametrix if $\nu\geq 0$. If $-1<\nu<0$, we
define $P^{(0)}(z)$ by \eqref{def of P case I}, \eqref{case III
nu>0}, where now $Q(z)$ is a piecewise constant matrix given by
\begin{equation}\label{S region 1 Case III}
Q(z)=I+\frac{i}{2\sin(\pi \nu)}(E_{2,1}+E_{4,3}),
\end{equation}
for $z$ in the region outside the lens,
\begin{equation}
Q(z)=\begin{pmatrix} 1 & 0 & 0 & 0 \\
\frac{i}{2\sin(\pi \nu)} & 1 & 0 & 0 \\
\frac{i}{2\sin(\pi \nu)}
&0 & 1 & 0\\
0 & -\frac{i}{2\sin(\pi \nu)} & \frac{i}{2\sin(\pi \nu)} & 1
\end{pmatrix},
\end{equation}
for $z$ in the region bounded by $(0,\delta)$ and the upper lip of
$\wtil L_2$, and
\begin{equation}\label{S region 3 Case III}
Q(z)=\begin{pmatrix} 1 & 0 & 0 & 0 \\
\frac{i}{2\sin(\pi \nu)} & 1 & 0 & 0 \\
-\frac{i}{2\sin(\pi \nu)}
&0 & 1 & 0\\
0 & \frac{i}{2\sin(\pi \nu)} & \frac{i}{2\sin(\pi \nu)} & 1
\end{pmatrix},
\end{equation}
for $z$ in the region bounded by $(0,\delta)$ and the lower lip of
$\wtil L_2$. To check the matching condition (4) in RH
problem~\ref{rhp:p0}  we need the inequalities $\Re\wtil
\lambda_1(z)<\Re\wtil\lambda_{2,3}(z)<\Re\wtil\lambda_4(z)$ for $z$
in a neighborhood of the origin, see Lemma~\ref{lem:inequality of
lamda functions} and \eqref{eq:relations lam2 and lam3}.

Finally we build the local parametrix in \textbf{Case~V}. The
construction is much simpler in this case. We set
\begin{equation}\label{case V nu>0}
\what P^{(0)}(z)=P^{(\infty)}(z).\end{equation} This describes the parametrix
if $\nu\geq 0$. If $-1<\nu<0$, we define $P^{(0)}(z)$ by \eqref{def of P case
I}, \eqref{case V nu>0}, where now
\begin{equation}\label{case V nu<0 bis}
Q(z) = \begin{pmatrix} 1 & 0 &0 &0 \\
\frac{i}{2\sin(\pi\nu)} & 1 & \frac{i e^{\mp\nu\pi i}}{2\sin(\pi\nu)}&0 \\
0&0&1&0 \\
0&0&\frac{i}{2\sin(\pi\nu)} & 1\end{pmatrix},
\end{equation}
if $\pm\Im z>0$. One easily checks that conditions (1)--(4) in RH
problem~\ref{rhp:p0} are satisfied.

Finally, by tracing back the transformation \eqref{def:P}, we obtain the local
approximation $S^{(0)}$ for $S$ from $P^{(0)}$.

\subsection{Final transformation}\label{section:final}
\label{subsection:finaltransfo}

Using the local parametrices described in Sections \ref{sec: local par nonzero}
and \ref{subsection:local:Bessel}, and the global
parametrix $S^{(\infty)}$, we define the final transformation as follows
\begin{equation}
R(z)=\left\{
       \begin{array}{ll}
         S(z)(S^{\Airy}(z))^{-1}, & \text{in the disks around } \{a_j,b_j,c_1,-c_2,c_3\}\setminus\{0\}, \\
         S(z)(S^{(0)}(z))^{-1}, & \text{in the disk } D(0,\delta)
         \text{ around the origin,} \\
        S(z)(S^{(\infty)}(z))^{-1}, & \text{elsewhere,}
       \end{array}
     \right.
\end{equation}
where $S$, $S^{\Airy}$, $S^{(0)}$, and $S^{(\infty)}$ are respectively defined in \eqref{def:TtoS 1}; Section \ref{sec: local par nonzero}; \eqref{def:P}, \eqref{Q first part Case I}, \eqref{def of P case I}, \eqref{case III nu>0}, and \eqref{case V nu>0}; and \eqref{Sinfty:combined}.

From our construction of the parametrices, it follows that $R$ satisfies the
following RH problem.
\begin{enumerate}
\item[(1)] $R$ is analytic in $\mathbb{C} \setminus \Sigma_R$, where the contour $\Sigma_R$ depends on
whether $\nu<0$ or $\nu\geq 0$ and is different for all five cases.

\item[(2)] $R$ has jumps $R_+=R_-J_R$ on $\Sigma_R$ that satisfy
\begin{equation}
J_R(z)=I+\mathcal O(1/n),
\end{equation}
uniformly for $z$ on the boundaries of the disks;
\begin{equation} \label{eq: estimate R nu positive}
J_R(x)=I+\mathcal O(x^{\nu} e^{-cn}),
\end{equation}
on $(-\delta,0)$ (Case I), or on $( 0,\delta)$ (Case III), or on $(-\delta,
0)\cup(0,\delta)$ (Cases II,~IV~and~V), for some constant $c>0$, if $\nu \geq
0$; and
\begin{equation}
J_R(z)=I+\mathcal O(e^{-cn|z|}),
\end{equation}
on the other parts of $\Sigma_R$, for some constant $c>0$.
\item[(3)] $R(z)=I+\mathcal O(1/z)$ as $z\to \infty$.
\end{enumerate}

In case $\nu \geq 0$, the estimate \eqref{eq: estimate R nu positive} is not
trivial. We prove it in Case I.  In that case $R(z)$ is clearly analytic in
$D(0,\delta) \setminus (-\delta,0]$ with the following jump on $(-\delta,0)$
\begin{align*}
R_-(x)^{-1}R_+(x) &= P^{(0)}_+(-x)J_P(-x)^{-1}\Pnot_-(-x)^{-1}, \\
                 &= P^{(0)}_+(-x)\left(I-e^{n(\wtil \lambda_{3,-}-\wtil \lambda_{2,+})(-x)}E_{2,3}\right)\Pnot_-(-x)^{-1}, \\
                 &= I-e^{n(\wtil \lambda_{3,-}-\wtil
                 \lambda_{2,+})(-x)}P^{(0)}(-x)E_{2,3}\Pnot(-x)^{-1}.
\end{align*}
Now for $\nu \geq 0$ the matrix $P^{(0)}(-x)E_{2,3}\Pnot(-x)^{-1}= \mathcal O
(|x|^\nu)$ as $x \to 0$, $x<0$. Indeed (following the proof of \cite[Lemma
5.6]{DKRZ}) we observe that
\[
P^{(0)}(-x)E_{2,3}\Pnot(-x)^{-1}=P^{(0)}(-x)
\begin{pmatrix} 0&1&0&0
\end{pmatrix}^T \begin{pmatrix} 0&0&1&0 \end{pmatrix}\Pnot(-x)^{-1}.
\]
Then if $\nu >0$ we find from \eqref{eq:zero behavior of P0} and the
fact that $\det P^{(0)}(z) \equiv 1$ that both factors
\[
P^{(0)}(-x)\begin{pmatrix} 0&1&0&0 \end{pmatrix}^T \qquad \text{and} \qquad \begin{pmatrix}
0&0&1&0 \end{pmatrix}\Pnot(-x)^{-1}
\]
behave like $\mathcal O(|x|^{\nu/2})$ as $x \to 0$, $x<0$, which
proves the statement for $\nu>0$. In case $\nu=0$, this approach
does not work as it would lead to a bound $\mathcal O(\log |x|)$.
However, we do find that the first factor $\Pnot(-x)\begin{pmatrix}
0&1&0&0
\end{pmatrix}^T$ remains bounded as $x \to 0$, $x<0$. For the second factor
we observe from \eqref{Hankel:zero2} and $\det \Psi^{\mathrm{Bessel}}(\zeta)
\equiv 1$ that
\[
\left( \Psi^{\mathrm{Bessel}}\right)^{-1}(\zeta)=\begin{pmatrix} \mathcal
O(\log |\zeta|) & \mathcal O(\log |\zeta|)
 \\ \mathcal
O(1) & \mathcal O(1)\end{pmatrix}, \qquad \text{as $\zeta \to 0$}.
\]
Combined with \eqref{Q first part Case I} and the fact that $E(z)$ is bounded near the origin we obtain that $\begin{pmatrix} 0&0&1&0 \end{pmatrix}\Pnot(-x)^{-1}$ is also bounded as $x \to 0$, $x<0$, which proves the statement in case $\nu=0$. Cases II, III, IV, and V can be similarly treated.

Then from standard arguments we may conclude that
\begin{equation}\label{large n behavior of R}
R(z)=I+\mathcal O \left(\frac{1}{n(|z|+1)}\right),
\end{equation}
as $n\to\infty$, uniformly for $z$ in the complex plane outside of $\Sigma_R$.

\subsection{Proof of Theorem \ref{th: noncrit limit}}\label{sec: proof noncrit theorem}

We follow the approach in \cite[Section 9.3]{DKM} and start the proof with a lemma.

\begin{lemma}\label{lemma: limit S}
For every $x \in S(\mu_1) \setminus \bigcup_{k=1}^N (D(a_k,\epsilon)
\cup D(b_k,\epsilon) )$ we have
\[
S_+^{-1}(y)S_+(x)= \begin{pmatrix}
I_2+\mathcal O(x-y) & * \\ *&*
\end{pmatrix}, \qquad \text{as $y \to x$,}
\]
uniformly in $n$. The $*$ entries denote unimportant $2 \times 2$ blocks.
\end{lemma}
\begin{proof}
The proof is standard, see e.g. \cite[Lemma 9.9]{DKM}.
\end{proof}

Recall the formula \eqref{eq: kernel in Y} that expresses $K_n$ in
terms of $Y$. The idea is then to write this expression in terms of
$R$ instead of $Y$ by applying all respective transformations
\[
X \mapsto U \mapsto T \mapsto S  \mapsto R,
\]
introduced in the steepest descent analysis. Meanwhile, we let $n$ tend to infinity, so that we can exploit the conclusion of the steepest descent analysis \eqref{large n behavior of R}.

Let $x,y>0$. First, unfolding the transformation $Y \mapsto X$ given in
\eqref{eq: def X} we get
\begin{multline*}
K_n(x,y)=\frac{y^{\nu/2}x^{-\nu/2}}{2\pi i (x-y)}\begin{pmatrix} 0 & w_{0,n}(y)  & w_{1,n}(y) & w_{2,n}(y) \end{pmatrix}\\
\times \diag \left(1, D(z)^{-1}W_{n,+}^{-T}(y)e^{n\Theta_+(y)} \right)
X_+^{-1}(y)X_+(x) \begin{pmatrix} 1&0&0&0 \end{pmatrix}^T.
\end{multline*}
Using \eqref{Wronskian:bis} this boils down to
\begin{equation} \label{eq: kernel in X}
K_n(x,y)=\frac{y^{\nu/2}x^{-\nu/2}e^{n(\theta_{1,+}(y)-V(y))}}{2\pi i
(x-y)}\begin{pmatrix} 0 & 1 & 0 & 0 \end{pmatrix} X_+^{-1}(y)X_+(x)
\begin{pmatrix} 1&0&0&0 \end{pmatrix}^T.
\end{equation}
Next, by the transformation $X \mapsto U$ described by \eqref{eq:def of U} we obtain
\[
e^{n(V(y)-V(x))}K_n(x,y)=y^{\nu/2}x^{-\nu/2}\frac{e^{n(\lambda_{2,+}(y)-\lambda_{1,+}(x))}}{2\pi
i (x-y)}\begin{pmatrix} 0 & 1 & 0 & 0 \end{pmatrix} U_+^{-1}(y)U_+(x)
\begin{pmatrix} 1&0&0&0 \end{pmatrix}^T.
\]

The opening of global lenses $U \mapsto T$ in \eqref{def:UtoT 1} does not
effect the above expression. The opening of the local lens $T \mapsto S$ in
\eqref{def:TtoS 1}, however, does have impact
\begin{multline*}
e^{n(V(y)-V(x))}K_n(x,y) =\frac{y^{\nu/2}x^{-\nu/2}}{2\pi i
(x-y)}\\ \times \begin{pmatrix} -\chi_{S(\mu_1)}(y)e^{n\lambda_{1,+}(y)} &
e^{n\lambda_{2,+}(y)} & 0 & 0 \end{pmatrix} S_+^{-1}(y)S_+(x)
\begin{pmatrix}
e^{-n\lambda_{1,+}(x)}& \chi_{S(\mu_1)}(x)e^{-n\lambda_{2,+}(x)}&0&0
\end{pmatrix}^T,
\end{multline*}
where $\chi_{S(\mu_1)}$ denotes the characteristic function of the
set $S(\mu_1)$.

Let $x \in S(\mu_1) \setminus \{a_k,b_k \mid k=1,\ldots,N\}$. The factors
$e^{n(V(y)-V(x))}$ and $y^{\nu/2}x^{-\nu/2}$ disappear as $y \to x$. Then Lemma
\ref{lemma: limit S} yields
\begin{align*}
K_n(x,x) &=\lim_{y \to x} \frac{e^{n(\lambda_{2,+}(y)-\lambda_{2,+}(x))}-e^{n(\lambda_{1,+}(y)-\lambda_{1,+}(x))}}{2\pi i(x-y)}+\mathcal O(1), \\
         &=\frac{n}{2 \pi i} \frac{\mathrm d}{\mathrm d x}(\lambda_{1,+}(x)-\lambda_{2,+}(x))+\mathcal O(1).
\end{align*}
Hence,
\begin{align*}
\lim_{n \to \infty} \frac{1}{n}K_n(x,x) &= \frac{1}{2 \pi i}  \frac{\mathrm d}{\mathrm d x}(\lambda_{1,+}(x)-\lambda_{2,+}(x)) \\
                                        &= \frac{1}{ \pi i}  \frac{\mathrm d}{\mathrm d x}(\lambda^\DKM_{1,+}(\sqrt x)-\lambda^\DKM_{2,+}(\sqrt x)) \\
                                        &= \frac{\rho_1^\DKM (\sqrt x)}{\sqrt x} \\
                                        &= \rho_1(x),
\end{align*}
where the second equality follows from \eqref{eq:def lambdafunction}, the third is taken from \cite[p.~116]{DKM}, and the last one is \eqref{equil:problems:relation}.

\section{Triple scaling limit in the quadratic/linear case}
\label{section:triple:scaling}

In this part we study the very concrete case of the chiral two-matrix model with potentials
\[
V(x)=x, \qquad W(y)=\frac{y^2}{2}+\alpha y.
\]
For this case we were able to construct a phase diagram in the
$(\alpha,\tau)$-plane, see Figure \ref{fig: phase diagram}. Very remarkable is
the occurrence of a multi-critical point for the parameter values
$\alpha=-1,\tau=1$. In this part we will study a triple scaling limit to this
point leading to the chiral version of the main result in \cite{DG}.
An essential point in the proof is the construction of the local parametrix
around zero. In this construction we will make use of the solution to a model
RH problem introduced in \cite{Del3}.

The proof of Theorem \ref{th: triple scaling limit} is again based
on a a steepest descent analysis. This analysis is very similar to
the one performed for the non-critical cases. In fact the first four
transformations are almost exactly equal. The only difference is
that we will need to modify the $\lam$-functions, which will be
introduced first.

\subsection{Modified $\lam$-functions}

Let us first introduce an auxiliary parameter $\gamma$ that is completely
determined by $\alpha$ and $\tau$ but will prove to be convenient for notation.
We define $\gamma=\gamma(\alpha,\tau)$ as the solution of
\begin{equation} \label{eq: gamma}
\alpha \tau^{2/3}=\frac{3}{\gamma}-9\gamma^2+5\tau^{4/3}\gamma,
\end{equation}
that tends to 1 as $\tau \to 1$ and $\alpha \to -1$. In the triple scaling
limit, i.e. we let $\alpha$ and $\tau$ depend on $n$ as in \eqref{eq: scaling
alpha tau} while $n \to \infty$, we have
\begin{equation} \label{eq: scaling gamma}
\gamma=1+\tfrac13 a n^{-1/3} + \left(  \tfrac{11}{144}a^2 + \tfrac{47}{48}b \right) n^{-2/3} + \mathcal O (n^{-1}).
\end{equation}

\begin{lemma}\label{lemma: modified lambda functions}
There exist functions $\lam_j$, $j=1,2,3,4$, analytic on $\C \setminus \R$ that
satisfy the following conditions.
\begin{itemize}
\item[\rm (a)] As $z \to \infty$ we have
\begin{align}
\lam_1(z) &= z-\log(z)+\ell_1+\mathcal O \left( z^{-1}\right), \\
\lam_2(z) &= \theta_1(z)+\frac13 \log z + \ell_2 + C z^{-1/3}+ D z^{-2/3}+\mathcal O \left(z^{-1} \right), \\
\lam_3(z) &= \theta_2(z)+\begin{cases}
\frac13 \log z + \ell_3 + C \omega z^{-1/3}+ D \omega^2 z^{-2/3}+\mathcal O \left(z^{-1} \right) & \text{ in } \C_+,\\
\frac13 \log z + \ell_4 + C \omega^2 z^{-1/3}+ D \omega z^{-2/3}+\mathcal O \left(z^{-1} \right) & \text{ in } \C_-,
\end{cases}\\
\lam_4(z) &= \theta_3(z)+\begin{cases}
\frac13 \log z + \ell_4 + C \omega^2 z^{-1/3}+ D \omega z^{-2/3}+\mathcal O \left(z^{-1} \right) & \text{ in } \C_+,\\
\frac13 \log z + \ell_3 + C \omega z^{-1/3}+ D \omega^2 z^{-2/3}+\mathcal O \left(z^{-1} \right) & \text{ in } \C_-,
\end{cases}\end{align}
where $C,D \in \R$ and $\ell_j,$ $j=1,2,3,4$, are constants satisfying
\begin{equation} \label{eq: ell crit}
\ell_3-\ell_2=\ell_2-\ell_4=\frac43 \pi i.
\end{equation}
\item[\rm (b)] There exists a constant $c>0$ such that the modified $\lam$-functions satisfy the following jump conditions.
\begin{itemize}
\item[\rm (i)] On $\R^+$ we have
\begin{align}
\lam_{1,\pm}&=\lam_{2,\mp}  && \text{on $(0,c)$,} \\
\lam_{j,+} &=\lam_{j,-}  && \text{on $(c,\infty)$}, \quad j=1,2, \\
\lam_{3,\pm} &=\lam_{4,\mp} && \text{on $(0,\infty)$.}
\end{align}
\item[\rm (ii)] On $\R^-$ we have
\begin{align}
\lam_{j,+}&=\lam_{j,-} -2\pi i, && j=1,4, \\
\lam_{2,\pm}&=\lam_{3,\mp} \pm 2\pi i.
\end{align}
\end{itemize}
\item[\rm (c)] In a neighborhood of the origin we have
\begin{align*}
\lam_1(z) &= zK(z)+\begin{cases}
F(z)z^{1/4}+G(z)z^{1/2}+H(z)z^{3/4} & \text{ in } \C_+,\\
iF(z)z^{1/4}-G(z)z^{1/2}-iH(z)z^{3/4}+2\pi i & \text{ in } \C_-,
\end{cases} \\
\lam_2(z) &= zK(z)+\begin{cases}
iF(z)z^{1/4}-G(z)z^{1/2}-iH(z)z^{3/4}+2\pi i & \text{ in } \C_+, \\
F(z)z^{1/4}+G(z)z^{1/2}+H(z)z^{3/4} & \text{ in } \C_-,
\end{cases}  \\
\lam_3(z) &= zK(z)+\begin{cases}
-iF(z)z^{1/4}-G(z)z^{1/2}+iH(z)z^{3/4}+2\pi i & \text{ in } \C_+, \\
-F(z)z^{1/4}+G(z)z^{1/2}-H(z)z^{3/4} & \text{ in } \C_-,
\end{cases} \\
\lam_4(z) &= zK(z)+\begin{cases}
-F(z)z^{1/4}+G(z)z^{1/2}-H(z)z^{3/4} & \text{ in } \C_+,\\
-iF(z)z^{1/4}-G(z)z^{1/2}+iH(z)z^{3/4}+2\pi i & \text{ in } \C_-.
\end{cases}
\end{align*}
where $F,G,H,K$ are {analytic functions} satisfying
\begin{align}
F(0) &= 4 e^{3\pi i/4} \gamma^{1/4} \left( -2 \gamma^2+\tfrac{1}{\gamma}+\tau^{4/3}\gamma \right), \label{eq: F} \\
G(0) &= 2i \gamma^{-1/2}\left( \tfrac32 \gamma^2 - \tfrac{1}{4\gamma}-\tfrac54 \tau^{4/3} \gamma \right), \label{eq: G} \\
H(0) &= 2e^{\pi i/4} \gamma^{-5/4} \left( \tfrac12 \gamma^2- \tfrac{1}{12\gamma}+\tfrac14 \tau^{4/3} \gamma \right) \label{eq: H}.
\end{align}
\end{itemize}
\end{lemma}
\begin{proof}
Define
\begin{equation}\label{eq:def modified lambdafunction}
\lam_j(z)=2\lam_j^{\DG}(\sqrt z), \qquad j=1,2,3,4, \qquad z \in \C
\setminus \R,
\end{equation}
where $\lam_j^{\DG}$ denote the modified $\lam$-functions introduced
in \cite[Section 3.4]{DG}. Then (a) follows from \cite[Lemma
3.11]{DG}. (b)(i) is direct from \cite[Lemma 3.10]{DG}. To prove
(ii) we need some symmetry conditions on $\lam_j^{\DG}$. By
\cite[(3.15),(3.23)--(3.24),(3.29)--(3.30)]{DG}, it follows that for
$x > 0$
\begin{align*}
\lam_j^{\DG}(ix) &= \lam_j^{\DG}(-ix) -\pi i, & j=1,4, \\
\lam_{j,\mp}^{\DG}(ix)&=\lam_{j,\pm}^{\DG}(-ix)+ \pi i, & j=2,3.
\end{align*}
This, together with \eqref{eq:def modified lambdafunction}, implies (b)(ii), where we take $c=(c^{\DG})^2$. Finally (c) follows from \cite[Lemma 3.12]{DG} where we put
\[
\begin{cases}
F(z)=2F^\DG(\sqrt z), &G(z)=2G^\DG(\sqrt z), \\
H(z)=2H^\DG(\sqrt z), &K(z)=2K^\DG(\sqrt z),
\end{cases} \qquad z \in \C \setminus \R.
\]
\end{proof}

\begin{remark}
The constants $\gamma$, $c$, etc. and the functions $F,G,H,K$,
$\lambda_j$, etc. all implicitly depend on $\alpha$ and $\tau$ (and
thus also on $n$ via the triple scaling limit). When dealing with
these functions associated with the critical values of the
parameters $\alpha=-1$, $\tau=1$, we add a star to the notation.
Thus, we write $\gamma^*$, $c^*$, $F^*,G^*,H^*,K^*$,
$\lambda_j^*\ldots$
\end{remark}

\subsection{The transformations $Y \mapsto X \mapsto U \mapsto T \mapsto S$}

The first transformations $Y \mapsto X \mapsto U \mapsto T \mapsto S $ of the steepest descent analysis in the critical case are almost the same as
before. The main difference is that we use the modified lambda functions rather than the
ordinary ones. Apart from that also some details have to be changed. We list them here.

\subsubsection*{Transformation $Y \mapsto X$.}
This transformation is exactly as in Section \ref{subsection:firsttransfo}.

\subsubsection*{Transformation $X \mapsto U$.}

Definition \ref{def: U} of $U(z)$ has to be slightly changed.
Besides the fact that we replace the $\lambda$-functions by the
modified $\lambda$-functions, we also have to change the definition
of the diagonal matrix $L$ by
\begin{equation} \label{eq: L crit}
L(z)=\begin{cases}
-\diag (\ell_1,\ell_2,\ell_3,\ell_4) & \text{for }\Im z >0, \\
-\diag (\ell_1,\ell_2,\ell_4,\ell_3) & \text{for }\Im z <0.
\end{cases}
\end{equation}
Then, under the assumption that $n \equiv 0 \mod 3$, $U$ solves RH problem \ref{rhp:U} if we put $N:=1$, $a_1:=0$, $b_1:=c$, and
$c_2=c_3:=0$.

\subsubsection*{Transformations $U \mapsto T \mapsto S$: opening of lenses.}

Here, we open the local lens $L_1$ around $[0,c]$ and unbounded lenses $L_2$ around $\R^-$ and $L_3$ around $\R^+$. We want to do this such that the off-diagonal entries of the jump matrices on the lips of these lenses tend exponentially fast to zero as $n \to \infty$. This is only possible outside a shrinking disk around the origin. It will be sufficient for our purposes to let this disk shrink with speed $n^{-1/3}$
\[
\D=\{z \in \C \mid |z|< \rho n^{-1/3}\},
\]
where $\rho>0$ is a constant that will be chosen sufficiently small later on.

\begin{lemma} \label{lemma: estimate lenses}
The lenses $L_j$, $j=1,2,3$, can be opened (independently of $n$) such that
\begin{align}
\Re (\lambda_1(z)-\lambda_2(z)) & \leq -d n^{-1/2},  && \text{\parbox[t]{.3 \textwidth}{for $z$ on the lips of $L_1$ but outside $\D$,}} \label{eq: estimate 1}\\
\Re(\lambda_3(z)-\lambda_2(z)) & \leq -d n^{-1/2} \max( 1,|z|^{2/3}), && \text{\parbox[t]{.3 \textwidth}{for $z$ on the lips of $L_2$ but outside $\D$,}} \label{eq: estimate 2}\\
\Re(\lambda_3(z)-\lambda_4(z)) & \leq -d n^{-1/2} \max(1,|z|^{2/3}), &&
\text{\parbox[t]{.3 \textwidth}{for $z$ on the lips of $L_3$ but outside
$\D$,}} \label{eq: estimate 3}
\end{align}
for sufficiently large $n$ and a fixed constant $d>0$. Moreover, there exists a constant $d'$ such that
\begin{equation} \label{eq: estimate R}
\Re( \lambda_2(x)-\lambda_1(x) ) \leq -d' x, \qquad x \in
(c^*+\epsilon,\infty),
\end{equation}
where $\epsilon>0$ is a small number.
\end{lemma}

\begin{proof}
Estimates \eqref{eq: estimate 1}--\eqref{eq: estimate 3} are
immediate from \eqref{eq:def modified lambdafunction} and
\cite[Lemmas 4.6 and 4.9]{DG} where we define $L_j=\left( L_j^\DG
\right)^2$. Estimate \eqref{eq: estimate R} follows in the same way
from the following result that holds in the context of \cite{DG} but
was not explicitly mentioned there: there exists $d>0$ such that for
sufficiently large $n$
\[
\Re(\lambda_2^\DG(x)-\lambda_1^\DG(x)) \leq -dx^2, \qquad x \in \left(-\infty,-(c^*)^\DG-\epsilon\right) \cup \left((c^*)^\DG+\epsilon,\infty\right).
\]
\end{proof}

In this way we arrive at the following RH problem for $S$.
\begin{rhp}\label{RHP:S crit}
The matrix-valued function $S$ is the unique solution of the
following RH problem
\begin{enumerate}
\item[\rm (1)] $S(z)$ is analytic for $z\in\cee\setminus\Sigma_S$,
where $\Sigma_S$ is the contour consisting of the real axis and the lips of the
lenses $L_i$, $i=1,2,3$. These lenses are chosen such that the estimates in
Lemma \ref{lemma: estimate lenses} hold.

\item[\rm (2)] For $z\in \Sigma_S$, $S$ has a jump
\begin{equation*}
S_{+}(z) = S_{-}(z) \left\{
      \begin{array}{ll}
        \diag\left(\left( J_{S}\right)_1(z),\left( J_{S}\right)_3(z)\right), &
\textrm{\parbox[t]{0.4\textwidth}{for $z$ in $\R^+$ and the lips of $L_1$, $L_3$,}}\\
        \diag\left(e^{\nu\pi i},\left( J_{S} \right)_2(z),e^{-\nu\pi i}\right), &
        \textrm{\parbox[t]{0.4\textwidth}{for $z$ in $\R^-$ and the lips of $L_2$,}}
      \end{array}
    \right.
\end{equation*}
where
\[
(J_S)_1= \begin{cases}
   \begin{pmatrix} 0 & 1 \\ -1 & 0 \end{pmatrix}, & \hbox{on $(0,c)$,}\\
   \begin{pmatrix} 1&e^{n(\lam_{2,+}-\lam_{1,-})} \\ 0&1 \end{pmatrix},& \hbox{on $(c,\infty)$,}\\
   \begin{pmatrix} 1 & 0 \\ e^{n(\lam_1-\lam_2)} & 1 \end{pmatrix}, & \hbox{on the lips of $L_1$,}\\
   I_2, & \hbox{on the lips of $L_3$,} \end{cases}
\]
\[
(J_S)_2=\begin{cases}
   \begin{pmatrix} 0 & 1 \\ -1 & 0 \end{pmatrix}, & \hbox{on $(-\infty,0)$,}\\
   \begin{pmatrix} 1 & -e^{\pm \nu\pi i}e^{n(\lam_3-\lam_2)} \\ 0 & 1 \end{pmatrix}, & \hbox{on the upper/lower lip of $L_2$,}
\end{cases}
\]
and
\[
(J_S)_3=\begin{cases}
    \begin{pmatrix} 0 & 1 \\ -1 & 0 \end{pmatrix}, & \hbox{on $(0, +\infty)$,}\\
    I_2, & \hbox{on the lips of $L_1$,}\\
    \begin{pmatrix} 1 & 0 \\ e^{n(\lam_3-\lam_4)} & 1 \end{pmatrix}, & \hbox{on the lips of $L_3$.}
\end{cases}
\]

\item[\rm (3)] As $z\to \infty$ with $\pm\Im z >0$, we have
\begin{multline*}
S(z)= \left[ I+ \mathcal O \left(z^{-1}\right)\right]\diag \left(1,
z^{1/3},z^{-1/3},1\right)\diag\left(z^{\nu/2},z^{-\nu/6}
A_{\pm}^{-T}\right)\diag \left( 1, 1,\sigma^{\pm},\sigma^{\mp 1} \right).
\end{multline*}

\item[\rm (4)]  $S(z)$ has the same behavior near the origin as $X(z)$, see
\eqref{eq:zero behavior of X}, provided that $z\to 0$ outside the
lenses that end in $0$.
\end{enumerate}
\end{rhp}

We will construct local and global parametrices for $P(z)$ as described in
Section~\ref{subsection:local:Bessel}, i.e. $P(z)$ is established from the
solution $S(z)$ to RH problem \ref{RHP:S crit} by the transformation
\eqref{def:P}. More precisely, we have

\begin{rhp}\label{RHP:critical case}
The matrix-valued function $P$ satisfies the following RH problem.
\begin{enumerate}
\item[\rm (1)] $P(z)$ is analytic for $z\in\cee\setminus\Sigma_P$,
where $\Sigma_P$ is the contour consisting of the real axis and the
lips of the lenses $\wtil L_i=-L_{i}$, $i=1,2,3$.
\item[\rm (2)]
For $z\in \Sigma_P$, $P$ has a jump
\begin{equation*}
P_{+}(z) = P_{-}(z) \left\{
      \begin{array}{ll}
       \diag\left(1,\left( J_P(z) \right)_2,1\right),
       & \textrm{\parbox[t]{0.4\textwidth}{for $z$ in $\R^+$ and the lips of $\wtil L_2$,}} \\
       \diag\left( \left( J_P(z) \right)_1,\left( J_P(z) \right)_3\right),
       & \textrm{\parbox[t]{0.4\textwidth}{for $z$ in $\R^-$ and the lips of $\wtil L_1$, $\wtil L_3$,}}
      \end{array}
    \right.
\end{equation*}
where
\begin{equation*}
(J_P)_1=\left\{
\begin{array}{ll}
\begin{pmatrix} 0 & 1 \\
-1 & 0
\end{pmatrix}, & \hbox{on $(-\infty,0)$,}\\
I_2, & \hbox{on the lips of $\wtil L_1$,}\\
\begin{pmatrix}
1 & -e^{\pm \nu\pi i}e^{n(\wtil\lam_2-\wtil\lam_1)} \\
0 & 1
\end{pmatrix}, & \hbox{on the upper/lower lip of $\wtil L_3$,}
          \end{array}
        \right.
\end{equation*}
\begin{equation*}
(J_P)_2=\begin{cases}
           \begin{pmatrix} 0 & 1 \\
                             -1 & 0
            \end{pmatrix}, & \hbox{on $(0,\infty)$,}\\
\begin{pmatrix}
1 & 0 \\
e^{n(\wtil\lam_2-\wtil\lam_3)} & 1
\end{pmatrix}, & \hbox{on the lips of $\wtil L_2$,}
              \end{cases}
\end{equation*}
and
\begin{equation*}
(J_P)_3=\begin{cases}
            \begin{pmatrix}
            0 & 1 \\ -1 & 0
            \end{pmatrix}, & \hbox{on $(-c,0)$,}\\
             \begin{pmatrix} e^{-\nu \pi i}& 0
            \\ -e^{n(\wtil\lam_{3,-}-\wtil\lam_{4,+})} & e^{\nu\pi i}
            \end{pmatrix},& \hbox{on $(-\infty,-c)$,}\\
            \begin{pmatrix}
            1 & -e^{\pm\nu\pi i+n(\wtil \lam_4-\wtil \lam_3)} \\
            0 & 1
            \end{pmatrix}, & \hbox{on the upper/lower lip of $\wtil L_1$,}\\
            I_2, & \hbox{on the lips of $\wtil L_3$,}
          \end{cases}
\end{equation*}
where
\begin{equation} \label{eq: lambda tilde crit}
\wtil \lam_i(z)=\lam_{5-i}(-z), \qquad i=1,2,3,4.
\end{equation}
\item[\rm (3)]  If $z\to 0$ outside the lenses that end in $0$,
we have
\begin{equation}\label{eq:zero behavior of Pbis}
\left\{
\begin{array}{ll}
P(z)\diag(|z|^{\nu/2},|z|^{-\nu/2},|z|^{\nu/2},|z|^{-\nu/2})=\mathcal O(1),& \hbox{if $\nu>0$,} \\
P(z)\diag((\log|z|)^{-1},1,(\log|z|)^{-1},1)=\mathcal O(1),& \hbox{if $\nu=0$,}\\
         P(z)=\mathcal O(|z|^{\nu/2}),\qquad P^{-1}(z)=\mathcal O(|z|^{\nu/2}), & \hbox{if $-1<\nu<0$.}
       \end{array}
     \right.
\end{equation}
\end{enumerate}
\end{rhp}

\subsection{Global parametrix}
By ignoring all the exponentially small terms in the RH problem for
$P$, we obtain a RH problem for $P^{(\infty)}$. This global
parametrix $P^{(\infty)}$ can be constructed by first constructing
$S^{(\infty)}$ along the lines of Section \ref{section:global} and
then again applying the transformation $S^{(\infty)} \mapsto
P^{(\infty)}$ as given by \eqref{Pinfty}. Its behavior around the
origin is given by
\begin{equation} \label{eq: asym P0 at zero crit}
\begin{cases}
P^{(\infty)}(z)=\what P_\pm z^{-3/8}+\mathcal O(z^{-1/8}), \\
\left(P^{(\infty)}\right)^{-1}(z)=\what Q_{\pm}z^{-3/8}+\mathcal
O(z^{-1/8}),\end{cases} \qquad \text{as }z\to 0, \quad \pm ~\Im z
>0,
\end{equation}
for constant matrices $\what P_\pm,\what Q_\pm$. Clearly \begin{equation}\label{PQ:orthogonal}\what
P_\pm \what Q_\pm=\what Q_\pm \what P_\pm=0.\end{equation} Furthermore, in view
of the jump conditions for $P^{(\infty)}$, it is readily seen that
\begin{align}
\what P_+&=\what P_- \diag \left( 1, \begin{pmatrix}
            0 & 1 \\ -1 & 0
            \end{pmatrix}, 1\right) \label{eq:positive x}
\\
\what P_+&=\what P_- e^{3 \pi i/ 4}\diag \left( \begin{pmatrix}
            0 & 1 \\ -1 & 0
            \end{pmatrix},
            \begin{pmatrix}
            0 & 1 \\ -1 & 0
            \end{pmatrix}\right). \label{eq:negative x}
\end{align}
Eliminating $\what P_+$ from these two formulas we get
\begin{align*}
\what P_-&=\what P_- e^{3 \pi i/ 4}
\begin{pmatrix}
            0 & 0 & -1 & 0 \\
            -1 & 0 & 0 & 0 \\
            0 & 0 & 0 & 1 \\
            0 & -1 & 0 & 0
            \end{pmatrix}.
\end{align*}
Iterating this relation we see
\begin{align}\label{eq:relation P}
\what P_-&=\what P_-
\begin{pmatrix}
            0 & 0 & 0 & i \\
            0 & 0 & -i & 0 \\
            0 & i & 0 & 0 \\
            -i & 0 & 0 & 0
            \end{pmatrix}.
\end{align}
Similarly, we have
\begin{align}
\what Q_+&= \diag \left( 1, \begin{pmatrix}
            0 & -1 \\ 1 & 0
            \end{pmatrix}, 1\right)\what Q_- \label{eq:positive x Q}
\\
\what Q_+&= e^{3 \pi i/ 4}\diag \left( \begin{pmatrix}
            0 & -1 \\ 1 & 0
            \end{pmatrix},
            \begin{pmatrix}
            0 & -1 \\ 1 & 0
            \end{pmatrix}\right)\what Q_-, \label{eq:negative x Q}
\\
\what Q_-&=
\begin{pmatrix}
            0 & 0 & 0 & -i \\
            0 & 0 & i & 0 \\
            0 & -i & 0 & 0 \\
            i & 0 & 0 & 0
            \end{pmatrix}
\what Q_-. \label{eq:relation Q}
\end{align}
These relations will be helpful later.


\subsection{Local parametrix at $c$}

As in Section \ref{sec: local par nonzero} the local parametrix $S^{(c)}$
around $c$ can be built in the standard way with the help of Airy functions,
see e.g.~\cite{DKM}. We omit the details here.

\subsection{Local parametrix at the origin}

In this section we build a local parametrix $\Pnot$ near the origin. Here the analysis is essentially different from the noncritical situation discussed in the previous section.

\subsubsection*{Transformation of the RH problem for $M(\zeta)$}

To build the local parametrix near the origin we will use a slightly modified
version of RH problem \ref{rhp:modelM}. We put $\til\nu=\nu+1/2$ and set
\begin{equation} \label{eq: def N}
N(\zeta) = \diag\left(\begin{pmatrix}0&i\\
1&0\end{pmatrix},1,i\right)M(\zeta)\diag\left(\begin{pmatrix}0&1\\-i&0\end{pmatrix},1,-i\right).
\end{equation}
The jumps for $N$ are shown in Figure~\ref{fig:modelRHP:2}.

\begin{figure}[t]
\vspace{14mm}
\begin{center}
   \setlength{\unitlength}{1truemm}
   \begin{picture}(100,70)(-5,2)
       \put(40,40){\line(1,0){40}}
       \put(40,40){\line(-1,0){40}}
       \put(40,40){\line(2,1){30}}
       \put(40,40){\line(2,-1){30}}
       \put(40,40){\line(-2,1){30}}
       \put(40,40){\line(-2,-1){30}}
       \put(40,40){\line(2,3){18.5}}
       \put(40,40){\line(2,-3){17}}
       \put(40,40){\line(-2,3){18.5}}
       \put(40,40){\line(-2,-3){17}}
       \put(40,40){\thicklines\circle*{1}}
       \put(39.3,36){$0$}
       \put(60,40){\thicklines\vector(1,0){.0001}}
       \put(20,40){\thicklines\vector(-1,0){.0001}}
       \put(60,50){\thicklines\vector(2,1){.0001}}
       \put(60,30){\thicklines\vector(2,-1){.0001}}
       \put(20,50){\thicklines\vector(-2,1){.0001}}
       \put(20,30){\thicklines\vector(-2,-1){.0001}}
       \put(50,55){\thicklines\vector(2,3){.0001}}
       \put(50,25){\thicklines\vector(2,-3){.0001}}
       \put(30,55){\thicklines\vector(-2,3){.0001}}
       \put(30,25){\thicklines\vector(-2,-3){.0001}}

       \put(60,41){$\Gamma_0$}
       \put(60,52.5){$\Gamma_1$}
       \put(47,57){$\Gamma_2$}
       \put(29,57){$\Gamma_3$}
       \put(16,52.5){$\Gamma_4$}
       \put(16,41){$\Gamma_5$}
       \put(16,30.5){$\Gamma_6$}
       \put(29,20){$\Gamma_7$}
       \put(48,20){$\Gamma_8$}
       \put(60,30){$\Gamma_{9}$}

       \put(65,45){$\small{\Omega_0}$}
       \put(58,58){$\small{\Omega_1}$}
       \put(37,62){$\small{\Omega_2}$}
       \put(17.5,58){$\small{\Omega_3}$}
       \put(10,46){$\small{\Omega_4}$}
       \put(10,34){$\small{\Omega_5}$}
       \put(17.5,22){$\small{\Omega_6}$}
       \put(38,17){$\small{\Omega_7}$}
       \put(58,22){$\small{\Omega_8}$}
       \put(65,33){$\small{\Omega_{9}}$}

       \put(78.5,40){$\small{\begin{pmatrix}1&0&0&0\\ 0&0&1&0\\ 0&-1&0&0\\ 0&0&0&1 \end{pmatrix}}$}
       \put(69,57){$\small{\begin{pmatrix}1&0&0&0\\ 0&1&0&0\\ 0&1&1&0\\ 0&0&0&1 \end{pmatrix}}$}
       \put(45,74){$\small{\begin{pmatrix}1& e^{\nu\pi i}&0&0\\ 0&1&0&0\\ 0&0&1&e^{\nu\pi i}\\ 0&0&0&1 \end{pmatrix}}$}
       \put(4,74){$\small{\begin{pmatrix}1&0&0&0\\ - e^{-\nu\pi i}&1&0&0\\ 0&0&1&0\\ 0&0& -e^{-\nu\pi i}&1 \end{pmatrix}}$}
       \put(-16,57){$\small{\begin{pmatrix}1&0&0&0\\ 0&1&0&0\\ 0&0&1&0\\ -1&0&0&1\end{pmatrix}}$}
       \put(-26,40){$\small{\begin{pmatrix}0&0&0&-1\\ 0&1&0&0\\ 0&0&1&0\\ 1&0&0&0 \end{pmatrix}}$}
       \put(-16,23){$\small{\begin{pmatrix}1&0&0&0\\ 0&1&0&0\\ 0&0&1&0\\ -1&0&0&1 \end{pmatrix}}$}
       \put(4,6){$\small{\begin{pmatrix}1&0&0&0\\ - e^{\nu\pi i}&1&0&0\\ 0&0&1&0\\ 0&0& -e^{\nu\pi i}&1 \end{pmatrix}}$}
       \put(45,6){$\small{\begin{pmatrix}1& e^{-\nu\pi i}&0&0\\ 0&1&0&0\\ 0&0&1& e^{-\nu\pi i}\\ 0&0&0&1 \end{pmatrix}}$}
       \put(69,23){$\small{\begin{pmatrix}1&0&0&0\\ 0&1&0&0\\ 0&1&1&0\\ 0&0&0&1\end{pmatrix}}$}

  \end{picture}
   \vspace{0mm}
   \caption{The figure shows the jump
   matrices
   $J_k$, $k=0,\ldots,9$, in the RH problem for  $N = N(\zeta)$.}
   \label{fig:modelRHP:2}
\end{center}
\end{figure}
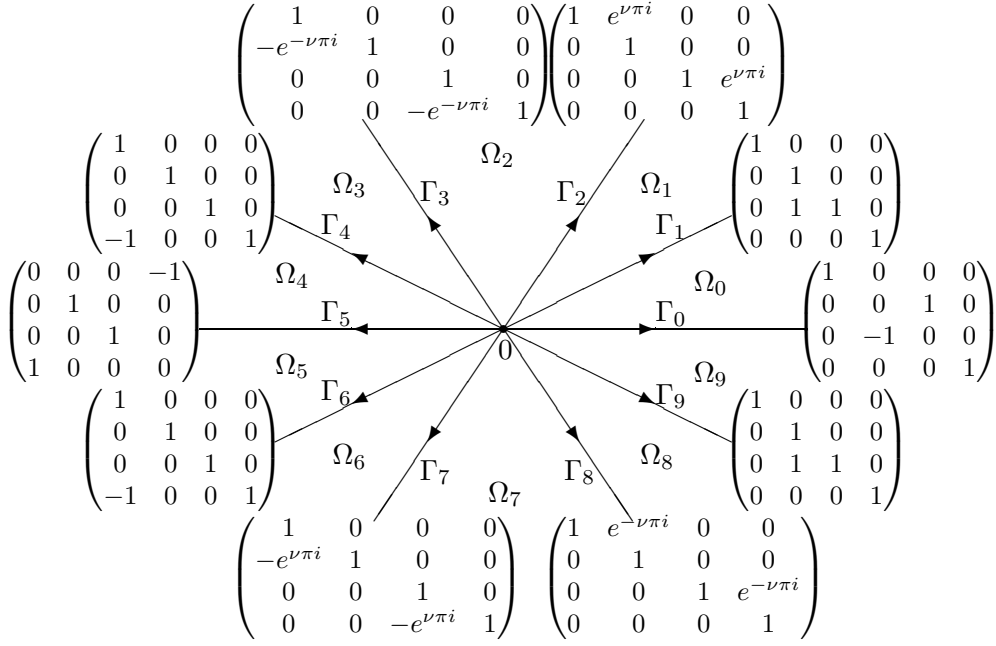

The asymptotics of $N$ as $\zeta \to \infty$ is given by
\begin{multline}
\label{N:asymptotics} N(\zeta) =
\left(I+\frac{N_1}{\zeta}+\frac{N_2}{\zeta^2}+\mathcal
O\left(\frac{1}{\zeta^3}\right)\right)
\diag(\zeta^{-1/4},(-\zeta)^{-1/4},(-\zeta)^{1/4},\zeta^{1/4})
\\
\times \wtil
A\diag\left(e^{-\psi_1(\zeta)-t\zeta},e^{-\psi_2(\zeta)+t\zeta},e^{\psi_2(\zeta)+t\zeta},e^{\psi_1(\zeta)-t\zeta}\right)
\end{multline}
with
\begin{equation}\label{mixing:matrix2}
\wtil A:=\frac{1}{\sqrt{2}}\begin{pmatrix} 1 & 0 & 0 & i \\
0 & 1 & -i & 0 \\
0 & -i & 1 & 0 \\
i & 0 & 0 & 1 \\
\end{pmatrix},
\end{equation}
and the behavior of $N(\zeta)$ for $\zeta\to 0$ is given by
\[
 N(\zeta) = \mathcal O(\zeta^{\til\nu}),\quad N^{-1}(\zeta) =
\mathcal O(\zeta^{\til\nu}), \qquad \hbox{if $\til\nu\leq 0$},
\]
and
\begin{align*} \left\{\begin{array}{ll}
N(\zeta)\diag(\zeta^{\til\nu},\zeta^{-\til\nu},\zeta^{\til\nu},\zeta^{-\til\nu})=\mathcal
O(1),& \quad
\zeta\in \Omega_1\cup\Omega_8,\\
N(\zeta)\diag(\zeta^{-\til\nu},\zeta^{\til\nu},\zeta^{-\til\nu},\zeta^{\til\nu})=\mathcal
O(1),&\quad
\zeta\in \Omega_3\cup\Omega_6,\\
N(\zeta)=\mathcal O(\zeta^{-\til\nu}),&\quad\zeta\not\in
(\Omega_1\cup\Omega_3\cup\Omega_6\cup\Omega_8),\end{array}\right.
\hbox{if $\til\nu\geq 0$}.
\end{align*}
Note that the behavior at infinity can be rewritten as
\begin{multline}
\label{N:asymptotics2} N(\zeta) =
\diag(\zeta^{-1/4},(-\zeta)^{-1/4},(-\zeta)^{1/4},\zeta^{1/4}) \wtil A
\left(I+\frac{\wtil N_{1,\pm}}{\zeta^{1/2}}+\frac{\wtil N_{2,\pm}}{\zeta}+\mathcal O\left(\frac{1}{\zeta^{3/2}}\right)\right) \\
\times
\diag\left(e^{-\psi_1(\zeta)-t\zeta},e^{-\psi_2(\zeta)+t\zeta},e^{\psi_2(\zeta)+t\zeta},e^{\psi_1(\zeta)-t\zeta}\right)
\end{multline}
as $\zeta \to \infty$ within the upper/lower half plane. Here
\begin{equation} \label{eq: def tilde N1}
\wtil N_{1,\pm}=\wtil A^{-1} \diag \left(1,e^{\mp \pi i/4},0,0 \right) N_1 \diag \left(0,0,e^{\mp \pi i/4},1 \right) \wtil A.
\end{equation}

For further use, we record the symmetry relation (see also \cite{Del3})
\begin{multline}\label{symmetry:N}
N(-\zeta;r_1,r_2,s,t) = \diag\left(\begin{pmatrix} 0 & 1 \\ -1 & 0
\end{pmatrix},\begin{pmatrix} 0 & 1 \\ -1 & 0
\end{pmatrix}\right)\\
\times N(\zeta; r_2,r_1,s,t)\diag\left(\begin{pmatrix} 0 & -1 \\ 1 & 0
\end{pmatrix},\begin{pmatrix} 0 & -1 \\ 1 & 0
\end{pmatrix}\right).
\end{multline}
Note that the order of $r_1$ and $r_2$ differs at both sides of the equality.

A corollary of this relation and \eqref{N:asymptotics2} is
\begin{multline} \label{eq: symmetry tilde N1}
\frac{\wtil N_{1,\mp}(r_1,r_2,s,t)}{(-\zeta)^{1/2}} =
\diag\left(\begin{pmatrix} 0 & 1 \\ -1 & 0
\end{pmatrix},\begin{pmatrix} 0 & 1 \\ -1 & 0
\end{pmatrix}\right)\\
\times \frac{\wtil N_{1,\pm}(
r_2,r_1,s,t)}{\zeta^{1/2}}\diag\left(\begin{pmatrix} 0 & -1 \\ 1 & 0
\end{pmatrix},\begin{pmatrix} 0 & -1 \\ 1 & 0 \end{pmatrix}\right), \quad \pm\Im \zeta >0.
\end{multline}

\subsubsection*{Construction of the local parametrix}

Here we construct the local parametrix $\Pnot$ around the origin. We will do this inside a shrinking disk $\D$, i.e. we want to
solve the following RH problem.

\begin{rhp} \label{rhpforP0} We look for $P^{(0)}$
satisfying the following conditions.
\begin{enumerate}
\item[\rm (1)] $P^{(0}(z)$ is analytic for $z\in D(0, \rho n^{-1/3}) \setminus \Sigma_P$,
where $D(0,\rho n^{-1/3})$ denotes the disk of radius $\rho
n^{-1/3}$ around $0$.
\item[\rm (2)] $P^{(0)}$ satisfies the jumps
\begin{align*}
    P^{(0)}_+ = P^{(0)}_- J_P, \qquad \text{on } \Sigma_P \cap D(0, n^{-1/3}),
\end{align*}
where $J_P$ is the jump matrix in RH problem \ref{RHP:critical case}.
  \item[\rm (3)] As $n\to\infty$, we have that
\begin{equation}\label{matching condition}
    P^{(0)}(z) =   \left(I+ Z(z) + \mathcal O(n^{-1/12}) \right) P^{(\infty)}(z),
    \quad \text{uniformly for } |z| = \rho n^{-1/3},
\end{equation}
where
\begin{equation} \label{eq:NisO1}
    Z(z) = \mathcal O(1), \qquad \text{as } n \to \infty, \quad \text{uniformly for } |z| = \rho n^{-1/3}.
\end{equation}
\end{enumerate}
\end{rhp}

We define the local parametrix as follows (compare with \cite{Del3}):
\begin{equation}\label{local:P:multicrit}
\Pnot(z) =
E_n(z)N\left(n^{2/3}f(z);r_1(z),r_2(z),n^{2/3}s(z),n^{1/3}t(z)\right)\Lam(z),
\end{equation}
where
$$ \Lam(z)=\diag\left(e^{n\wtil\lam_{1}(z)},\ldots,e^{n\wtil\lam_{4}(z)}\right), $$
and, as before, $\wtil \lambda_j(z)=\lambda_{5-j}(-z)$, $j=1,2,3,4$.
The parameters $r_1$, $r_2$, $s$, and $t$ depend  in a mild way on $z$ and also on $n$, although this is not indicated in the notation. The prefactor $E_n(z)$ does depend on $n$ and is analytic in a neighborhood of the origin except on the negative real line where it has a jump.

We set
\begin{equation} \label{eq: def f}
f(z)=z^{1/2},
\end{equation}
where we put the branch cut along $(-\infty,0]$. Furthermore we define for $z \in D(0,\rho n^{-1/3})\setminus (-\rho n^{-1/3},0]$
\begin{align}
r_1(z)&= \frac32 e^{-\pi i/4}H(-z)-\frac32e^{\pi i/4}(F(-z)-F(0))z^{-1/2}, \label{eq: def r1 crit} \\
r_2(z)&= \frac32 e^{-\pi i/4}H(-z)+\frac32e^{\pi i/4}(F(-z)-F(0))z^{-1/2}, \label{eq: def r2 crit} \\
s(z) &= \frac12 e^{-3\pi i/4}F(0), \label{eq: def s crit} \\
t(z) &= iG(-z), \label{eq: def t crit}
\end{align}
where $F$, $G$, and $H$ are the analytic functions introduced in item (c)
of Lemma \ref{lemma: modified lambda functions}. Clearly $t(z)$ and
$s(z)$ are analytic in the disk $D(0,\rho n^{-1/3})$ for $n$ large
enough. $r_1(z)$ and $r_2(z)$ are only analytic in $D(0,\rho
n^{-1/3})\setminus (-\rho n^{-1/3},0]$, but they are each others
analytic continuation across the interval $(-\rho n^{-1/3},0)$.

In the construction of the parametrix, apart from the analytic structure of the parameters $r_1(z),$ $r_2(z)$, $s(z)$, and $t(z)$, we will also need their asymptotic behavior as $n \to \infty$.

\begin{lemma}\label{lemma: rst}
The functions $s(z)$ and $t(z)$ are analytic for $z\in \D$ and $n$ sufficiently large. The functions $r_1(z)$ and $r_{2}(z)$ are analytic in $\D \setminus (-\rho n^{-1/3},0]$ and satisfy
\begin{equation} \label{eq: jump r12}
r_{1,\pm}(x)=r_{2,\mp}(x), \qquad -\rho n^{-1/3} <x<0.
\end{equation}
Moreover, there exists a constant $\hat G \in \C$, independent of $z$ and $n$, such that\begin{align}\label{eq:limitr0s0t0}
\begin{cases}
\lim_{n \to \infty} r_{j}(n^{-1/3}z)=2,\\
\lim_{n \to \infty} n^{2/3}s( n^{-1/3}z)=\frac12(a^2-5b),\\
\lim_{n \to \infty} n^{1/3}t( n^{-1/3}z)=2a+\hat G z,
\end{cases}\end{align}
for $|z|< \rho$ and $j=1,2$. \end{lemma}
\begin{proof}
The analytic structure of these functions for sufficiently large $n$ was already discussed. The limiting behavior for $r_{j}$ is obvious. To obtain the limiting behavior for $t$ we write
\[ \begin{cases}
G(z)=G(0)+G_1 z +\mathcal O(z^2), \\
G^*(z)=G^*(0)+G_1^* z +\mathcal O(z^2), \\
\end{cases}\qquad \text{as }z \to 0. \]
Hence
\[
n^{1/3}t(n^{-1/3}z)=in^{1/3}G(0)-iG_1z+\mathcal O(n^{-1/3}),
\]
as $n \to \infty$, which follows from \eqref{eq: def t crit}. Using this, \eqref{eq: G}, and the fact that $G_1 \to G_1^*$ as $n \to \infty$, we obtain
\[
\lim_{n \to \infty}n^{1/3} t( n^{-1/3}z)=-2\lim_{n \to \infty}n^{1/3} \gamma^{-1/2}\left( \tfrac32 \gamma^2-\tfrac{1}{4 \gamma}-\tfrac54 \tau^{4/3}\gamma \right)-i G_1^* z.
\]
By inserting the limiting behavior for $\gamma$ and $\tau$ as given in \eqref{eq: scaling gamma} and  \eqref{eq: scaling alpha tau}, we obtain the statement.

The proof for $s(z)$ is easier. Using \eqref{eq: def s crit} we obtain
\[
n^{2/3}s( n^{-1/3}z)=\tfrac{1}{2}n^{2/3}e^{\frac{-3\pi i}{4}}F(0).
\]
Using \eqref{eq: F} we get
\[
\lim_{n \to \infty} n^{2/3}s( n^{-1/3}z)=2\lim_{n \to \infty}n^{2/3}
\gamma^{1/4} \left( -2 \gamma^2+ \tfrac{1}{\gamma} +
\tau^{4/3}\gamma \right),
\]
which in combination with \eqref{eq: scaling gamma} and  \eqref{eq: scaling alpha tau}, finishes the proof.
\end{proof}

From Lemma \ref{lemma: rst} it follows that \eqref{local:P:multicrit} is well-defined (postponing the definition of $E_n$ for a moment). Indeed, it follows from standard arguments that if the solution to RH problem  \ref{rhp:modelM} exists, it depends analytically on the parameters $r_j$, $s$, and $t$. Combining this observation with Lemma \ref{lemma: rst} we see that for the choice of $r_j,s$ and $t$ we made, the solution to RH problem \ref{rhp:modelM} exists and hence \eqref{local:P:multicrit} is well-defined for sufficiently large $n$.

\begin{lemma} \label{lemma: matching crit 0}
Given definitions \eqref{eq: def r1 crit}--\eqref{eq: def t crit}, the following formulas hold modulo $2\pi i$.
\begin{equation}
  \begin{cases}
    n(\wtil\lambda_1(z)+zK(-z))=\psi_1(n^{2/3}z^{1/2};r_1(z),n^{2/3}s(z))+nt(z)z^{1/2},  \\
    n(\wtil\lambda_2(z)+zK(-z))=\psi_2(n^{2/3}z^{1/2};r_2(z),n^{2/3}s(z))-nt(z)z^{1/2}, \\
    n(\wtil\lambda_3(z)+zK(-z))=-\psi_2(n^{2/3}z^{1/2};r_2(z),n^{2/3}s(z))-nt(z)z^{1/2}, \\
    n(\wtil\lambda_4(z)+zK(-z))=-\psi_1(n^{2/3}z^{1/2};r_1(z),n^{2/3}s(z))+nt(z)z^{1/2}.
  \end{cases}
\end{equation}
\end{lemma}
\begin{proof}
This is a straightforward verification.
\end{proof}

It remains to define the prefactor
\begin{multline}\label{def:Ez}
E_n(z) = e^{nzK(-z)}\Pinf(z) \wtil A^{-1}\\
\times \diag\left(n^{1/6}f(z)^{1/4},n^{1/6}(-f(z))^{1/4},n^{-1/6}(-f(z))^{-1/4},n^{-1/6}f(z)^{-1/4}\right)
\end{multline}
with $\wtil A$ in \eqref{mixing:matrix2}. This prefactor is not analytic in the full disk $\D$, but has a jump as indicated in the following lemma.

\begin{lemma}\label{lemma:jump:E}
$E_n(z)$ is analytic in $D(0,\rho n^{-1/3})\setminus (-\rho
n^{-1/3},0]$ with a jump
\begin{equation}\label{jump:E}
E_{n,+}(x) = E_{n,-}(x)\diag\left(\begin{pmatrix}0&1\\ -1& 0\end{pmatrix},\begin{pmatrix}0&1\\
-1& 0\end{pmatrix}\right),\qquad -\rho n^{-1/3} <x< 0.
\end{equation}
\end{lemma}

\begin{proof}
It follows from the definition of $E_n(z)$ and the fact that
$P^{(\infty)}(z)$ is analytic for $z\in D(0,\rho n^{-1/3})\setminus
(-\rho n^{-1/3},\rho n^{-1/3})$ that $E_n(z)$ is also analytic for
$z\in D(0,\rho n^{-1/3})\setminus\er$. For $x\in(0,\rho n^{-1/3})$,
one checks that the jumps of $P^{(\infty)}$ and of the rightmost
factor in \eqref{def:Ez} cancel each other out so that $E(x)$ is
analytic for $x\in (0,\rho n^{-1/3})$. A similar calculation yields
the jump \eqref{jump:E} of $E_n(x)$ for $x\in (-\rho n^{-1/3},0)$.
\end{proof}

Now let us check that $\Pnot(z)$ in \eqref{local:P:multicrit} has the correct
jumps. This is straightforward for the jumps on the 5 rays which are not
$\er^-$. To check the jump for $x\in\er^-$, we calculate (see also \cite{Del3})
\begin{align*}
\Pnot_{+}(x) &= E_{n,-}(x) \diag\left(\begin{pmatrix}0&1\\ -1&
0\end{pmatrix},\begin{pmatrix}0&1\\-1& 0\end{pmatrix}\right)
\\& \qquad \times N(n^{2/3}f_+(x);r_{1,+}(x),r_{2,+}(x),n^{2/3}s(x),n^{1/3}t(x))\Lam_+(x)\\
&=
E_{n,-}(x)N(-n^{2/3}f_+(x);r_{2,+}(x),r_{1,+}(x),n^{2/3}s(x),n^{1/3}t(x))
\\
& \qquad \times \Lam_-(x)
\diag\left(\begin{pmatrix}0&1\\ -1& 0\end{pmatrix},\begin{pmatrix}0&1\\
-1& 0\end{pmatrix}\right)\\
&= \Pnot_{-}(x)\diag\left(\begin{pmatrix}0&1\\ -1& 0\end{pmatrix},\begin{pmatrix}0&1\\
-1& 0\end{pmatrix}\right),
\end{align*}
where in the first step we used \eqref{local:P:multicrit} and \eqref{jump:E}, and in the second step we used
\eqref{symmetry:N} and commuted with $\Lam(x)$. The final equality follows from \eqref{eq: jump r12} and \eqref{local:P:multicrit}. This yields the required jumps on $\er^{-}$.

The following lemma states how the local parametrix matches with the
outer parametrix on the boundary of the shrinking disk $\D$,
settling item $(3)$ in the local RH problem \ref{rhpforP0}.
\begin{lemma} \label{lemma: critical matching} For $z$ on the
shrinking circle the following matching condition uniformly holds as
$n\to \infty$
\begin{equation} \label{eq: matching P}
\Pnot(z) \left( \Pinf(z) \right)^{-1}=I+Z(z)+\mathcal O(n^{-1/12}), \qquad |z|=\rho n^{-1/3},
\end{equation} where
\begin{equation} \label{eq: def Z}
Z(z)=\frac{ \what P_\pm \wtil
N_{1,\pm}(r_{1}(z),r_{2}(z),n^{2/3}s(z),n^{1/3}t(z)) \what
Q_\pm}{n^{1/3}z}, \quad \pm ~\Im z >0,
\end{equation}
and where $\wtil N_{1,\pm}$, $\what P_\pm$ and $\what Q_\pm$ are
defined in \eqref{eq: def tilde N1} and \eqref{eq: asym P0 at zero
crit}, respectively.
\end{lemma}
\begin{proof}
From \eqref{local:P:multicrit}, it follows that
\[
\Pnot(z) \left( \Pinf(z) \right)^{-1}=
E_n(z)N(n^{2/3}f(z);r_{1}(z),r_{2}(z),n^{2/3}s(z),n^{1/3}t(z))\Lambda(z)
\left(\Pinf(z) \right)^{-1}.
\]
If $|z|=\rho n^{-1/3}$ then $\left|n^{2/3}z^{1/2}\right|\to \infty$
as $n \to \infty$. Hence, we can use \eqref{N:asymptotics2},
\eqref{def:Ez} and Lemma \ref{lemma: matching crit 0} to obtain, for
$\pm \Im z >0$,
\begin{align*}
& \Pnot(z) \left( \Pinf(z) \right)^{-1}
\nonumber\\
&=\Pinf (z)\left( I+\frac{\wtil
N_{1,\pm}(r_{1}(z),r_{2}(z),n^{2/3}s(z),n^{1/3}t(z))}{n^{1/3}z^{1/4}}+\mathcal
O\left(\frac{1}{n^{2/3}z^{1/2}}\right)\right) \left( \Pinf
(z)\right)^{-1}.
\end{align*}
This, together with \eqref{eq: asym P0 at zero crit}, leads us to
\eqref{eq: matching P} and \eqref{eq: def Z}.
\end{proof}
Note that $Z(z)$, as defined in \eqref{eq: def Z}, remains bounded as $n \to \infty$ with $|z|=\rho n^{-1/3}$.

\subsection{Transformation $S \mapsto \wtil R$}
Using the global parametrix $S^{(\infty)}(z)$ and the local
parametrices $S^{(c)}$ and $S^{(0)}$ around $c$ and $0$ we define
the transformation $S \mapsto \wtil R$ as
\begin{equation} \label{eq: def V}
\wtil R=\begin{cases}
S (S^{(0)})^{-1}, & \text{in the disk $\D$ around 0,} \\
S (S^{(c)})^{-1}, & \text{in a fixed disk around $c$,} \\
S (S^{(\infty)})^{-1}, & \text{elsewhere.}
\end{cases}
\end{equation}

Then $\wtil R$ is defined and analytic outside $\Sigma_S$ and the
two disks around $0$ and $c$, with an analytic continuation across
those parts of $\Sigma_S$ on which the jumps of the parametrices
coincide with those of $S$. What remains are the jumps on a contour
$\Sigma_{\wtil R}$ that consists of the two circles, the part of the
interval $[c^*,\infty)$ outside the disk and the lips of local and
global lenses outside the disks.

By setting the orientations of circles to be clockwise, we have that
$\wtil R$ satisfies the following RH problem.

\begin{rhp} The matrix-valued function $\wtil R$ satisfies the following RH problem.
\begin{itemize}
\item[(1)] $\wtil R$ is analytic in $\C \setminus \Sigma_{\wtil R}$.
\item[(2)] $\wtil R$ satisfies the jump relation $\wtil R_+=\wtil R_- J_{\wtil R}$ on
$\Sigma_{\wtil R}$ with jump matrices
\[
J_{\wtil R}= \begin{cases}
S^{(0)}(S^{(\infty)})^{-1} & \text{on the boundary of $\D$,} \\
S^{(c)}(S^{(\infty)})^{-1} & \text{on the boundary of the disk around $c$,} \\
S^{(\infty)}J_S(S^{(\infty)})^{-1} & \text{elsewhere on $\Sigma_{\wtil R}$.}
\end{cases}
\]
\item[(3)] As $z \to \infty$, we have
\[
\wtil R(z)=I+\mathcal O(1/z).
\]
\end{itemize}
\end{rhp}

The jump matrix $J_{\wtil R}$ is not close to the identity matrix on
the shrinking circle around 0, since by Lemma \ref{lemma: critical
matching} we have
\[
J_{\wtil
R}(z)=S^{(0)}(z)(S^{(\infty)}(z))^{-1}=P^{(0)}(-z)(P^{(\infty)}(-z))^{-1}=I+Z(-z)+\mathcal
O(n^{-1/12}), \qquad \text{as $n \to \infty$}
\]
uniformly for $|z|=\rho n^{-1/3}$, with $Z(z)=\mathcal O(1)$.

The other jump matrices, however, are close to the identity matrix as $n$ gets large.

\subsection{Final transformation $\wtil R \mapsto R$: nilpotent structure.}
The presence of the bounded term  $Z(z)$ in Lemma \ref{lemma:
critical matching} requires an extra transformation.

\begin{lemma}\label{lemma:nilpotent}
The function $Z(z)$, defined in \eqref{eq: def Z}, has the following properties.
\begin{itemize}
\item[\rm (a)] $Z(z)$ is meromorphic in a neighborhood of zero with a simple pole in zero. Hence we can write
\[
Z(z)=\frac{Z_0}{z}+\left(Z(z)-\frac{Z_0}{z} \right),
\]
where $Z_0=\Res (Z, 0)$ is independent of $z$, and $Z(z)-Z_0/z$ is
analytic in $z$.
\item[\rm (b)] $Z(z)$ is nilpotent of degree two, moreover
\[
Z(z_1)Z(z_2)= 0, \qquad \text{for any $z_1,z_2$ in a neighborhood of the origin.}
\]
\end{itemize}
\end{lemma}
\begin{proof}
To prove (a) note that it is clear from \eqref{eq: def Z} that $Z(z)$ is
analytic in a neighborhood of zero with cut along the real line. It is then
sufficient to show that there is no jump on the real line. When $x\in(0,\rho)$,
there are no jumps for the functions $r_1$, $r_2$, $s$, and $t$. Hence, it
follows from \eqref{eq: def tilde N1} and \eqref{mixing:matrix2} that
\begin{align}
\wtil N_{1,+} =\frac{1}{4}
\begin{pmatrix}
1 & 0 & 0 & i \\
0 & -i & -1 & 0 \\
0 & 1 & -i & 0 \\
-i & 0 & 0 & 1
\end{pmatrix}
\wtil N_{1,-}\begin{pmatrix}
1 & 0 & 0 & -i \\
0 & -i & 1 & 0 \\
0 & -1 & -i & 0 \\
i & 0 & 0 & 1
\end{pmatrix}.
\end{align}
On account of \eqref{eq:positive x} and \eqref{eq:positive x Q}, we
obtain from \eqref{eq: def Z} that
\begin{align*}
Z_{+}(x) &= \frac{1}{4 n^{1/3}x}\what P_- \begin{pmatrix}
1 & 0 & 0 & i \\
0 & 1 & -i & 0 \\
0 & i & 1 & 0 \\
-i & 0 & 0 & 1
\end{pmatrix} \wtil N_{1,-}
\begin{pmatrix}
1 & 0 & 0 & -i \\
0 & 1 & i & 0 \\
0 & -i & 1 & 0 \\
i & 0 & 0 & 1
\end{pmatrix} \what Q_-= Z_{-}(x),
\end{align*}
where the second equality follows from \eqref{eq:relation P} and
\eqref{eq:relation Q}. On the negative real line we also need
\eqref{eq: symmetry tilde N1}. Indeed, for $-\rho<x<0$ and
$\zeta_{\pm}:= n^{2/3}x_{\pm}^{1/2}$, we have
\begin{align*}
Z_{+}(x) &= \frac{\what P_+}{x_+^{3/8}}
\frac{\wtil N_{1,+}(r_{1,+}(x),r_{2,+}(x),n^{2/3}s(x),n^{1/3}t(x)}{n^{1/3}x_+^{1/4}}\frac{\what Q_+ }{x_+^{3/8}} \\
&= \frac{\what P_-}{x_-^{3/8}} \frac{\wtil
N_{1,-}(r_{2,+}(x),r_{1,+}(x),n^{2/3}s(x),n^{1/3}t(x)}{(-\zeta_+)^{1/2}}\frac{\what
Q_- }{x_-^{3/8}},
\end{align*}
where we have used \eqref{eq: symmetry tilde N1} and the relations
\eqref{eq:negative x} and \eqref{eq:negative x Q}.
Combining this with \eqref{eq: jump r12} and the fact that $\zeta_+=-\zeta_-$
we see that $Z(z)$ is continuous across $(-\rho, 0)$ and, thus, also analytic
in a punctured neighborhood of the origin. Recalling \eqref{eq: def Z} we then
obtain (a).

(b) follows from \eqref{eq: def Z} and the observation that $\what Q_\pm \what P_\pm=0$, see \eqref{PQ:orthogonal}.
\end{proof}

As a corollary we also get
\begin{equation} \label{eq: nilpotentency 2}
Z(z)Z_0=0, \qquad Z_0Z(z)=0, \qquad Z_0^2=0.
\end{equation}

Then we define the transformation $\wtil R \mapsto R$ as
\begin{equation} \label{eq: def R crit}
R(z)=\begin{cases}
\wtil R(z)\left(I+Z(-z)+\frac{Z_0}{z} \right), & \text{for $z \in \D \setminus \Sigma_{\wtil R}$,} \\
\wtil R (z)\left(I+\frac{Z_0}{z} \right), & \text{for $z \in \C
\setminus (\D \cup \Sigma_{\wtil R})$.}
\end{cases}
\end{equation}

Then $R$ is defined and analytic in $\C \setminus \Sigma_R$ where
$\Sigma_R=\Sigma_{\wtil R}$ and $R$ satisfies a RH problem of the
following form.

\begin{rhp} \textrm{}
\begin{itemize}
\item[\rm (1)] $R$ is analytic in $\C \setminus \Sigma_R$.
\item[\rm (2)] $R$ satisfies the jump conditions $R_+=R_-J_R$ on $\Sigma_R$, with $J_R$ described below.
\item[\rm (3)] $R(z)=I+\mathcal O(1/z)$ as $z \to \infty$.
\end{itemize}
\end{rhp}

The jump matrix $J_R$ for $|z|=\rho n^{-1/3}$ is by \eqref{eq: def R crit} and \eqref{eq: nilpotentency 2}
\begin{align*}
J_R(z) &= \left(I+Z(-z)+\frac{Z_0}{z} \right)^{-1}J_{\wtil R}(z)\left( I+\frac{Z_0}{z}\right)\\
&=\left(I-Z(-z)-\frac{Z_0}{z} \right)\left(I+Z(-z)+\mathcal
O(n^{-1/12})\right)\left( I+\frac{Z_0}{z}\right)=I+\mathcal
O(n^{-1/12}),
\end{align*}
where we also use the fact that $Z(z)$ and $Z_0/z$ are bounded for $|z|=\rho n^{-1/3}$.

The transformation \eqref{eq: def R crit} does not change the jump matrices on the other parts of $\Sigma_R$ in an essential way. Hence $J_R$ tends to the identity matrix on these parts as well, with a rate of convergence that is the same as that for $J_V$.

We have now achieved the goal of the steepest descent analysis. $R(z)$ tends to the identity matrix as $z \to \infty$ and the jump matrices for $R$ tend to the identity matrix as $n\to \infty$, both uniformly on $\Sigma_R$ and in $L^2(\Sigma_R)$. By standard arguments, see \cite{Dei} and in particular \cite{BK3} for the case of a moving contour, this leads to the conclusion of our steepest descent analysis
\begin{equation}\label{eq: conclusion steepest descent crit}
R(z)=I+\mathcal O\left( \frac{1}{n^{1/12}(1+|z|)}\right),
\end{equation}
as $n \to \infty$, uniformly for $z \in \C \setminus \Sigma_R$.

\subsection{Proof of Theorem \ref{th: triple scaling limit}}
\label{sec: proof of triple scaling limit}

%
The idea of the proof is similar to that of the proof of Theorem \ref{th:
noncrit limit}, i.e.~we write the expression for the correlation kernel in
terms of $R$ instead of $Y$ by unfolding all transformations performed in the
steepest descent analysis. Since the first few transformations are the same as
in the regular cases, we see from the arguments in Section \ref{sec: proof
noncrit theorem} that if $0<x,y<c$,
\begin{multline} \label{eq: K in terms of S}
e^{n(y-x)}K_n(x,y) =\frac{y^{\nu/2}x^{-\nu/2}}{2\pi i (x-y)}\begin{pmatrix}
-e^{n\lambda_{1,+}(y)} & e^{n\lambda_{2,+}(y)} & 0 & 0 \end{pmatrix}
\\ \times S_+^{-1}(y)S_+(x) \begin{pmatrix} e^{-n\lambda_{1,+}(x)}&
e^{-n\lambda_{2,+}(x)}&0& 0 \end{pmatrix}^T.
\end{multline}
Moreover, for $x,y\in \D$, we have by \eqref{eq: def V} and \eqref{def:P}
\begin{multline*}
S_+(z)=\wtil R_+(z) S_+^{(0)}(z)=\wtil R_+(z)P_-^{(0)}(-z)
 J\diag (e^{\nu \pi i/2}, e^{-\nu \pi i/2}, e^{\nu \pi i/2}, e^{-\nu \pi
i/2}), \qquad z=x,y,
\end{multline*}
where $J$ is given in \eqref{def:J}.
It then follows that
\begin{multline*}
e^{n(y-x)}K_n(x,y) =\frac{y^{\nu/2}x^{-\nu/2}}{2\pi i (x-y)}\begin{pmatrix} 0 & 0 & e^{n\lambda_{2,+}(y)+\nu \pi i/2} & -e^{n\lambda_{1,+}(y)-\nu \pi i/2} \end{pmatrix}\\
\times \left(P_-^{(0)}(-y)\right)^{-1}\wtil R^{-1}(y)\wtil R(x)P_-^{(0)}(-x)
\begin{pmatrix} 0& 0& e^{-n\lambda_{2,+}(x)-\nu \pi i/2}&
e^{-n\lambda_{1,+}(x)+\nu \pi i/2}
\end{pmatrix}^T.
\end{multline*}
By \eqref{local:P:multicrit}, \eqref{eq: def f}, and \eqref{eq:
lambda tilde crit} this becomes
\begin{multline} \label{eq: K in terms of V}
e^{n(y-x)}K_n(x,y) = \frac{y^{\nu/2}x^{-\nu/2}}{2\pi i (x-y)}\begin{pmatrix} 0 & 0 & e^{\nu \pi i/2} & -e^{-\nu \pi i/2} \end{pmatrix} \\
\times N_-\left(n^{2/3}(-y)^{1/2}; r_{1,-}(-y),r_{2,-}(-y), n^{2/3}s(-y), n^{1/3}t(-y)\right)^{-1}  E_{n,-}^{-1}(-y)\wtil R^{-1}(y)\wtil R(x)E_{n,-}(-x) \\
\times   N_{-}\left(n^{2/3}(-x)^{1/2};r_{1,-}(-x),
r_{2,-}(-x),n^{2/3}s(-x),n^{1/3}t(-x)\right)
\begin{pmatrix} 0& 0& e^{-\nu \pi i/2}& e^{\nu \pi i/2}
\end{pmatrix}^T.
\end{multline}
Now we scale $x$ and $y$ with $n$ such that
\begin{equation}\label{xy:scaled}
x=\frac{u}{n^{4/3}} \qquad \textrm{ and } \qquad y=\frac{v}{n^{4/3}},
\end{equation}
where $u,v >0$. Then for large $n$, $x$ and $y$ belong to the disk $\D$, so
that \eqref{eq: K in terms of V} holds. We want to take the limit as $n \to
\infty$. Note that under these conditions
\[
\lim_{n \to \infty} e^{n(y-x)}=1,
\]
and by \eqref{eq:limitr0s0t0}
\begin{align*}
r_j(z) &\to 2, && j=1,2, \\
n^{2/3}s(z) &\to \tfrac12 (a^2-5b),  \\
n^{1/3}t(z) & \to 2a,
\end{align*}
as $n \to \infty$ and $z=x,y$.
Furthermore, it follows from \eqref{eq: conclusion steepest descent crit} and Cauchy's formula that
\begin{equation}\label{eq:cauchyonR}
R^{-1}(y)(R(y)-R(x))= \mathcal O \left(
\frac{x-y}{n^{1/12}}\right)=\mathcal O \left(n^{-17/12}\right),
\end{equation}
as $n\to \infty$ where the constant is uniform for $u,v$ in compact subsets of $\R$. Then also
\begin{align}\nonumber
\wtil R^{-1}(y)\wtil R(x)&=\left(I+Z(-y)+\frac{Z_0}{y}\right)R^{-1}(y)R(x)\left(I+Z(-x)+\frac{Z_0}{x}\right)^{-1} \\
           \label{proof:triple:On}  &=\left(I+Z(-y)+\frac{Z_0}{y}\right)R^{-1}(y)R(x)\left(I-Z(-x)-\frac{Z_0}{x}\right)  = I+\mathcal O(n^{-1}),
\end{align}
as $n \to \infty$, where the constant is again uniform for $u,v$ in compact subsets of $\R$.
Here  we used \eqref{eq: nilpotentency 2} to invert the rightmost matrix. To prove the last equality in \eqref{proof:triple:On}, note that the matrix function
$Z(z)-\frac{Z_0}{z}$, which is analytic by Lemma~\ref{lemma:nilpotent}(a),
can be written as a power series in the variable $n^{1/3}z$ with coefficients having a limit for $n\to\infty$,
thanks to \eqref{eq: def Z} and \eqref{eq:limitr0s0t0}.
Applying this with $z=x,y$ given in \eqref{xy:scaled} we get the claimed $\mathcal O(n^{-1})$ estimate in \eqref{proof:triple:On}.
In fact, the same reasoning yields the following more precise version 
of \eqref{proof:triple:On},
\begin{equation}\label{En:bound:1a}
\wtil R^{-1}(y)\wtil R(x) = I +  \what P_+\mathcal O(n^{-1}) + \mathcal O(n^{-1})\what Q_+ + o(n^{-1}),\qquad n\to\infty,
\end{equation}
where the matrices $\what P_+$, $\what Q_+$ originate from \eqref{eq: def Z}. (We could also write $\what P_-$, $\what Q_-$ instead.)

Next we estimate the factor $E_n(z)$ given in \eqref{def:Ez}.
We claim that the transformed matrix \begin{equation}\label{En:tilde}
\wtil E_n(z) := E_n(z)\diag\left(\begin{pmatrix}1&i\\ i& 1\end{pmatrix},\begin{pmatrix}1&i\\
i& 1\end{pmatrix}\right)\diag(\zeta^{-1/4},\zeta^{1/4},\zeta^{-1/4},\zeta^{1/4}),\qquad \zeta=n^{4/3}z,
\end{equation}
is analytic near $z=0$. Indeed it has no jumps,
by virtue of Lemma~\ref{lemma:jump:E}, and moreover it behaves as $\mathcal O(z^{-3/4})$ as $z\to 0$ so there is no pole at $z=0$.

\begin{lemma} We have
\begin{equation}
\lim_{n \to \infty} \wtil E_{n}^{-1}(-y)\wtil R^{-1}(y)\wtil
R(x)\wtil E_{n}(-x)=I,
\end{equation}
uniformly for $u,v$ in compact subsets of $\er$.
\end{lemma}

\begin{proof} In the proof below all the $\mathcal O$ and $o$ terms will be uniform for $u,v$
in compact subsets of $\er$.
We start by writing
\begin{multline}\label{En:bound:3}  \wtil E_{n}^{-1}(-y)\wtil R^{-1}(y)\wtil
R(x)\wtil E_{n}(-x) = \wtil E_{n}^{-1}(-y)\wtil E_{n}(-x)\\ + \wtil E_{n}^{-1}(-y)(\wtil R^{-1}(y)\wtil
R(x)-I)\wtil E_{n}(-x).
\end{multline}
Let us estimate the first term in the right hand side.
We have $\wtil E_n(z)^{\pm 1}=\mathcal O\left(n^{1/2} \right)$, $z=x,y$,  as $n\to \infty$, which is a
special case of \eqref{En:bound:1} below. Then by the analyticity of $\wtil E_n(z)$ we obtain
\begin{equation}\label{En:bound:2}
\wtil E_{n}^{-1}(y)(\wtil E_n(y)-\wtil E_{n}(x))  = \mathcal O \left( (x-y)n\right)=\mathcal O \left( n^{-1/3} \right),
\end{equation}
as $n \to \infty$.
Consequently the first term in the right hand side of \eqref{En:bound:3} goes to the identity matrix for $n\to\infty$.

Next we estimate the second term in the right hand side of \eqref{En:bound:3}.
On account of \eqref{xy:scaled}, \eqref{def:Ez} and \eqref{eq: asym P0 at zero crit} we have
\begin{equation}\label{En:bound:1} \wtil E_n(z) = \what P_+ \mathcal O\left(n^{1/2}\right) + o\left(n^{1/2}\right),\qquad
\wtil E_n(z)^{-1} = \mathcal O\left(n^{1/2}\right)\what  Q_+  + o\left(n^{1/2}\right),\qquad z=x,y, \end{equation}
as $n\to \infty$.
From \eqref{En:bound:1} and \eqref{En:bound:1a} we see that
\begin{multline*}  \wtil E_{n}^{-1}(-y)(\wtil R^{-1}(y)\wtil
R(x)-I)\wtil E_{n}(-x) \\ =  \mathcal O\left(n^{1/2}\right)\left(\what  Q_+ \what P_+\mathcal O(n^{-1})\what P_+  + \what  Q_+ \mathcal O(n^{-1})\what Q_+\what P_+ \right) \mathcal O\left(n^{1/2}\right)+ o(1) =  o(1),
\end{multline*}
for $n\to\infty$, where the second equality follows from the
orthogonality relation $\what  Q_+ \what P_+=0$ in \eqref{PQ:orthogonal}. Hence
 the second
term in the right hand side of \eqref{En:bound:3} goes to zero for $n\to\infty$.
\end{proof}

To use the above lemma, we should first express the matrix $E_n$ in \eqref{eq: K in terms of V}
in terms of its transformed counterpart $\wtil E_n$ in \eqref{En:tilde}.
This substitution releases an extra factor which multiplies from the left the matrix $N$ in \eqref{eq: K in terms of V}.
By combining this with the above estimates we find
\begin{multline}\label{proof:crit:1}
\lim_{n \to \infty}\frac{1}{n^{4/3}}K_n\left(\frac{u}{n^{4/3}},\frac{v}{n^{4/3}}\right) =\frac{u^{-\nu/2}v^{\nu/2}}{2\pi i (u-v)}\begin{pmatrix} 0 & 0 & e^{\nu \pi i/2} & -e^{-\nu \pi i/2} \end{pmatrix}\\
\times \what N\left(-iv^{1/2};2,2,\tfrac12 (a^2-5b),2a\right)^{-1}\what N\left(-iu^{1/2};2,2,\tfrac12 (a^2-5b),2a\right)
\begin{pmatrix} 0& 0& e^{-\nu \pi i/2}& e^{\nu \pi i/2}
\end{pmatrix}^T,
\end{multline}
with
\begin{equation*}
\what N\left(-i\zeta^{1/2}\right) := \diag(\zeta^{1/4},\zeta^{-1/4},\zeta^{1/4},\zeta^{-1/4})
\diag\left(\begin{pmatrix}1&-i\\ -i& 1\end{pmatrix},\begin{pmatrix}1&-i\\
-i& 1\end{pmatrix}\right)N\left(-i\zeta^{1/2}\right),
\end{equation*}
for $\zeta=u,v$. Equivalently, by \eqref{eq: def N}, \eqref{symmetry:N} and \eqref{Mhat},
$$ \what N\left(-i\zeta^{1/2}\right) = \diag(1,-i,i,1)\what M(\zeta)\diag\left(1,i,\begin{pmatrix}0& -1\\ -i&0\end{pmatrix}\right).
$$
Inserting this in \eqref{proof:crit:1} we get
\begin{multline*}
\lim_{n \to
\infty}\frac{1}{n^{4/3}}K_n\left(\frac{u}{n^{4/3}},\frac{v}{n^{4/3}}\right)
=\frac{u^{-\nu/2}v^{\nu/2}}{2\pi i (u-v)}\begin{pmatrix} 0 & 0 &
e^{-\nu \pi i/2} & ie^{\nu \pi i/2} \end{pmatrix}\\
 \times \what M\left(v;\tfrac12 (a^2-5b),2a,\nu \right)^{-1} \what M\left(u;\tfrac12 (a^2-5b),2a,\nu \right)
 \begin{pmatrix} 0& 0& -e^{\nu \pi i/2} & -ie^{-\nu \pi i/2}
 \end{pmatrix}^T,
\end{multline*}
where we recall that $\what M(z)=\what M\left(z;s,t,\nu\right)$ is defined in \eqref{Mhat}
with $M(\zeta)=M(\zeta;s,t, \nu)$ denoting the unique solution to RH
problem \ref{rhp:modelM} with parameters given in \eqref{eq:
parameters}. The above kernel is clearly equal to the right hand side of \eqref{kernel:crit:thm}.
This completes the proof of Theorem \ref{th: triple scaling limit}.


\section*{Acknowledgment}
We thank Gernot Akemann and Arno Kuijlaars for their helpful
comments. SD and LZ are Postdoctoral Fellows of the Fund for
Scientific Research - Flanders (FWO), Belgium. DG is a Research
Assistant of the Fund for Scientific Research - Flanders (FWO),
Belgium.


\begin{thebibliography}{99}
\bibitem{DLMF}
    Digital Library of Mathematical Functions. 2012-03-23. National Institute of Standards and
    Technology from http://dlmf.nist.gov/
\bibitem{AkeDam}
    G. Akemann and P.H. Damgaard,
    Individual eigenvalue distributions of chiral random two-matrix theory and the determination of
    $F_{\pi}$,
    J. High Energy Phys. 0803 (2008) 073.
\bibitem{ADMN}
    G. Akemann, P.H. Damgaard, U. Magnea and S. Nishigaki,
    Universality of random matrices in the microscopic limit and the Dirac operator spectrum,
    Nucl. Phys. B 487 (1997), 721--738.
\bibitem{ADOS}
    G. Akemann, P.H. Damgaard, J. C. Osborn and K. Splittorff,
    A new chiral two-matrix theory for Dirac spectra with imaginary chemical potential,
    Nucl. Phys. B 766 (2007), 34--76.
\bibitem{AkeIps}
    G. Akemann and Ipsen,
    The $k$-th smallest Dirac operator eigenvalue and the pion decay constant,
    J. Phys. A: Math. Theor. 45 (2012), 115205.
\bibitem{AV}
    G. Akemann and Vernizzi,
    Characteristic polynomials of complex random matrix models,
    Nucl. Phys. B 660 (2003), 532--556.
\bibitem{BasAke}
    F. Basile and G. Akemann,
    Equivalence of QCD in the epsilon-regime and chiral Random Matrix Theory with or without chemical
    potential,
    J. High Energy Phys. 0712 (2007), 043.
\bibitem{BEH}
    M. Bertola, B. Eynard, and J. Harnad,
    Differential systems for biorthogonal polynomials appearing in 2-matrix models and the associated Riemann-Hilbert problem,
    Comm. Math. Phys. 243 (2003), 193--240.
\bibitem{BK3}
    P.M. Bleher and A.B.J. Kuijlaars,
    Large $n$ limit of Gaussian random matrices with external
    source, part III: double scaling limit,
    Comm. Math. Phys. 270 (2007), 481--517.
\bibitem{Claeys2}
    T. Claeys, A.B.J. Kuijlaars, and M. Vanlessen,
    Multi-critical unitary random matrix ensemble and the general
    Painlev\'e~II equation,
    Ann. Math. 167 (2008), 601--642.
\bibitem{DK1}
    E. Daems and A.B.J. Kuijlaars,
    A Christoffel-Darboux formula for multiple orthogonal polynomials,
    J. Approx. Theory 130 (2004), 188--200.
\bibitem{DHSS}
    P.H. Damgaard, U.M. Heller, K. Splittorff and B. Svetitsky,
    A new method for determining $F_{\pi}$ on the lattice,
    Phys. Rev. D 72 (2005), 091501.
\bibitem{DHSST2}
    P.H. Damgaard, U.M. Heller, K. Splittorff, B. Svetitsky and D. Toublan,
    Microscopic eigenvalue correlations in QCD with imaginary isospin chemical
    potential,
    Phys. Rev. D 73 (2006), 105016.
\bibitem{DegrandSch}
    T. DeGrand and S. Schaefer,
    Parameters of the lowest order chiral Lagrangian from fermion eigenvalues,
    Phys. Rev. D 76 (2007), 094509.
\bibitem{Dei}
    P. Deift, \emph{Orthogonal Polynomials and Random Matrices: a Riemann-Hilbert
    approach}, Courant Lecture Notes in Mathematics Vol. 3, Amer. Math. Soc.,
    Providence R.I. 1999.
\bibitem{Del3}
    S. Delvaux,
    Non-intersecting squared Bessel paths at a hard-edge tacnode,
    to appear in Comm. Math. Phys., arXiv:1204.4430.
\bibitem{DKRZ}
    S. Delvaux, A.B.J. Kuijlaars, P. Rom\'{a}n and L. Zhang,
    Non-intersecting squared Bessel paths with one positive starting and ending point,
    J. Anal. Math. 118 (2012), 105--159.
\bibitem{DKZ}
    S. Delvaux, A.B.J. Kuijlaars and L. Zhang,
    Critical behavior of non-intersecting Brownian motions at a tacnode,
    Comm. Pure Appl. Math. 64 (2011), 1305--1383.
\bibitem{Deschout1}
    K. Deschout and A.B.J. Kuijlaars,
    Double scaling limit for modified Jacobi-Angelesco polynomials,
    in ``Notions of Positivity and the Geometry of Polynomials'' (P. Br\" and\'en, M. Passare, and M. Putinar, eds.),
    Trends in Mathematics, Springer, Basel, 2011, pp. 115--161.
\bibitem{DG}
    M. Duits and D. Geudens,
    A critical phenomenon in the two-matrix model in the quartic/quadratic case,
    to appear in Duke Math. J., arXiv:1111.2162.
\bibitem{DGK}
    M. Duits, D. Geudens and A.B.J. Kuijlaars,
    A vector equilibrium problem for the two-matrix model in the quartic/quadratic case,
    Nonlinearity 24 (2011), 951--993..
\bibitem{Duits2}
    M. Duits and A.B.J. Kuijlaars,
    Universality in the two matrix model: a Riemann-Hilbert steepest descent analysis,
    Comm. Pure Appl. Math. 62 (2009), 1076--1153.
\bibitem{DKM}
    M. Duits, A.B.J. Kuijlaars and M.Y. Mo,
    The Hermitian two matrix model with an even quartic potential,
    Memoirs Amer. Math. Soc. 217 No. 1022 (2012), vi+105 pp.
\bibitem{DKM2}
    M.~Duits, A.B.J.~Kuijlaars and M.Y.~Mo,
    Asymptotic analysis of the two matrix model with a quartic potential,
    arXiv:1210.0097.
\bibitem{EMcL}
    N.M. Ercolani and K.T.-R. McLaughlin,
    Asymptotics and integrable structures for biorthogonal polynomials associated to a random two-matrix model,
    Phys. D 152/153 (2001), 232--268.
\bibitem{EM}
    B. Eynard and M. Mehta,
    Matrices coupled in a chain. I. Eigenvalue correlations,
    J. Phys. A 31 (1998), 4449--4456.
\bibitem{FK}
    H.M.~Farkas and I.~ Kra,
    Riemann surfaces,
    Graduate Texts in Mathematics 71, Springer-Verlag, New York-Berlin, 1980.
\bibitem{GZ}
    D. Geudens and L. Zhang,
    Transitions between critical kernels: from the tacnode kernel and critical kernel in the two-matrix model to
    the Pearcey kernel, arXiv:1208.0762.
\bibitem{HK} A. Hardy and A.B.J. Kuijlaars,
    Weakly admissible vector equilibrium problems,
    J. Approx. Theory 164 (2012), 854--868.
\bibitem{KF}
    E. Kanzieper and V. Freilikher,
    Random matrix models with log-singular level confinement: method of fictitious fermions,
    Philos. Magazine B 77 (1998), 1161–-1172.
\bibitem{Kap}
    A.A. Kapaev, The Riemann-Hilbert problem for the bi-orthogonal polynomials,
    J. Phys. A 36 (2003), 4629–-4640.
\bibitem{KMFW2}
    A.B.J. Kuijlaars, A. Martinez-Finkelshtein and F. Wielonsky,
    Non-intersecting squared Bessel paths: Critical time and double scaling limit,
    Comm. Math. Phys 308 (2011), 227--279.
\bibitem{KMcL}
    A.B.J. Kuijlaars and K.T.-R. McLaughlin,
    A Riemann-Hilbert problem for biorthogonal polynomials,
    J. Comput. Appl. Math. 178 (2005), 313--320.
\bibitem{KMVV2004}
    A.B.J.~Kuijlaars, K.T-R.~McLaughlin, W.~Van~Assche, and
    M.~Vanlessen, The {R}iemann-{H}ilbert approach to strong asymptotics
    for orthogonal polynomials on {$[-1,1]$},
    Adv. Math. 188 (2004), 337--398.
\bibitem{KVL}
    A.B.J. Kuijlaars and M. Vanlessen,
    Universality for eigenvalue correlations at the origin of the spectrum,
    Comm. Math. Phys. 243 (2003), 163--191.
\bibitem{LBHW}
    C. Lehner, J. Bloch, S. Hashimoto and T. Wettig,
    Geometry dependence of RMT-based methods to extract the low-energy constants $\Sigma$ and
    $F$,
    J. High Energy Phys. 1105 (2011), 115.
\bibitem{MaPa}
    V.A.~Marchenko and L.A. Pastur,
    Distribution of eigenvalues for some sets of random matrices,
    Math.~USSR Sb.~1 457 (1967).
\bibitem{Mehta}
    M.L. Mehta,
    Random Matrices, 3rd edition,
    Elsevier/Academic Press, Amsterdam, 2004.
\bibitem{Mo}
    M.Y. Mo,
    Universality in the two matrix model with a monomial quartic and a general even polynomial potential,
    Comm. Math. Phys. 291 (2009), 863--894.
\bibitem{Osb}
    J.C. Osborn,
    Universal results from an alternate random matrix model for QCD with a baryon chemical
    potential,
    Phys. Rev. Lett. 93 (2004), 222001--222004.
\bibitem{SaffTotikBook}
    E.B.~Saff and V.~Totik, Logarithmic Potentials with External
    Fields,
    Springer-Verlag, Berlin, 1997.
\bibitem{Verb1}
    E.V. Shuryak and J.J.M. Verbaarschot,
    Random matrix theory and spectral sum rules for the Dirac operator in QCD,
    Nucl. Phys. A 560 (1993), 306--320.
\bibitem{VAGK}
    W. Van Assche, J.S. Geronimo, and A.B.J. Kuijlaars,
    Riemann-Hilbert problems for multiple orthogonal polynomials,
    Special Functions 2000: Current Perspectives and Future Directions
    (J. Bustoz et al., eds.), Kluwer, Dordrecht, 2001, pp. 23--59.
\bibitem{Verb3}
    J.J.M. Verbaarschot,
    The spectrum of the QCD Dirac operator and chiral random matrix theory: the threefold way,
    Phys. Rev. Lett. 72 (1994), 2531--2533.
\bibitem{Verb2}
    J.J.M. Verbaarschot and I. Zahed,
    Spectral density of the QCD Dirac operator near zero virtuality,
    Phys. Rev. Lett. 70 (1993), 3852--3855.
\bibitem{Zverovich}
    E.I.~Zverovich,
    Boundary value problems in the theory of analytic functions in H\"older
    classes on Riemann surfaces,
    Russ. Math. Surv. 26 (1971), 117--192.


\end{thebibliography}
\end{document}